\documentclass[aps,pra,amsmath,amssymb,nofootinbib,superscriptaddress,showpacs,floatfix,twocolumn,longbibliography]{revtex4-1} 

\usepackage[utf8]{inputenc}
\usepackage{latexsym}
\usepackage{graphicx}
\usepackage{subfig}
\usepackage{enumerate}
\usepackage{amsmath}
\usepackage{dsfont}
\usepackage{dcolumn}
\usepackage{bm,bbm}
\usepackage{amsthm,amssymb,euscript}
\usepackage[normalem]{ulem}
\usepackage{textcomp}
\usepackage{tabularx}
\usepackage{setspace}
\usepackage{ctable}
\usepackage{sidecap}
\usepackage{placeins}
\usepackage{threeparttable}
\usepackage{multirow}
\usepackage[hidelinks]{hyperref}
\usepackage{diagbox}
\usepackage{color}
\usepackage{gensymb}

\let \oldbm \bm
\renewcommand{\vec}[1]{\oldbm{#1}}

\def\bk{{\vec k}}

\def\bt{{\vec t}}

\def\bR{{\vec R}}
\def\bG{{\vec G}}

\def\bn{{\vec n}}
\def\bm{{\vec m}}

\def\br{{\vec r}}
\def\bt{{\vec t}}

\def\btau{{\boldsymbol \tau}}
\def\bsigma{{\boldsymbol \sigma}}

\def\sgn{\mathop{\mathrm{sgn}}}

\def\gcd{\mathop{\mathrm{gcd}}}

\def\Z{\mathbbm{Z}}
\def\T{\mathcal{T}}

\def\S{\mathcal{S}}
\def\I{\mathcal{I}}

\def\H{\mathcal{H}}
\def\K{\mathcal{K}}

\newcommand{\beq}{\begin{equation}}
\newcommand{\eeq}{\end{equation}}
\newcommand{\beqarray}{\begin{eqnarray}}
\newcommand{\eeqarray}{\end{eqnarray}}

\newcommand{\Mod}{\!\! \mod}

\newcommand{\Rom}[1]{\uppercase\expandafter{\romannumeral#1}}
\newcommand{\rme}{e}
\newcommand{\rmi}{i}
\newcommand{\rmd}{d}

\newcommand{\ex}[1]{\left\langle #1 \right\rangle}
\newcommand{\ket}[1]{|#1\rangle}
\newcommand{\bra}[1]{\langle #1 |}

\makeatletter
\renewcommand*\env@matrix[1][*\c@MaxMatrixCols c]{%
	\hskip -\arraycolsep
	\let\@ifnextchar\new@ifnextchar
	\array{#1}}
\makeatother

\newtheorem{theorem}{Theorem}

\allowdisplaybreaks
\begin{document}

\title{Shift insulators: rotation-protected two-dimensional topological crystalline insulators}

\author{Shang Liu}
\affiliation{Department of Physics, Harvard University, Cambridge, MA 02138, USA}

\author{Ashvin Vishwanath}
\email{avishwanath@g.harvard.edu}
\affiliation{Department of Physics, Harvard University, Cambridge, MA 02138, USA}

\author{Eslam Khalaf}
\email{eslam\_khalaf@fas.harvard.edu}
\affiliation{Department of Physics, Harvard University, Cambridge, MA 02138, USA}
\affiliation{Max Planck Institute for Solid State Research, Heisenbergstrasse\ 1, 70569 Stuttgart, Germany}

\begin{abstract}
We study a two-dimensional (2D) tight-binding model of a topological crystalline insulator (TCI) protected by rotation symmetry. The model is built by stacking two Chern insulators with opposite Chern numbers which transform under conjugate representations of the rotation group, e.g. $p_\pm$ orbitals. Despite its apparent similarity to the Kane-Mele model, it does not host stable gapless surface states. Nevertheless the model exhibits topological responses including the appearance of quantized fractional charge bound to rotational defects (disclinations) and the pumping of angular momentum in response to threading an elementary magnetic flux, which are described by a mutual Chern-Simons coupling between the electromagnetic gauge field and an effective gauge field corresponding to the rotation symmetry. In addition, we show that although the filled bands of the model do not admit a symmetric Wannier representation, this obstruction is removed upon the addition of appropriate atomic orbitals, which implies `fragile' topology. As a result, the response of the model can be derived by representing it as a superposition of atomic orbitals with positive and negative integer coefficients.  Following the analysis of the model, which serves as a prototypical example of 2D TCIs protected by rotation, we show that all TCIs protected by point group symmetries which do not have protected surface states are either atomic insulators or fragile phases. Remarkably, this implies that gapless surface states exist in free electron systems if and only if there is a stable Wannier obstruction. We then use dimensional reduction to map the problem of classifying 2D TCIs protected by rotation to a zero-dimensional (0D) problem which is then used to obtain the complete non-interacting classification of such TCIs as well as the reduction of this classification in the presence of interactions.
\end{abstract}

\maketitle

\section{Introduction}
Conventional topological insulators (TIs) of free fermions are protected by internal symmetries such as time-reversal and are characterized by several equivalent distinguishing features \cite{Molenkamp13, Hasan10, Moore09, Qi11}: (i) the bulk of a topological insulator cannot be adiabatically deformed to a trivial insulator without closing the energy gap, (ii) the surface of a topological insulator hosts anomalous gapless symmetry-protected states, and (iii) there is an obstruction to finding a basis of symmetric localized Wannier states (in dimensions larger than 1) \cite{Thonhauser06, Soluyanov11, Marzari12}. The first of these distinctions is a relative one which only characterizes whether two phases are topologically distinct without specifying which of them is topological whereas the second and third ones are absolute distinctions which do not require comparison to a reference state\footnote{The second distinction can also be thought of as a relative distinction characterizing the interface between two phases. In a typical setting, one considers the interface with a ``vacuum'' state which is trivial by definition to make it absolute}. 
Other signatures of conventional topology include stable excitations bound to topological defects and pumping of charge, spin or polarization in response to various probes \cite{laughlin1981quantized,thouless1983quantization,fu2006time,ran2008spin,qi2008topological}.

In contrast to the notion of a topological phase with {\em internal} symmetries, the definition of topological phases in the presence of crystalline symmetries poses additional subtleties. Recall that free-fermion topological phases with {\em internal} symmetries always possess gapless edge states which also serves to identify the trivial phase. Hence the topology is absolute, and the trivial insulator is clearly distinguished. Further, for the symmetry classes relevant to electronic insulators, where particle hole symmetry is absent, all of these topological phases present an obstruction to the construction of symmetric Wannier states \cite{Thonhauser06, Soluyanov11}.  

For topological phases protected by crystalline symmetries such as translation, rotation, or inversion, three cases need to be distinguished. First, there are phases with protected edge states, which most closely resemble the internal symmetry protected topological insulators. These include, for example, the four-fold rotation or reflection symmetric topological crystalline insulators (TCIs), in which stable gapless modes appear on symmetric surfaces \cite{Fu11,Hsieh12, Tanaka12,zhang2016topological}. They feature an obstruction to Wannier localization which is stable, in the sense that it is not resolved even when filled bands of atomic insulators are supplied. 

Next, there are `fragile' topological insulators \cite{Po17fragile}, which also possess an obstruction to the construction of symmetric Wannier states for the occupied bands. However, these are resolved by the addition on filled atomic bands. The resulting state can be viewed as an atomic insulator $A$. If we consider insulators in the atomic limit, (where electrons are strictly localized to sites)  to be trivial, then this fragile topology can be unwound by the addition of these trivial bands, $A'$. In other words, the fragile topological insulators can be represented as a difference between two sets of atomic insulators $A-A'$. 

Finally, we have atomic insulators. Although these may appear to be trivial, symmetry imposes distinctions between them, related to whether they can be adiabatically deformed into one another while preserving symmetry \cite{Po17Lattice,Bradlyn17, Po17}. These distinctions are captured within K-theory, but the topology is relative since no band structure is singled out as trivial. Here, we will follow the usual convention of referring to all these states as TCIs, although a safer definition maybe to restrict that term to (i) above.

The relative topological distinction between gapped Hamiltonians is captured by K-theory \cite{Kitaev09} which classifies them into equivalence classes under symmetric adiabatic deformations. The K-group was worked out for TCIs (in all symmetry classes) with order-two symmetries, which include mirror \cite{Morimoto13, Chiu14}, inversion \cite{Lu14}, and two-fold rotation, in the work of Shiozaki and Sato \cite{Shiozaki14}. Recently, it was extended in the case of broken time-reversal symmetry (class A) to include all 17 wallpaper groups in two dimensions \cite{Slager17} and 230 space groups in three dimensions \cite{Shiozaki18}. The K-group does not, however, provide any information about the absolute topological signatures such as the existence of surface states and Wannier representability. 

The understanding of surface states of TCIs received a significant boost recently with the realization that their existence is not restricted to surface planes which are invariant under the protecting spatial symmetry. Instead, surface states can be also observed by considering symmetry-compatible surfaces, which only preserve the symmetry as a whole. In this case, a TCI in $d$ dimensions may host surface states on a surface whose co-dimension is less than $d-1$. Such types of surface states have been known to appear for instance upon applying a magnetic field to a topological insulator \cite{Sitte12} or superconductor \cite{Volovik10}, or inside topological defects \cite{Teo10, Benalcazar14} but the role of spatial symmetries in stabilizing them was only recently understood. This led to the notion of ``higher-order TIs'' which are TCIs with gapless corner or hinger modes \cite{Benalcazar17, Benalcazar18, Song17, Schindler17, Langbehn17}. Unlike conventional TCIs, higher-order TIs can be protected by symmetries which do not leave any surface plane invariant such as inversion \cite{Fang17, Khalaf17, Khalaf18, Geier18, Trifunovic18}, roto-inversion \cite{Khalaf17,vanMiert18}, or screws \cite{Khalaf17}. In these cases, any given surface plane breaks the symmetry leading to a gapped dispersion, but the hinges between different planes represent domain walls which host gapless surface states.

In a parallel development, a comprehensive understanding of Wannier obstructions that can be identified from symmetry representations was achieved in several recent works \cite{Bradlyn17, Po17, Watanabe17}. In Refs.~\onlinecite{Bradlyn17, Po17}, the symmetry representations for all possible atomic insulators in the 230 groups were discussed. Ref.~\onlinecite{Watanabe17} extended these results further to include the 1651 magnetic space groups. This approach provided an explicit representation for the trivial phases as well as a diagnosis for the obstruction to finding symmetric Wannier states in the non-trivial ones. Importantly, there are two qualitatively different origins of the obstructions, which can be distinguished by the addition of atomic degrees of freedom. While the obstruction remains for the stable two and three dimensional topological phases such as topological insulators, it is resolved in other cases which have been dubbed `fragile' topological phases \cite{Po17fragile}. Fragile topological phases can be thought of as combinations of atomic insulators with integer coefficients, but where some of the coefficients are allowed to be negative.

This work is motivated by two main questions. The first one is understanding the nature and response of fragile phases in a setting in which they arise naturally. Some of the models for fragile phases known so far \cite{Po17fragile, Cano18, Bouhon18} has been built specifically to illustrate the existence of fragile Wannier obstructions which might give the impression that fragile phases represent a somewhat pathological case. Instead, we show here that fragile phases are ubiquitous in TCIs protected by symmetries which do not support any surface states, such as rotation symmetry in 2D which is the main focus of this work. In addition, we show that fragile phases can sometimes be distinguished from atomic insulators by investigating their response to standard probes such as topological defects or flux threading. 

The second question concerns the relationship between stable Wannier obstructions and surface sates. We know that anomalous surface states in electronic systems implies a stable Wannier obstruction but it is unclear whether it is a necessary condition i.e. whether it is possible to have a stable Wannier obstruction in a TCI which does not posses any anomalous surface states. We will show here that this is not possible by establishing that, within the layer construction of TCIs \cite{Hermele17, Huang17,Fulga16}, the absence of surface states implies that the phase can be built by repeating or ``layering'' 0D units. This is then used to show that these TCIs are either atomic insulators or fragile phases. Although we restrict ourselves to point group symmetries, we conjecture that such relation holds in general.

For most of this work, we will focus on a prototypical example for a 2D TCI protected by rotation symmetry. The model, which we will dub ``shift insulator'' (in reference to the ``shift'' defined in Ref.~\onlinecite{Wen92} which, for example, is sensitive to the orbital spin of the different Landau levels), is built by stacking two rotationally-symmetric Chern insulators with opposite Chern numbers corresponding to conjugate representations of the $n$-fold rotation group, e.g. $p_\pm$ orbitals. It can be viewed as an analog of the Kane-Mele model \cite{Kane05a, Kane05b} where the protecting symmetry (time-reversal) is replaced by spatial rotation. Despite having no surface states, we will show that this model exhibits several interesting features which are usually associated with topology: (i) it has several distinct phases which cannot be symmetrically deformed to each other without closing the gap, (ii) it exhibits a topological response in the form of quantized fractional charge bound to rotational defects (disclinations) and a quantized angular momentum pumping in response to the application of magnetic flux, both features being captures by an effective Chern-Simons coupling between the electromagnetic gauge field and an effective gauge field corresponding to the rotation symmetry, and (iii) there exists an obstruction to the construction of symmetric localized Wannier functions. The disclination charge response in fullerene Haldane models, a closely related context, was discussed in Ref.~\onlinecite{Coh2013}. Despite the apparent non-triviality of the model, we will show that its topology is fragile i.e. it admits a symmetric Wannier representation upon the addition of some localized atomic orbitals. It follows that the topological response of the model can be fully explained using a picture of localized atomic orbitals. The number and type of these atomic orbitals is different for different phases, which explains how these phases can exhibit different values for some quantized invariants. This serves to show that fragile phases and even atomic insulators can exhibit seemingly topological features.

Following the analysis of the model, we consider the general problem of TCIs protected by point group symmetries. We show that for these TCIs, the absence of surface states implies they can be built within the layer construction \cite{Hermele17, Huang17,Fulga16} by repeating (or layering) a 0D unit. This is, in turn, used to establish they are either atomic insulators or fragile phases,  thereby showing that the existence of stable Wannier obstructions is equivalent to the existence of surface states. Our analysis is then used to obtain a complete non-interacting classification of 2D TCIs protected by rotation as well as the interaction-induced reduction of such classification. As an example, we use these results to show that the classification for the shift insulator is reduced from the non-interacting $\Z$ to $\Z_{12}$ in the presence of interactions.

\section{Model}
\label{Model}
We begin this section by introducing the shift insulator model and its symmetries. Afterwards, we investigate the different phases of the model by analyzing the possible symmetry-allowed mass terms which can be added to it in the continuum limit. We then analyze the edge theory and show that the edge can be completely gapped out using a specific symmetry allowed perturbation.

\subsection{Hamiltonian and Symmetries}
Given a 2D lattice with $n$-fold rotational symmetry, we can assign definite angular momenta $l=0,\dots,n-1$ to any given orbital. The model for the shift insulator is obtained by considering two orbitals with angular momenta $l_0$ and $-l_0$ which form bands with Chern number $C$ and $-C$ respectively. For most of this paper, we will focus on the case $n=6$, $l_0 = 1$ and $C=1$ which can be implemented by stacking two (6-fold symmetric) Haldane models \cite{Haldane88} with $p_\pm$ orbitals corresponding to opposite Chern numbers $\pm 1$. 

Recall that the Haldane model is a two-band model defined on a honeycomb lattice with the tight-binding Hamiltonian given by
\begin{equation}
	H_\text{Haldane}=-t\sum_{\ex{i,j}}c^\dagger_ic_j+\lambda\sum_{\ex{\ex{i,j}}}\rmi\nu_{ij}c^\dagger_ic_j. 
\end{equation}
Here, both $t$ and $\lambda$ are real numbers. $\nu_{ij}=+1~(-1)$ if the hopping direction from $j$ to $i$ is right-handed (left-handed) around the plaquette center, i.e. along (against) the directions indicated in Fig.~\ref{HaldaneModel}. 
\begin{figure}[h]
	\centering
	\includegraphics[width=0.5\linewidth]{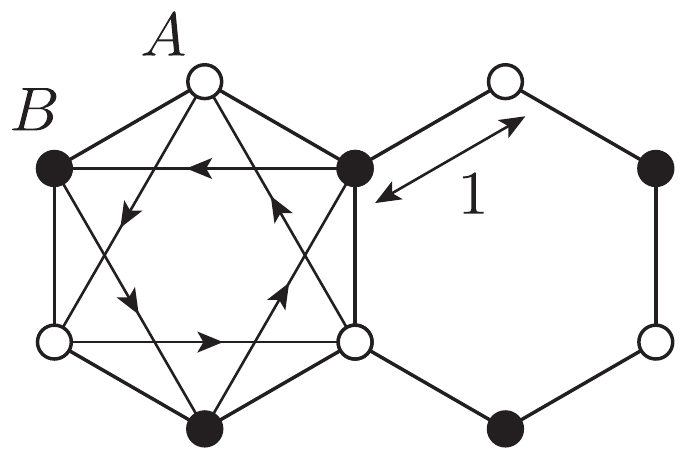}
	\caption{Sign convention for $\nu_{ij}$. $\nu_{ij}=+1~(-1)$ if the hopping direction from $j$ to $i$ is along (against) the directions indicated in the left plaquette. }
	\label{HaldaneModel}
\end{figure}

We now consider $p_\pm=p_x\pm\rmi p_y$ orbitals denoted by $\ket{\vec{r},\pm}$ at each site. Under the action of $2\pi/6$ rotation around some plaquette center on the honeycomb lattice, which we denote by $\hat C_6$, these states transform as
\begin{equation}
	\hat C_6\ket{\br,\pm}=\rme^{\mp\rmi\pi/3}\ket{O_6 \br,\pm},  
\end{equation}
where $O_6$ is the natural action of 6-fold rotation on 2D vectors given by 
\beq
O_6 (x,y) = \frac{1}{2} (x  - \sqrt{3} \, y, \sqrt{3}\, x + y).
\eeq
This implies the following transformation properties for the annihilation operators $c_{\br,\pm}$
\begin{equation}
	\hat C_6 c_{\br,\pm} \hat C_6^{-1}=\rme^{\pm\rmi\pi/3}c_{O_6 \br,\pm}. 
\end{equation}

The Hamiltonian for the shift insulator can then be written by assigning each of the two orbitals a Haldane model with opposite $\lambda$ couplings (therefore opposite Chern numbers). More specifically, 
\begin{equation}
	H=H_{\text{Haldane},p_+}(t,\lambda)+H_{\text{Haldane},p_-}(t,-\lambda). 
	\label{Shift}
\end{equation}

We now go to momentum space by considering fully periodic boundary conditions and taking the following Fourier transform convention: 
\begin{equation}
	c_{\bk}=\frac{1}{\sqrt{N}}\sum_{\vec{r}} \rme^{-\rmi \bk \cdot \br} c_{\br}~~~\Leftrightarrow~~~c_{\br}=\frac{1}{\sqrt{N}}\sum_{\vec{k}} \rme^{\rmi \bk \cdot \br} c_{\bk},
\end{equation}
where all discrete degrees of freedom, such as sublattices or orbital, are suppressed, and the length unit for $\vec{r}$ is shown in Fig.~\ref{HaldaneModel}. The shift insulator Hamiltonian can now be written as
\begin{equation}
	H=\sum_{\bk} c^\dagger_\bk
	h_\bk
	c_\bk, \quad
	c_\bk\equiv
	(c_{A,+,\bk}, c_{B,+,\bk},c_{A,-,\bk}, c_{B,-,\bk})^T
	,
	\label{Defck}
\end{equation}
with
\begin{gather}
h_\bk = -t [ \cos k_y + 2 \cos (k_y/2) \cos(\sqrt{3} k_x/2)] \sigma_x \nonumber \\+ t[\sin k_y - 2 \sin(k_y/2) \cos(\sqrt{3} k_x/2)] \sigma_y \nonumber \\+ 4\lambda \sin (\sqrt{3} k_x/2)[\cos (\sqrt{3} k_x/2) - \cos (3 k_y/2)] \sigma_z \tau_z.
\label{hk}
\end{gather}
Here, $\sigma$ and $\tau$ indicate the Pauli matrices in $\{A,B\}$ sublattices and angular momentum $\pm$ orbitals, respectively. The Hamiltonian (\ref{hk}) is invariant under 6-fold rotation symmetry implemented as
\beq
\label{C6}
U_6 h_{O_6 \bk} U_6^\dagger = h_\bk, \qquad U_6^\dagger U_6 = U_6^6 = 1,
\eeq
where $U_6$ defined by $\hat C_6 c_{\vec{k}} \hat C_6^{-1}=U_6c_{O_6\vec{k}}$ is a unitary action given explicitly by
\beq
\label{U6}
U_6 = \sigma_x e^{i \frac{\pi}{3} \tau_z}.
\eeq
In addition, the Hamiltonian is invariant under the spinless time-reversal symmetry given by $\T = \tau_x \K$, with $\K$ denoting complex conjugation. We will consider both variants of the model with and without time-reversal symmetry with the main difference being in the type of symmetry-allowed terms that can be added to this model.

We note that the Hamiltonian $h_\bk$ is not periodic in $\bk$ under the addition of a reciprocal lattice vector $\bG$. Instead, it changes by a gauge transformation
\beq
h_{\bk + \bG} = V_\bG h_\bk V_\bG^\dagger, \quad V_\bG = \left(\begin{array}{cc} e^{i\bG \cdot \bt_A} & 0 \\ 0 &  e^{i\bG \cdot \bt_B}  \end{array} \right)_{\rm AB},
\label{VG}
\eeq
where $\bt_{A,B}$ correspond to the position of the A/B sublattice sites relative to the center of the unit cell (which we take to be the rotation center). They are given by 
\beq
\label{tAB}
\bt_{A,B} = (\sqrt{3}/2, \pm 1/2).
\eeq
The reason for the relation (\ref{VG}) is that the Hamiltonian (\ref{hk}) takes into account the position of the A/B sublattice sites inside the unit cell. It follows that the periodicity of the Bloch states in momentum space has the form $\psi_{\bk + \bG} = V_\bG \psi_\bk$. The extra phase factor $e^{i \bk \cdot \bt_\alpha}$ enters all the formulas for Fourier transform. It will be sometimes easier to deal with periodic quantities by performing the unitary transformation 
\beq
h_\bk \rightarrow \H_\bk = V_\bk^\dagger h_\bk V_\bk \Rightarrow \H_{\bk+\bG} = \H_\bk.
\label{VhV}
\eeq
The Bloch states of $\H_\bk$ are periodic $\psi_{\bk + \bG} = \psi_\bk$. Such transformation changes the rotation operator $U_6$, making it momentum-dependent, as follows

\begin{equation}
	U_{6\bk} = V_\bk^\dagger U_6 V_{O_6 \bk} = 
	\begin{pmatrix}
	0 & e^{i\frac{\pi}{3}\tau_z}e^{-\frac{i}{2}(\sqrt{3}k_x+3k_y)}\\
	e^{i\frac{\pi}{3}\tau_z}  & 0
	\end{pmatrix}_{\rm AB}. 
\end{equation}
The extra phase factor corresponds to the fact that under rotation, a sublattice B site transforms into a sublattice A site in the same unit cell while a sublattice A site transforms to a sublattice B site in a different unit cell. 

\subsection{Phases and non-interacting classification}
We now discuss the possible phases of the model. We start by considering a fixed value of $t$ and investigate the phases of the model as a function of $\lambda$. We notice that when $\lambda=0$, the bands are gapless at two points $K,K'\equiv(\pm\frac{4\pi}{3\sqrt{3}},0)$ in the Brillouin zone as shown in Fig.~\ref{BrillouinZone}. This suggests that the model has two distinct phases for $\sgn(\lambda) = \pm 1$ separated by a gap closing phase transition. However, to establish this, we need to ensure that such critical point cannot be removed by adding any symmetry-allowed perturbation which is done by expanding around the $K$ and $K'$ valleys and consider the low-energy effective theory, which has the form of a Dirac Hamiltonian
\beq
\label{HD}
h_\bk = v_F (k_x \sigma_x \gamma_z + k_y \sigma_y) - m \sigma_z \gamma_z \tau_z.
\eeq
Here, $v_F=3t/2$, $m=3\sqrt{3}\lambda$, $\gamma$ denotes the Pauli matrices in valley ($K$ or $K'$) space while $\sigma$ and $\tau$ denote the Pauli matrices in the sublattice (A or B) and orbital ($p_\pm$), respectively, as in (\ref{hk}). $\hat C_6$ and $\T$ symmetries are implemented in the continuum theory as
\beq
U_6 = - \sigma_x \gamma_x e^{-i \frac{\pi}{3} \sigma_z \gamma_z} e^{i \frac{\pi}{3} \tau_z}, \quad \T = \gamma_x \tau_x \K,
\eeq
where the extra factor of $\gamma_x$ reflects the fact that both $\hat C_6$ and $\T$ symmetries map the valleys to each other.

\begin{figure}[h]
	\centering
	\includegraphics[width=0.7\linewidth]{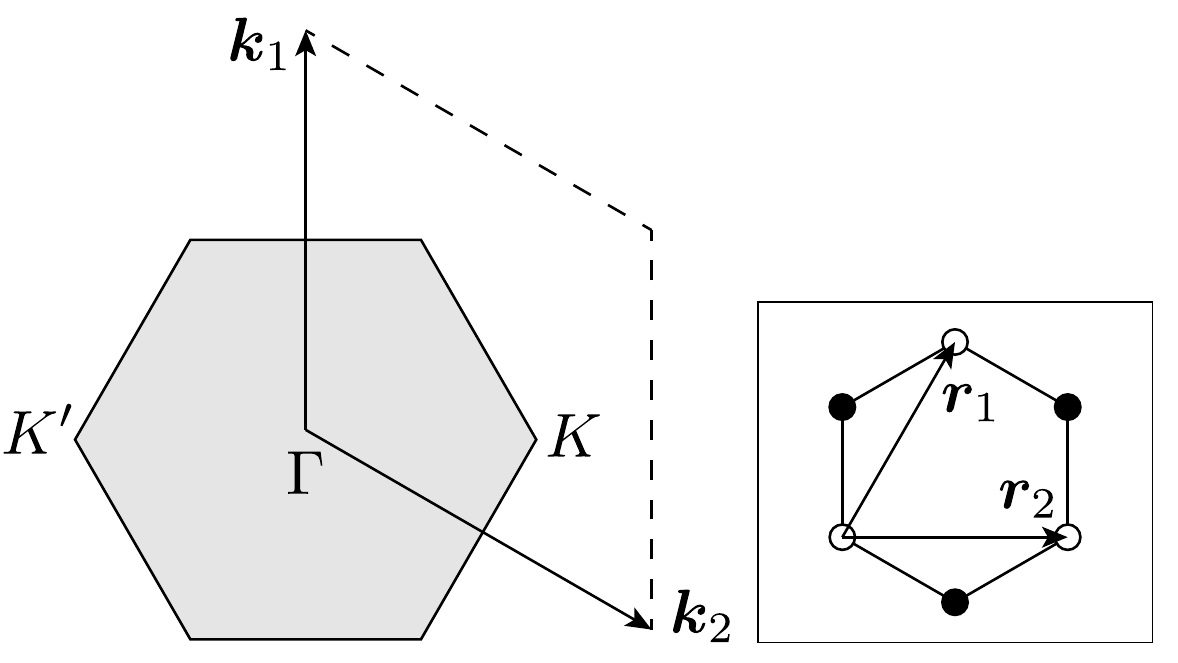}
	\caption{Brillouin zone (gray region) of the honeycomb lattice. $\vec{k}_1,\vec{k}_2$ are the translation basis of the momentum space, dual to the real space translation basis $\vec{r}_1,\vec{r}_2$ as shown in the inset figure. $K,K'$ are the two gapless points when $\lambda=0$. }
	\label{BrillouinZone}
\end{figure}

Let us now consider all possible symmetry-allowed terms. We first note that we only need to consider terms which anticommute with $\sigma_x \gamma_z$, $\sigma_y$, and $\sigma_z \gamma_z \tau_z$. The reason is that any term which commutes with either $\sigma_x \gamma_z$ or $\gamma_y$ can be removed by a gauge transformation and its main role would be moving the zero of the Dirac Hamiltonian (\ref{HD}) at $\lambda=0$ away from the point $k_x = k_y = 0$. In addition, any term which commutes with $\sigma_z \gamma_z \tau_z$ will move the gapless point in parameter space away from $\lambda=0$ without removing it.

Restricting ourselves to mass terms which anti-commute with $\sigma_x \gamma_z$, $\sigma_y$, and $\sigma_z \gamma_z \tau_z$ leaves us with the following possibilities: $\sigma_z \gamma_{0,z} \tau_{x,y}$ or $\sigma_x \gamma_{x,y} \tau_{x,y}$. All these mass terms has the form $\bm_\bk \cdot \btau \otimes \Lambda$ where $\bm_\bk$ is a vector in the $x-y$ plane and $\Lambda = \sigma_z \gamma_{0,z}$ or $\sigma_x \gamma_{x,y}$. We note that, imposing $\T$ symmetry would rule-out some of these mass terms. More specifically, $\T$ restricts $\Lambda$ to $\sigma_z$ or $\sigma_x \gamma_{x,y}$. We now consider the action of $\hat C_3 = \hat C_6^2$ given by $U_3 = e^{i \frac{\pi}{3} \sigma_z \gamma_z} e^{-i \frac{\pi}{3} \tau_z}$ on these terms. Since $\Lambda$ commutes with $\sigma_z \gamma_z$, $\hat C_3$ only acts on the $\tau$ part in each of the mass terms leading to
\begin{align}
\label{C3m}
\bm_{\bk} \cdot \btau &=  \bm_{O_3 \bk} \cdot e^{-i \frac{\pi}{3} \tau_z} \btau e^{i \frac{\pi}{3} \tau_z} = \bm_{O_3 \bk} \cdot (O_3^{-1} \btau) \nonumber \\
&= (O_3 \bm_{O_3 \bk}) \cdot \btau \quad \Rightarrow \quad \bm_{O_3 \bk} = O_3^T \bm_\bk.
\end{align}
This immediately implies that $\bm_\bk$ vanishes at $\bk = 0$ since $\bm_0 = O_3 \bm_0$ is only possible if $\bm_0 = 0$. As a result, we conclude that all possible mass terms which can be added to (\ref{HD}) vanish at $\bk = 0$ and will therefore not open a gap in the critical Hamiltonian $\lambda = 0$.

To obtain the non-interacting classification, we consider stacking $k$ copies of the Hamiltonian (\ref{HD}) on top of each other. We can now perform a very similar analysis to the one performed above by considering mass terms of the form $\bm_\bk \cdot \btau \otimes \Lambda \otimes \Gamma$ with $\Gamma$ a $k \times k$ hermitian matrix. Since rotation acts diagonally in the copies, we can derive the same condition (\ref{C3m}) for all the mass terms and deduce that they all vanish at $\bk = 0$. As a result, we conclude that, in the absence of interactions, the phases constructed by stacking several copies of the shift insulator are all distinct, leading to a $\Z$ classification. We remark that our classification relies on $U(1)$ charge conservation which will be always assumed in this work. 

\subsection{Edge theory}
\label{EdgeTheory}
Let us now investigate the edge theory for the model in (\ref{hk}). We start by considering periodic boundary conditions in the $y$ direction and open boundary conditions along the $x$ direction as illustrated in Fig.~\ref{Haldanepbc} and compute spectra numerically. We take $t=1,\lambda = 0.2$ and take the number of sites for each sublattice in each plaquette row to be $N=100$. 
\begin{figure}[h]
	\centering
	\includegraphics[width=0.8\linewidth]{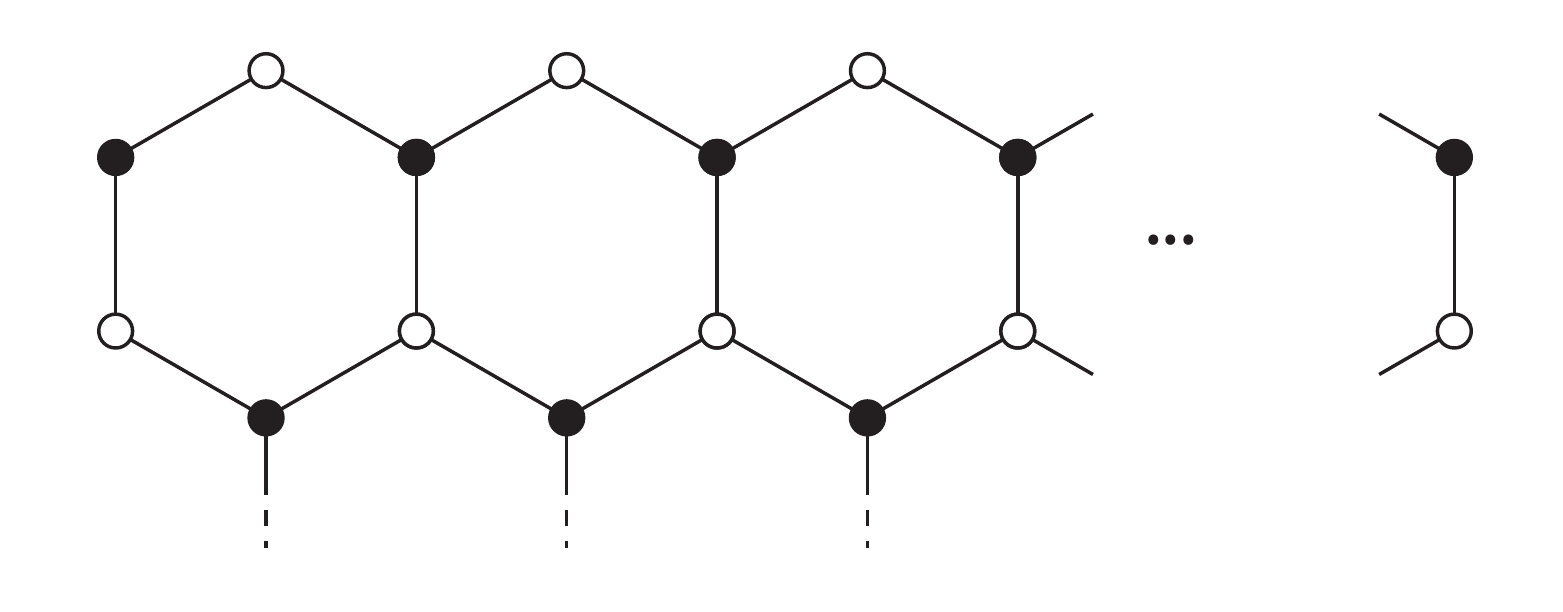}
	\caption{Periodic boundary condition in $y$ direction. }
	\label{Haldanepbc}
\end{figure}
We can see from the left panel of Fig.~\ref{JtermTurnOn} that there is a pair of gapless linearly dispersing modes. To investigate the stability of such gapless modes against symmetry-preserving perturbations, we add the following $C_6$-symmetric (and also $\T$-symmetric) hopping term to the tight-binding Hamiltonian
\begin{equation}
	H_J=J\sum_{\ex{i,j}}\mu_{ij}c^\dagger_{i,+}c_{j,-}+h.c., 
	\label{HJTerm}
\end{equation}
where $\mu_{ij}=\exp[-\rmi(2\varphi(\vec{r}_{ij})-\pi)]$ is a phase factor depending on the hopping direction from $j$ to $i$ and is illustrated in Fig.~\ref{Jterm}. 

\begin{figure}[h]
	\centering
	\includegraphics[width=0.4\linewidth]{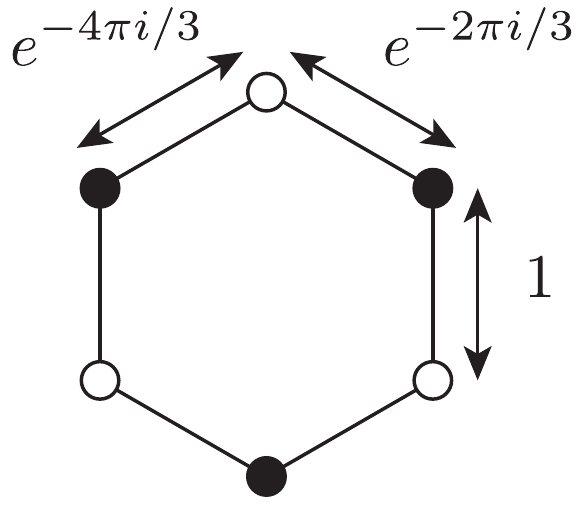}
	\caption{Values of $\mu_{ij}$ for different hopping directions from $j$ to $i$. }
	\label{Jterm}
\end{figure}

The spectrum computed numerically for $J=0.1$ is shown in the right panel of Fig.~\ref{JtermTurnOn} and we can clearly see a gap opening in the edge spectrum indicating the instability of the edge modes to the addition of symmetry-preserving perturbations. 

\begin{figure}[h]
	\centering
	\includegraphics[width=1\linewidth]{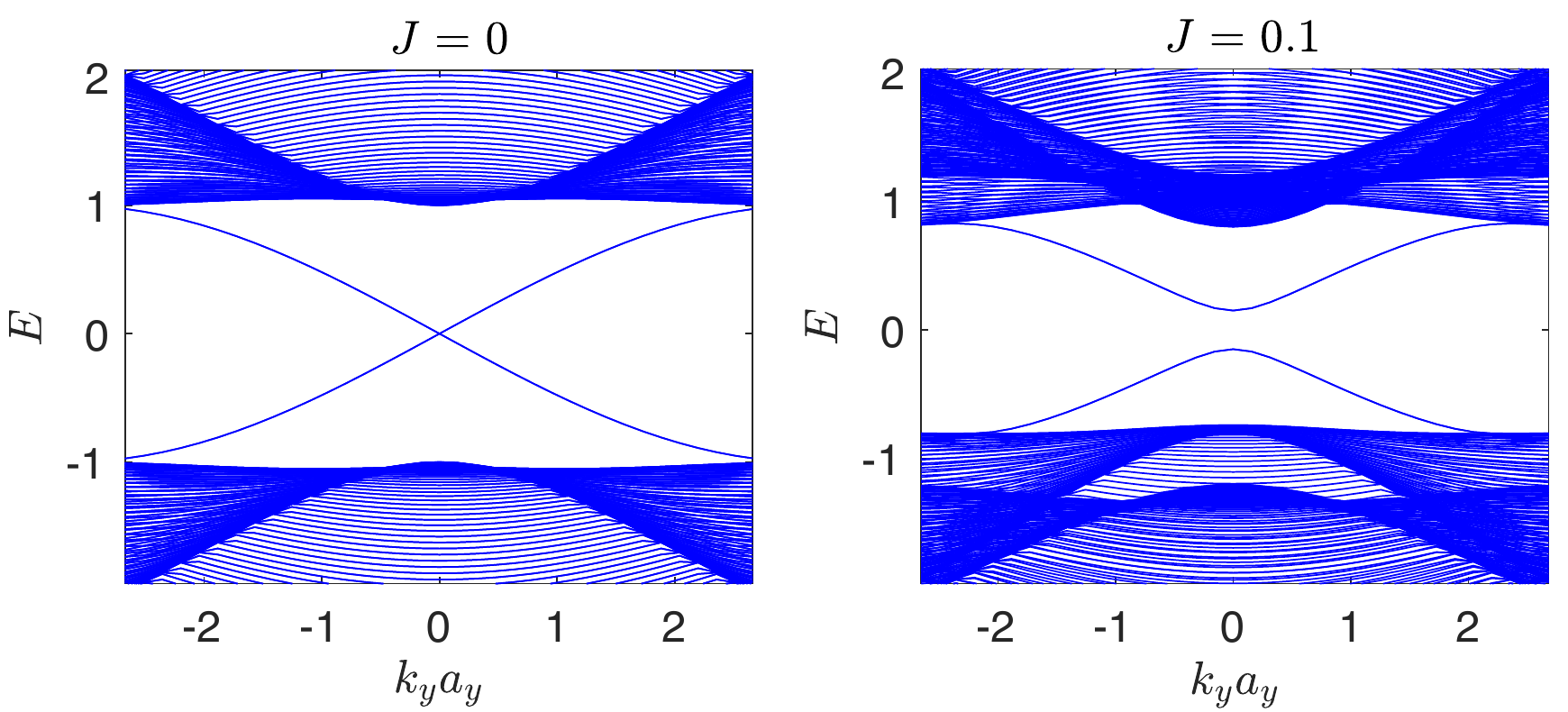}
	\caption{Comparison between $J=0$ (left) and $J=0.1$ (right) with $t=1,\lambda=0.2$ fixed. }
	\label{JtermTurnOn}
\end{figure}

The absence of edge modes can be verified by deriving the edge theory following Refs.~\onlinecite{Khalaf17, Khalaf18}. This is done by considering the low energy Hamiltonian (\ref{HD}) on some $C_6$ symmetry-compatible surface and denoting the in-plane normal to the surface by $\bn = (\cos \varphi, \sin \varphi,0)$. The edge is implemented by taking the mass parameter $m_0$ in (\ref{HD}) to change spatially $m_0 \rightarrow M(\br)$ such that $M(\br) = m_0$ deep inside the sample and $M(\br) = -m_0$ outside it. We decompose the momentum as $\bk = k_t \bt + k_n \bn$ with $\bt$ denoting the unit vector along the tangent to the edge $\bt = (-\sin \varphi, \cos \varphi, 0)$. Following the standard procedure, the details of which are relegated to Appendix \ref{Edge}, we get the edge Hamiltonian
\beq
\H_{\rm edge} = v_F (\bk \cdot \bt) (\tilde \bsigma \cdot \bt) = v_F k_t (\tilde \bsigma \cdot \bt),
\label{Hedge} 
\eeq
with $\tilde \bsigma = (\sigma_x \gamma_z, \sigma_y, \sigma_z \gamma_z)$. The edge Hamiltonian has the spectrum $\pm v_F |k_t|$ which is manifestly gapless.

In order to investigate the stability of the edge Hamiltonian, we follow the previous section and add to the bulk Hamiltonian the mass terms $\bm_\br \cdot \btau \otimes \Lambda$ with $\Lambda = \sigma_z \gamma_{0,z}$ or $\sigma_x \gamma_{x,y}$ in the absence of $\T$ symmetry and $\sigma_z$ or $\sigma_x \gamma_{x,y}$ in the presence of $\T$ symmetry.  Here, we write the mass term as a function of position $\bm_\br$ since the surface breaks translation symmetry. Repeating the argument leading to (\ref{C3m}), we deduce that the mass transforms under 3-fold as rotations $\bm_{O_3 \br} = O_3^T \bm_\br$. 

We now show that we can gap-out the edge in the presence or absence of time-reversal symmetry. We consider the mass term $\sigma_z \bm_\br \cdot \btau$ whose edge projection is (see Appendix \ref{Edge} for derivation)
\beq
 \gamma_z [(\bm \cdot \bn) (\bn \cdot \tilde \bsigma) - (\bm \cdot \bt) \tilde \sigma_z],
\eeq
which upon adding to the edge Hamiltonian leads to the spectrum $\pm \sqrt{v_F^2 k_t^2 + \bm_\br^2}$. This spectrum is gapped as long as $\bm_\br$ does not vanish which can be easily achieved, e.g. by choosing $\bm_\br = (\cos \varphi, \sin \varphi, 0)$. 

\section{Topological response}
\label{Response}
Given the absence of gapless edge modes, it is natural to ask whether we can find any topological signature of the shift insulator model. In this section, we will address this question from the perspective of topological response. 

Since the model consists of two Chern insulators with opposite Chern numbers and (atomic) orbital angular momenta, we would expect angular momentum pumping in the presence of a magnetic flux. Although its quantitative details can be very complicated (we also need to consider the contribution from lattice angular momenta), we will anticipate the existence of such an effect. Let us imagine describing this effect in some low energy theory by a mutual Chern-Simons term $-(S/2\pi)B\wedge\rmd A$ between the electromagnetic gauge field $A_\mu$ and an emergent gauge field $B_\mu$ associated with the rotation symmetry. We can then rewrite this term as $-(S/2\pi)A\wedge\rmd B$ which also implies a electromagnetic charge response to the flux of the field $\rmd B$. In the rest of this section, we will give precise definitions of these two complementary responses at the lattice level and then try to detect this mutual Chern-Simons term. We will mainly focus on the Haldane model with a general orbital angular momentum $L_z$. The response of the shift insulator model can then be obtained by simply adding the response of two copies of the Haldane model with opposite signs of $L_z$ and $\lambda$ (cf.~Eq.~\ref{Shift}). 
In Sec.~\ref{Disclination}, we interpret the $\rmd B$ flux as disclinations (rotation symmetry defects) and numerically measure the number of electrons trapped by these defects, then in Sec.~\ref{TorusMonopoleFlux} we put the system on a torus and measure the change of ground state angular momentum due to monopole fluxes. Analytical derivations of these topological responses are given in the rest subsections. 

\subsection{Disclination}
\label{Disclination}
\subsubsection{Constructing a disclination}
In the case of honeycomb lattice, a disclination is made by reducing or increasing the number of $1/6$ sectors around a rotation center. To be more precise, the Hamiltonian of a disclination system is constructed in the following way: we first construct the ordinary Hamiltonian based on the orbitals of a fan-shaped sector consisting of $(6-n_\Omega)$ number of $60\degree$ wedges, and then identify the open edges by $\ket{\varphi+2\pi(1-n_\Omega/6)}=\exp(-\rmi(2\pi\Phi+n_\Omega\pi L_z/3))\ket{\varphi}$, where $\varphi$ is the ordinary polar angle as shown in Fig.~\ref{MakingDisclinationSchematic}a, $\Phi$ is the magnetic flux through the central plaquette in unit of $\Phi_0=h/e$ and $L_z$ is the orbital angular momentum in unit of $\hbar$. Under this identification, overlapped couplings are required to be identical due to the $C_6$ symmetry and should be counted only once. In numerical calculations, we constructed the fan-shaped sector by gluing up a few $60\degree$ triangles, as illustrated in Fig.~\ref{MakingDisclinationSchematic}b. 
\begin{figure}[h]
	\centering
	\includegraphics[width=1\linewidth]{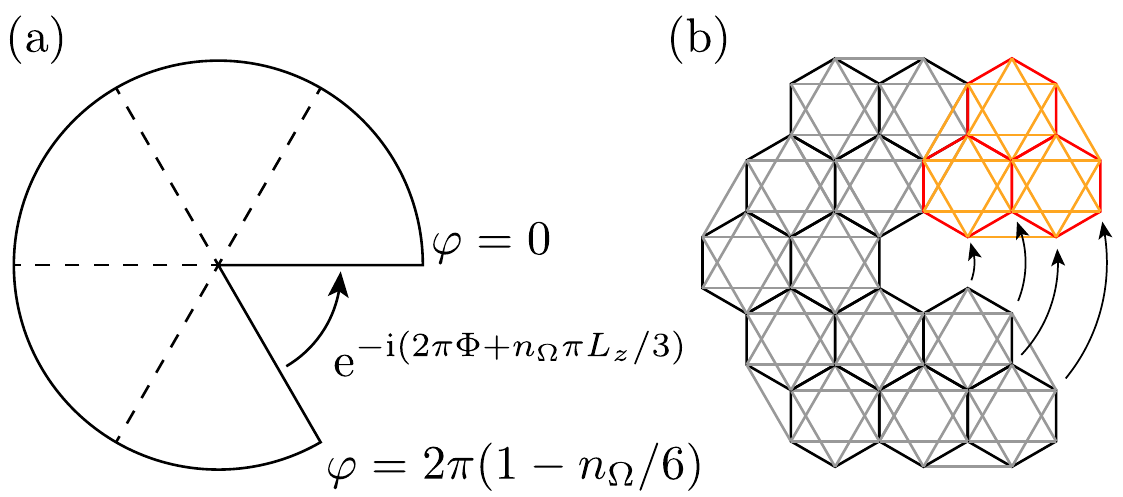}
	\caption{(a) Schematic illustration of the construction of disclinations. (b) An example of our setup in numerical calculations. }
	\label{MakingDisclinationSchematic}
\end{figure}

\subsubsection{Disclination charge}
We found that disclinations in the Haldane model and the shift insulator in general trap fractional electric charge, which is a signature of nontrivial topology in the presence of $C_6$ rotation symmetry. We will now explain this result in detail. 

Let us first consider the Haldane model. For an isolated disclination located at $\vec{r}=0$ and a Fermi energy within the bulk gap, we denote by $\Delta Q(r)$ the \emph{extra number of electrons} inside radius $r$, i.e. total number of electrons subtracted by the bulk half-filling background which is $1/2$ the number of lattice sites. $\Delta Q(r)$ should converge to a constant when $r$ is deep in the bulk, and this is the trapped charge that we are looking for. The disclination charge can also be defined for a more general $C_6$ symmetric lattice model in a similar way: we choose a simple loop that encloses the disclination at a large distance and does not cross through any lattice sites, count the number of electrons inside this loop, and then subtract off a “bulk background” which corresponds to the bulk distribution of the electron density; equivalently, we imagine there are some nucleus charges which exactly cancel the electron charge distribution in the bulk, and we are only counting the extra number of electrons near the disclination. 

We computed Haldane model disclination charge numerically with a half-filling Fermi energy on even total number of sites and a system size much larger than the correlation length. Results for a few interesting cases with $\Phi=0$ are listed in Table~\ref{HaldaneDeltaQ}. 
\begin{table}[h]
	\centering
	\subfloat[$t=1,\lambda=0.2$]{
	\begin{tabular}{c|ccc}
		\hline\hline
		\backslashbox{$L_z$}{$n_\Omega$} & 0 & 1 & 2 \\ [0.5ex] % inserts table %heading
		\hline
		$0$ & $0$ & $1/4$  &$0$* \\
		$+1$ & $0$ & $1/12$ & $1/6$ \\
		$-1$ & $0$ & $5/12$ & $-1/6$ \\
		$1/2$ & $0$ & $1/6$ & $1/3$ \\
		$-1/2$ & $0$ & $1/3$ & $-1/3$\\
		\hline
	\end{tabular}
	}
	\quad
	\subfloat[$t=1,\lambda=-0.2$]{
	\begin{tabular}{c|ccc}
		\hline\hline
		\backslashbox{$L_z$}{$n_\Omega$} & 0 & 1 & 2 \\ [0.5ex] % inserts table %heading
		\hline
		$0$ & $0$ & $1/4$  &$0$* \\
		$+1$ & $0$ & $5/12$ & $-1/6$ \\
		$-1$ & $0$ & $1/12$ & $1/6$ \\
		$1/2$ & $0$ & $1/3$ & $-1/3$ \\
		$-1/2$ & $0$ & $1/6$ & $1/3$ \\
		\hline
	\end{tabular}
	}
	\caption{Haldane model disclination charge with zero magnetic flux. Here we also include examples of half-integer $L_z$ since they are physically allowed. }
	\label{HaldaneDeltaQ}
\end{table}

These results all fit into the following formula which we will derive later, except for the two starred ones ($L_z=0$, $n_\Omega=2$ and $\lambda=\pm 0.2$): 
\begin{align}
\Delta Q(\text{bulk})&=-\sgn(\lambda)\left( \Phi+\frac{n_\Omega}{6}L_z \right)\nonumber \nonumber\\
&+\frac{1}{4}n_\Omega\sgn(t)+k\in[-\frac{1}{2},\frac{1}{2}], 
\label{HaldaneChargeExperienceFormula}
\end{align}
where $k$ is a proper integer such that $\Delta Q(\text{bulk})\in[-\frac{1}{2},\frac{1}{2}]$ is satisfied. Then how about the two exceptional cases where $L_z=0$ and $n_\Omega=2$? First note that the formula above actually has ambiguity in these cases: $\Delta Q$ can be $\pm 1/2$ and we do not know which is the correct choice. Physical reason for this ambiguity is the following. When $L_z=0$, $n_\Omega=2$ and in the thermodynamic limit, there are two degenerate $E=0$ eigenstates: one is a bound state near the disclination and the other one is an edge state at the outer boundary. In the half-filling case, it is ambiguous which of these two states should be occupied, therefore $\Delta Q$ also has an ambiguity. In a finite system, however, this degeneracy is slightly lifted and the true eigenstates have significant distribution near both the disclination apex and the outer boundary, therefore $\Delta Q$ lies in between $\pm 1/2$. Moreover, the system has a particle-hole symmetry when $L_z=0$ and $n_\Omega=2$ as we will see later, and this is the reason why $\Delta Q$ is exactly pinned to zero. The disclination charge of the Haldane model with $L_z=0$ and $t>0$ was previously obtained in Ref.~\onlinecite{Coh2013}. 

Similarly, we computed the disclination charge of the $\mathrm{C}_6$ shift insulator with half-filling, and the values of $\Delta Q(\text{bulk})$ for a few interesting cases are listed in Table~\ref{DeltaQC6}. These results are consistent with what we found for the Haldane model: up to an integer, the disclination charge of the $\mathrm{C}_6$ shift insulator is the sum of contributions from the two Haldane model components, namely the following, 
\begin{equation}
	\Delta Q(\text{bulk})=-\frac{1}{3}n_\Omega\sgn(\lambda)+\frac{1}{2}n_\Omega\sgn(t)+\text{integer}. 
	\label{DisclinationChargeShift}
\end{equation}
\begin{table}[h]
	\centering
	\begin{tabular}{c|c|ccc}
		\hline\hline
		$\Phi$ &\backslashbox{$\lambda$}{$n_\Omega$} & 0 & 1 & 2 \\ [0.5ex] % inserts table %heading
		\hline
		\multirow{2}{*}{$0$} & $0.2$ & $0$ & $1/6$  &$1/3$ \\
		& $-0.2$ & $0$ & $5/6$ & $-1/3$ \\
		\hline
		\multirow{2}{*}{$0.24$} & $0.2$ & $0$ & $1/6$ & $-2/3$\\
		& $-0.2$ & $0$ & $5/6$ & $2/3$\\
		\hline
		\multirow{2}{*}{$0.25$} & $0.2$ & $0$ & $1/6$ & $-2/3$\\
		& $-0.2$ & $0$ & $-1/6$ & $2/3$\\
		\hline
	\end{tabular}
	\caption{Shift insulator disclination charge for $t=1$ and interlayer coupling $J=0.2$ (see Eq.~\ref{HJTerm}). For the $\Phi=0$ case, we also computed the charge with $J=-0.1,0,0.1$ and the results are the same. }
	\label{DeltaQC6}
\end{table}

We see from Table~\ref{DeltaQC6} that threading magnetic flux through the disclination hole may change the trapped charge by an integer. This may sound unexpected because our system is now fully gapped and there is no edge state going between the upper and lower bands. It turns out from numerics that, at half-filling, some edge state or disclination bound state near the upper band can be occupied, and therefore an integer jump of the disclination charge can happen when there is an energy crossover between edge and bound states. 

From these numerical results, we find that the disclination charge seems to always take nice fractional numbers when the magnetic flux $\Phi$ is turned off. This is in fact not a coincidence and the disclination charge has to satisfy certain quantization rules, as we now explain. To avoid analyzing the edge effect, let us imagine putting a few disclinations on a closed surface. There are many ways of doing so, as illustrated in Fig.~\ref{ClosedSurfaceProof}, in which the most familiar example might be the `buckyball' where we have twelve pentagon disclinations on a sphere. Suppose the disclinations are all far away from each other so that the disclination charge for each of them is well-defined. We then have the rule: 
\begin{align}
&\text{total disclination charge}+\text{bulk background} \nonumber\\
&=\text{total number of electrons}. 
\label{ClosedSurfaceSumRule}
\end{align}
The total number of electrons on a closed surface is obviously an integer. For all the examples in Fig.~\ref{ClosedSurfaceProof}a-c, the bulk background is also an integer. To see this, simply note that there are even number of triangular wedges, which is in fact true in general, and that any two triangular wedges must have an integer background charge since they can combine into a torus which is disclination free. With these observations in mind, we can now easily derive some quantization rules for the disclination charge from Eq.~\ref{ClosedSurfaceSumRule}. Using the construction in Fig.~\ref{ClosedSurfaceProof}b, we know that $\Delta Q$ for a square disclination ($n_\Omega=2$) must be an integer multiple of $1/6$. Applying this result to Fig.~\ref{ClosedSurfaceProof}c, we know that for a general $N$-gon disclination, $\Delta Q$ must be an integer multiple of $1/12$. It is interesting to note that these constraints are satisfied by the formula \eqref{HaldaneChargeExperienceFormula} if and only if $L_z$ takes integer or half-integer values. 
Many other rules can be similarly derived with different closed surface constructions. 

\begin{figure}
	\centering
	\includegraphics[width=0.9\linewidth]{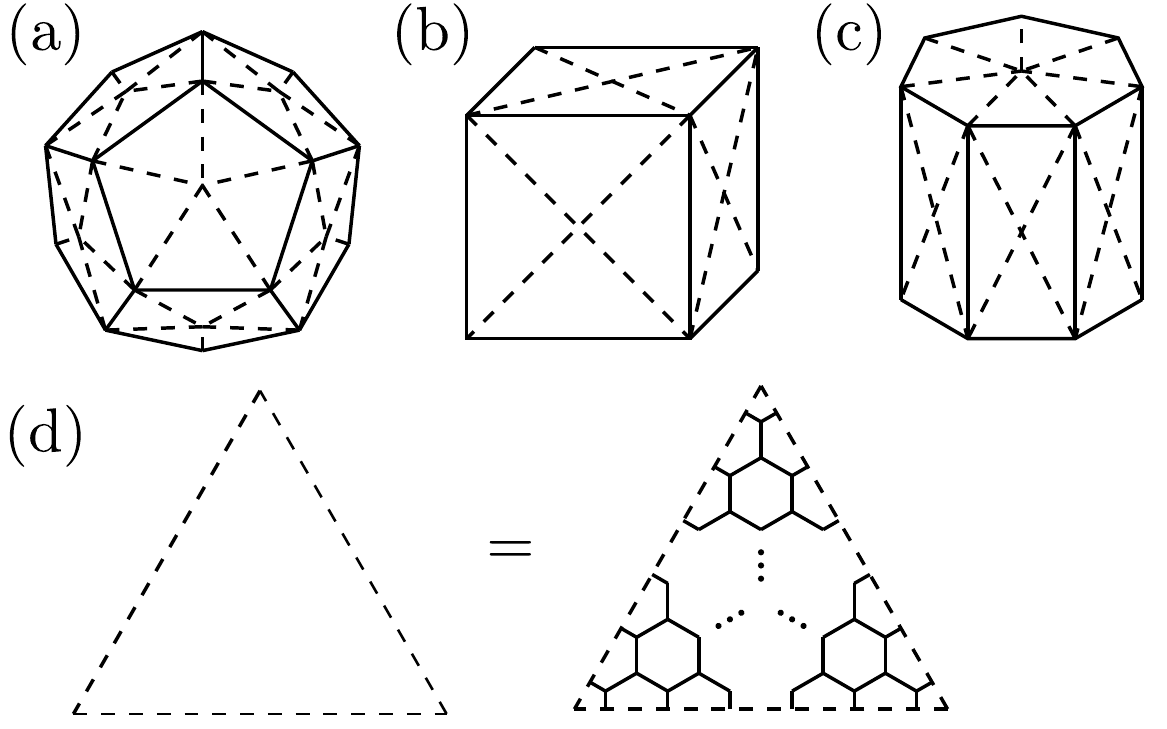}
	\caption{Quantization of the disclination charge. (a) Twelve pentagon disclinations on a sphere in the shape of a dodecahedron, e.g. the buckyball. (b) Six square disclinations on a sphere in the shape of a cube. (c) Two $N$-gons and $N$ squares on a sphere in the shape of a cylinder. 
	(d) Each of the small triangles in the above three examples contains many plaquettes. The triangle vertices coincide with plaquette centers and the triangle edges are perpendicular to plaquette edges. } In (a)-(c), disclinations are all located at face centers, and all corners are regular since there are six surrounding wedges. 
	\label{ClosedSurfaceProof}
\end{figure}

It is now easy to see that the disclination charge modulo integers is topologically stable: as we continuously tune the Hamiltonian along a symmetry-preserving path which does not close the bulk gap, the disclination charge can not continuously change since it is quantized. However, it is indeed possible that $\Delta Q$ jumps by an integer, because we may need to add or remove a few electrons to make sure the Fermi energy always lies in the bulk gap as we tune the Hamiltonian. 

Later in Section \ref{TopologicalAtomic}, we will show that the disclination charge \eqref{DisclinationChargeShift} of shift insulators can not be reproduced by any atomic insulators with the same number of filled bands. Together with the topological stability proved above, this implies that the shift insulator indeed has nontrivial topology (regarding atomic insulators as trivial) protected by the rotation symmetry. 

\subsection{Torus monopole flux}\label{TorusMonopoleFlux}
\subsubsection{Setup}\label{C6SymmetricGauge}
A $C_6$ symmetric torus (modulus $\tau=\rme^{\rmi\pi/3}$) is made by identifying the opposite edges of a regular hexagon. We choose a lattice orientation such that this hexagon (expanded torus) has parallel edges to the plaquettes. The edge length is taken to be $3N$ times that of a plaquette, where $N$ is a positive integer. In Fig.~\ref{ExpandedTorus}a, we show one specific example with $N=3$. 
\begin{figure}[h]
	\centering
	\includegraphics[width=1\linewidth]{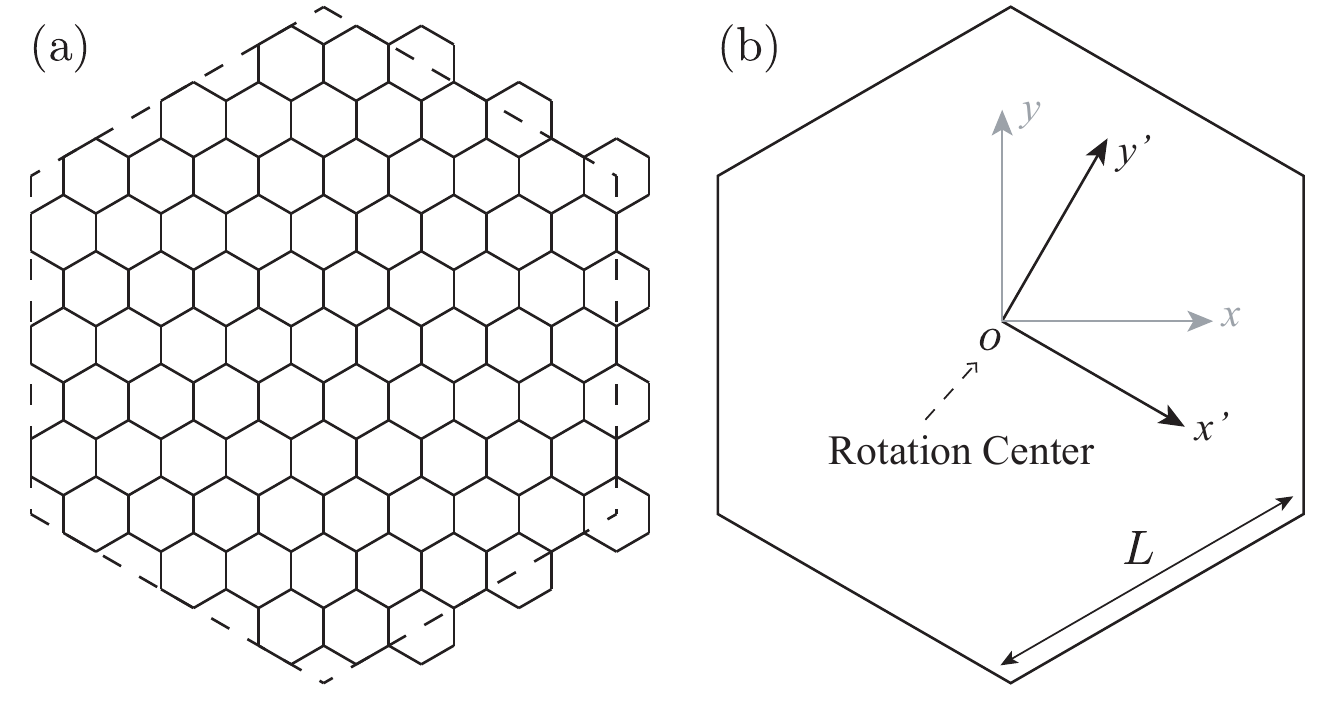}
	\caption{(a) An $N=3$ example of the expanded honeycomb lattice tori used in our calculation. (b) The $(x',y')$ coordinates on the expanded torus. }
	\label{ExpandedTorus}
\end{figure}

We require a $C_6$ rotation symmetry, it is therefore convenient to find a $C_6$ symmetric gauge potential for monopole fluxes. We consider a total magnetic flux $\Phi$ injected through one point on the torus and then spreads out evenly from the whole surface. Note that when the injected flux $\Phi$ happens to be an integer multiple of $\Phi_0$, this field configuration is equivalent to a uniformly distributed monopole field, since the injection of flux quanta has no effect on the lattice model given that it does not cross through any hopping path. It is however useful to keep in mind such a flux injection picture, because it provides us a continuous connection between different integer monopole charges $m\equiv \Phi/\Phi_0$. 
We choose this flux injection point as our $C_6$ rotation center. Using a hexagon-shaped expansion of the torus and a Cartesian coordinate system $(x',y')$ defined in Fig.~\ref{ExpandedTorus}b, we write down the following gauge potential $\vec{A}(\vec{r})$ ($\vec{r}\neq 0$) for this field configuration: 
\begin{align}
&A_{x'}=-\frac{3L}{2}\Lambda\left(y'-\sgn(y')\frac{\sqrt{3}L}{2}\right)\delta(x'), \\
&A_{y'}=\Lambda\left(x'-\sgn(x')\frac{3L}{4}\right),
\end{align}
where $\Lambda=2\Phi/(3\sqrt{3}L^2)$ and $L$ is the system size (see Fig.~\ref{ExpandedTorus}b). This gauge potential is not yet $C_6$ symmetric, but we can easily symmetrize it by rotating and taking average. 

We match the rotation center with a plaquette center, and the rotation generator $\hat C_6$ is defined as 
\begin{equation}
\hat C_6\ket{\vec{r}}=\rme^{-\rmi\pi L_z/3}\ket{O_6\vec{r}}, 
\end{equation}
which rotates wavefunctions counterclockwise by $\pi/3$. Note that with our specific choice of the lattice orientation, no regularization for the discontinuities in the gauge potential written down above is needed. 

\subsubsection{Angular momentum of the gapped ground state}
Now we present our numerical result on the rotation eigenvalue of the gapped torus ground state of Haldane model with monopole fluxes, i.e. when $\Phi/\Phi_0$ is an integer. 

As we change the injected magnetic flux $\Phi$ continuously, there are states transporting between the upper and lower bands. More specifically, when the monopole charge $m$ changes from $0$ to another integer $m_0$, $\sgn(\lambda)m_0$ number of states are transported from the lower band to the upper band. An example is shown in Fig.~\ref{StateTransport}. 
\begin{figure}[h]
	\centering
	\includegraphics[width=0.7\linewidth]{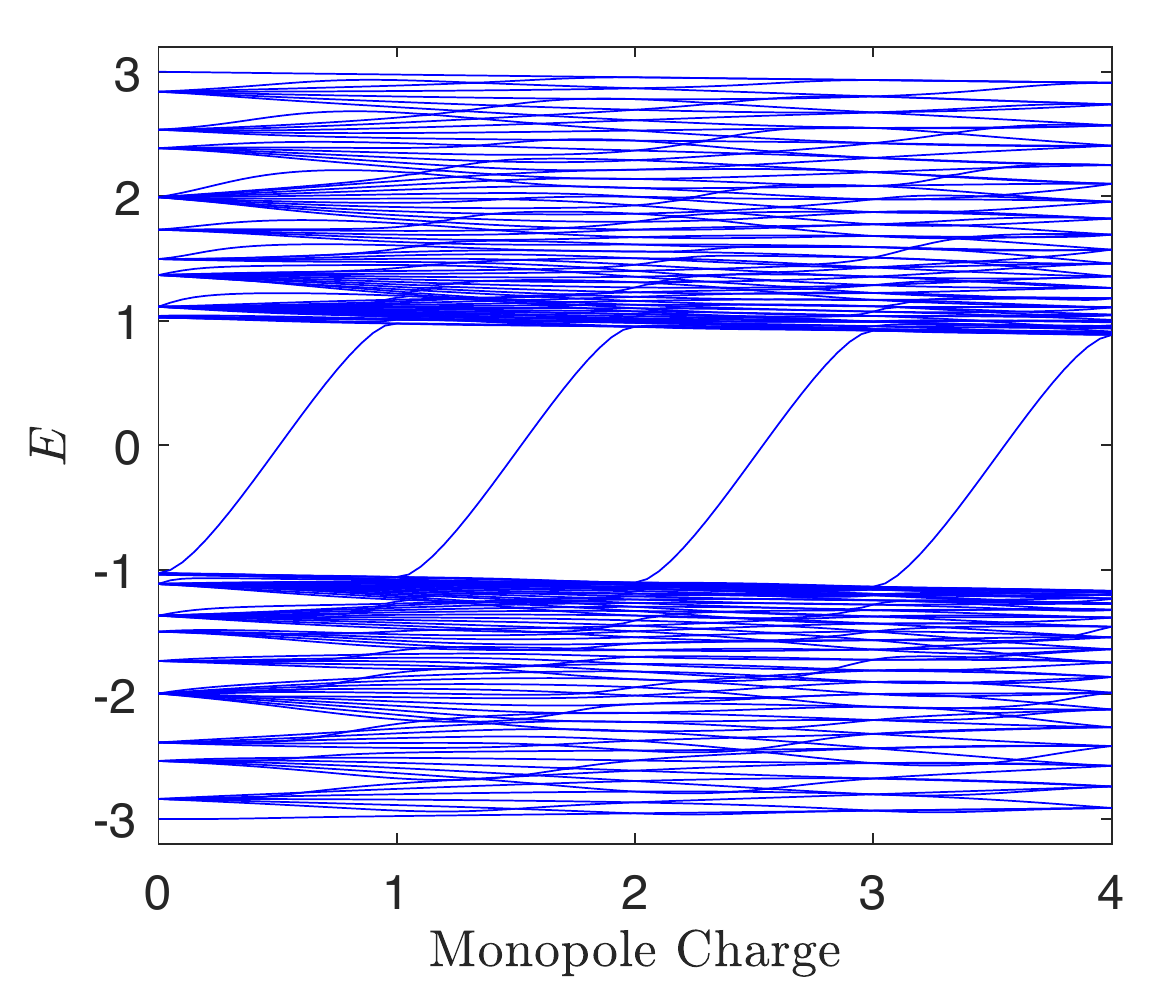}
	\caption{Torus energy spectrum of the Haldane model at $t=1,\lambda=0.2$ with respect to the monopole charge $\Phi/\Phi_0$. The system size is set to be $N=3$ (81 plaquettes). }
	\label{StateTransport}
\end{figure}

For a given integer monopole charge $m$, when there is an obvious band gap ($N$ is not too small and $|m|$ not too large), we define the gapped ground state by filling up all states in the lower band. We can then compute the eigenvalue of this state under the action of $\hat C_6$. We say the ground state has spin $s$ if this eigenvalue is $\exp(-\rmi s\pi/3)$, note that $s$ is an integer defined modulo $6$. 

\begin{table}[h]
	\centering
	\begin{tabular}{c|c|c|ccccccccccc}
		\hline\hline
		$t$ & $\lambda$ & \backslashbox{$N$}{$m$} & $0$ & $1$ & $-1$ & $2$ & $-2$ & $3$ & $-3$ & $4$ & $-4$ & $5$ & $-5$\\ [0.5ex] 
		\hline
		\multirow{4}{*}{$1$} & \multirow{2}{*}{$0.2$} & 2,4,6,8,10 &1 &2 &5 &2 &2 &1 &4 &5 &5 &2 &5 \\
		& & 3,5,7,9,11 &4 &5 &2 &5 &5 &4 &1 &2 &2 &5 &2\\
		\cline{2-14}
		& \multirow{2}{*}{$-0.2$} & 2,4,6,8,10 &5 &1 &4 &4 &4 &2 &5 &1 &1 &1 &4 \\
		& & 3,5,7,9,11 &2 &4 &1 &1 &1 &5 &2 &4 &4 &4 &1 \\
		\hline
		\multirow{2}{*}{$-1$} & $0.2$ & $2,\cdots,11$ &1 &5 &2 &2 &2 &4 &1 &5 &5 &5 &2 \\
		\cline{2-14}
		& $-0.2$ & $2,\cdots,11$ &5 &4 &1 &4 &4 &5 &2 &1 &1 &4 &1 \\
		\hline
	\end{tabular}
	\caption{Spin of the gapped torus ground state of Haldane model in a few cases. Orbital angular momentum $L_z$ is set to zero. }
	\label{GndStateSpin}
\end{table}
In Table~\ref{GndStateSpin}, we show our result of the ground state spin in a few interesting cases with $L_z=0$. These numbers all satisfy the following formula: 
\begin{align}
s(m)&=\frac{3m}{2}\sgn(t)+\left( 1-\frac{m^2}{2}\right)\sgn(\lambda)\nonumber \\
&+
\begin{cases}
3N &(t>0)\\
0 & (t<0)
\end{cases}
\mod{6}, 
\end{align}
from which we can compute the monopole induced spin change: 
\begin{equation}
\Delta s(m)\equiv s(m)-s(0)=\frac{3m}{2}\sgn(t)-\frac{m^2}{2}\sgn(\lambda). 
\end{equation}
We will later give derivations for this result. 

So far we restrict ourselves to the special case $L_z=0$, but it is not hard to generalize to a general nonzero $L_z$. The only difference is a phase factor in the definition of $\hat C_6$. From the state transporting phenomenon discussed previously, we immediately have
\begin{equation}
\Delta s(m)=\frac{3m}{2}\sgn(t)-\frac{m^2}{2}\sgn(\lambda)-mL_z\sgn(\lambda). 
\end{equation}

Similar to the disclination charge response, the spin pumping response of shift insulators are obtained by summing over the contributions from the two Haldane model components, and it can not be reproduced by any atomic insulators with the same number of filled bands as we will later show in Section \ref{TopologicalAtomic}. 

\subsection{Exact lattice theory approach to the topological responses}
In this subsection, we give an exact lattice theory approach to the topological responses considered above, focusing on the Haldane model. We start with some general properties of the model and then derive the formulas for disclination charge and monopole induced ground state spin change. 

\subsubsection{Symmetries and Dualities}\label{ParticleHoleDualitySection}
Let us first investigate useful symmetries and dualities in our disclination model. All results for the $n_\Omega=0$ disclination are also true for the ordinary Haldane model, although these two are slightly different (see Fig.~\ref{MakingDisclinationSchematic}b). 

Recall from Fig.~\ref{MakingDisclinationSchematic} that a disclination is made by first constructing the ordinary Hamiltonian based on the orbitals of a fan-shaped sector, and then identifying the open edges by $\ket{\varphi+2\pi(1-n_\Omega/6)}=\exp(-\rmi(2\pi\Phi+n_\Omega\pi L_z/3))\ket{\varphi}$. We will refer to this gauge choice as the branch cut gauge hereafter. 
When considering a Haldane model with all of its orbitals having the same angular momentum $L_z$, the effect of this angular momentum can be characterized by an effective magnetic flux $\Phi_\text{eff}=n_\Omega L_z/6$. We will therefore set $L_z=0$ throughout this subsection. 

The original $C_6$ rotation symmetry is generalizable to a $C_{6-n_\Omega}$ symmetry for general $n_\Omega$ and $\Phi$. In the branch cut gauge, this symmetry is generated by 
\begin{equation}
\hat C_{6-n_\Omega}: c_{(r,\varphi)}\mapsto \rme^{-2\pi\rmi\Phi/(6-n_\Omega)}c_{(r,\varphi+\pi/3)}, 
\label{GeneralizedRotSym}
\end{equation}
with the identification
\begin{equation}
c_{(r,\varphi+2\pi(1-n_\Omega/6))}=\rme^{2\rmi\pi\Phi}c_{(r,\varphi)}. 
\end{equation}
We will see later that the $C_6$ symmetry defined in Sec.~\ref{C6SymmetricGauge} with the symmetric gauge is equivalent to the definition here. 

There are two $\mathbb{Z}_2$ transformations which relate the four possible sign choices of $t$ and $\lambda$, and their generators are given by
\begin{align}
	&S: 
	\begin{pmatrix}
	c_A\\
	c_B
	\end{pmatrix}
	\mapsto
	\begin{pmatrix}
	-c_A\\
	c_B
	\end{pmatrix},\\
	&\mathcal{T}: c_{\vec{r}}\mapsto \mathcal{K}c_{\vec{r}}, 
\end{align}
with $\mathcal{K}$ being the complex conjugation operator. The sublattice pseudo-spin $S$ flips the sign of $t$, while the spinless time-reversal $\mathcal{T}$ flips the sign of $\lambda$ and maps the flux $\Phi$ to $-\Phi+\frac{1}{2}n_\Omega$ where $-\Phi$ comes from the complex conjugation of $\exp(-2\pi\rmi\Phi)$, and the $\frac{1}{2}n_\Omega$ term is due to the mismatch of sublattices at the glued edges when $n_\Omega$ is odd. We also define another anti-unitary transformation $\mathcal{P}\equiv\mathcal{T}S$ under which we have
\begin{align}
H(t,\lambda,n_\Omega,\Phi)&\stackrel{\mathcal{P}}{\cong} H(-t,-\lambda,n_\Omega,-\Phi+\frac{1}{2}n_\Omega)\nonumber\\
&=-H(t,\lambda,n_\Omega,-\Phi+\frac{1}{2}n_\Omega), 
\label{ParticleHoleDuality}
\end{align}
where $H(t,\lambda,n_\Omega,\Phi)$ denotes the Hamiltonian of a Haldane model disclination characterized by the parameters $t,\lambda,n_\Omega$ and $\Phi$. This is a particle-hole duality between Haldane model disclinations with the same hopping amplitudes $(t,\lambda)$ and the underlying lattice geometry. In particular, when $\Phi=\frac{1}{4}n_\Omega+\frac{1}{2}k$ with $k\in\mathbb{Z}$, this duality becomes a particle-hole symmetry \cite{Coh2013}. 

We note that instead of using the complex conjugation, one can as well apply a reflection to reverse the sign of $\lambda$. This operation, however, will change the geometry of the lattice unless it is reflection symmetric. 

\subsubsection{Exact zero energy bound states and edge states}\label{ZeroEnergyStatesSection}
We will take the following strategy to compute the Haldane model disclination charge: suppose we know the charge at some specific value of magnetic flux, and we verify that there is no degeneracy at the Fermi energy between disclination bound states and edge states, then the charge for almost all other values of magnetic flux can be derived from Hall conductance. The requirement of no degeneracy is important; it guarantees that the charge is continuous at this point and the Hall conductance argument is applicable. 
It is therefore important to first understand the properties of disclination bound states and edge states, which we now elaborate. More specifically, we will prove that for a $p_z$ Haldane model disclination with the magnetic flux 
\begin{equation}
\Phi=-\frac{1}{4}(6-n_\Omega)\sgn(t)\sgn(\lambda)\mod\mathbb{Z}, 
\end{equation}
the following statements hold: 
\begin{enumerate}
	\item There exists a zero energy bound state near the disclination. 
	\item This zero energy bound state level cannot have even degeneracy. 
	\item When the total number of states is even, there also exists at least one zero energy edge state at the outer boundary. 
\end{enumerate}

Let us start with $n_\Omega=0$. It is well-known that a magnetic $\pi$-flux in the ordinary Haldane model will trap a zero energy bound state, a consequence of quantum Hall effect and the particle-hole duality. The same thing must also happen for our disclination model with $n_\Omega=0$. We expect this state to be generically nondegenerate, so it must be self-dual under the particle-hole symmetry and also be an eigenstate of the rotation symmetry. These symmetry properties imply the following $\pi/3$ periodicity of the bound state wave function $\Psi(\vec{r})$: 
\begin{equation}
	\Psi(\varphi+\pi/3)=\pm\rmi\Psi(\varphi). 
\end{equation}
To determine which sign is taken, we need to look at the actual wave function. 
We checked from numerics that, when $t=1,\lambda=0.2,n_\Omega=0,\Phi=1/2$, there is indeed a nondegenerate zero energy bound state near the disclination, and its wave function on the central plaquette is shown in Fig.~\ref{E0BoundStateWaveFunction}. 
Together with the $\mathbb{Z}_2$ transformations relating different signs of the hopping terms, this tells us that the zero energy bound state for $|t|=1$ and $|\lambda|=0.2$ is nondegenerate and satisfies
\begin{equation}
\Psi(\varphi+\pi/3)=-\sgn(t)\sgn(\lambda)\rmi\Psi(\varphi). 
\label{PeriodicityCondition}
\end{equation}
So far we confined ourselves to specific values of $t$ and $\lambda$, but the results we obtained above are actually quite stable: the particle-hole symmetry guarantees that as long as the bulk is gapped, there is always at least one zero energy bound state satisfying the periodicity condition \eqref{PeriodicityCondition}. 
\begin{figure}[h]
	\centering
	\includegraphics[width=0.5\linewidth]{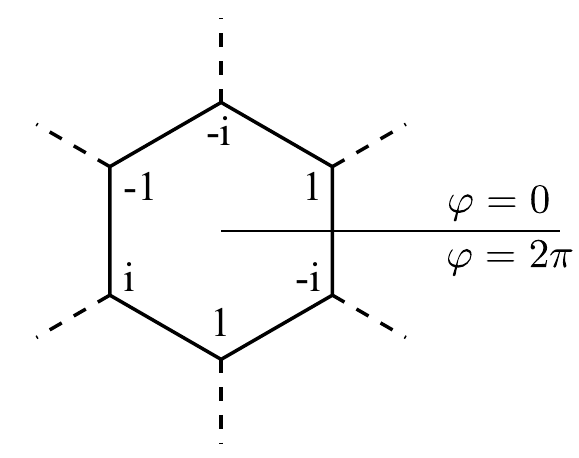}
	\caption{Wave function of the zero energy bound state of the Haldane model disclination with $n_\Omega=0,\Phi=1/2,t=1$ and $\lambda=0.2$. }
	\label{E0BoundStateWaveFunction}
\end{figure}

Now we are ready to look at general values of $n_\Omega$. Suppose we release the branch cut of the disclination discussed above ($n_\Omega=0,\Phi=1/2$) and extend it to a Riemann surface with infinite layers, then the zero energy bound state we just found can also be extended to this whole Riemann surface according to its $\pi/3$ periodicity condition. We can now cut off an arbitrary fan-shaped sector from this Riemann surface and glue it back into another disclination, we then find a zero-energy bound state for each $n_\Omega$! The phase jump across the new branch cut is $n_\Omega$-dependent, and we can compute it in the polar coordinate $\varphi$ as 
\begin{align}
&\Psi(\varphi+(6-n_\Omega)\pi/3)=(-\sgn(t)\sgn(\lambda)\rmi)^{6-n_\Omega}\Psi(\varphi)\nonumber\\
&=\exp(-\rmi\frac{\pi}{2}(6-n_\Omega)\sgn(t)\sgn(\lambda))\Psi(\varphi),  
\end{align}
which corresponds to a magnetic flux $\Phi=-\frac{1}{4}(6-n_\Omega)\sgn(t)\sgn(\lambda)\mod\mathbb{Z}$. We therefore obtain the desired conclusion. 

Now we would like to prove that, when the total number of states is even, with the same magnetic flux and a large system size, there also exists a zero energy edge state at the outer boundary. 
We first prove that, the zero energy bound state level we just found cannot have \emph{even} degeneracy. As before, we start with $|t|=1$ and $|\lambda|=0.2$. 
Suppose even degeneracy occurs. We extend the state $\ket{\Psi}$ we just found into an orthogonal basis of this degenerate bound state subspace, with all basis vectors being simultaneous eigenvectors of $H$ and $\hat C_{6-n_\Omega}$. Under the particle-hole symmetry $\mathcal{P}$, these vectors are mapped into another set of orthogonal eigenvectors of both $H$ and $\hat C_{6-n_\Omega}$ with the same set of eigenvalues, which is because 
\begin{align}
[\mathcal{P},H]&=0,\\
\hat C_{6-n_\Omega}\mathcal{P}&=-\rme^{4\pi\rmi\Phi/(6-n_\Omega)}\mathcal{P}\hat C_{6-n_\Omega}. 
\end{align}
Note that $\ket{\Psi}$ must be self dual under $\mathcal{P}$, since we know it is the only zero energy bound state when $n_\Omega=0$. This implies that another eigenvector, say $\ket{\Psi'}$, from this basis must have a self dual rotation eigenvalue, i.e.  $\mathcal{P}\ket{\Psi'}$ must have the same rotation eigenvalue as $\ket{\Psi'}$, though $\mathcal{P}\ket{\Psi'}$ need not equal to $\ket{\Psi'}$. Suppose $\hat C_{6-n_\Omega}\ket{\Psi'}=\lambda\ket{\Psi'}$, we have
\begin{equation}
\hat C_{6-n_\Omega}(\mathcal{P}\ket{\Psi'})=-\lambda^*\rme^{4\pi\rmi\Phi/(6-n_\Omega)}\mathcal{P}\ket{\Psi'}\equiv \lambda'\mathcal{P}\ket{\Psi'}. 
\end{equation}
Then $\lambda=\lambda'$ requires that $\lambda=\pm\rmi\rme^{2\pi\rmi\Phi/(6-n_\Omega)}$ (for some $n_\Omega$ not both of them are allowed), which is equivalent to the following $\pi/3$ periodicity condition (see Eq.~\ref{GeneralizedRotSym}): 
\begin{equation}
\Psi'(\varphi+\pi/3)=\mp\rmi\Psi'(\varphi). 
\end{equation}
We can now map $\Psi'$ to a zero energy eigenstate for $n_\Omega=0, \Phi=1/2$ using the same Riemann surface technique, and it is not hard to see that at $n_\Omega=0$, $\ket{\Psi'}$ is still orthogonal to $\ket{\Psi}$. We then have a contradiction: there should only be one zero energy bound state for $n_\Omega=0$. We have therefore proved that the zero energy bound state level can only have odd degeneracy when $|t|=1,|\lambda|=0.2$. This can be directly generalized to general $t$ and $\lambda$ as long as the bulk gap does not close, since this degeneracy can only change by even numbers due to the particle-hole symmetry. 

Now using the particle-hole symmetry again, we know that when the total number of states is even, there also exist at least one zero energy edge state at the outer boundary. 

\subsubsection{Computing the disclination charge}
Our strategy for computing the Haldane model disclination charge is already outlined at the beginning of Sec.~\ref{ZeroEnergyStatesSection}.
Let us consider an even total number of sites with half-filling as we did in the numerical calculation. From Sec.~\ref{ParticleHoleDualitySection}, we know that the system is particle-hole symmetric when $\Phi=\frac{1}{4}n_\Omega+\frac{1}{2}k$ with $k\in\mathbb{Z}$. We then know for sure that the charge $\Delta Q$ is zero in these cases. However, from Sec.~\ref{ZeroEnergyStatesSection}, we know that when
\begin{equation}
\Phi=\left( \frac{1}{4}n_\Omega -\frac{1}{2}\right)\sgn(t)\sgn(\lambda)\mod\mathbb{Z}, 
\label{DegeneracyFluxes}
\end{equation}
there is a degeneracy between bound states and edge states at the Fermi energy $E=0$, so $\Delta Q$ can be discontinuous here. 
We then consider the following fluxes: 
\begin{equation}
\Phi=\left( \frac{1}{4}n_\Omega -\frac{1}{2}\right)\sgn(t)\sgn(\lambda)\pm\frac{1}{2}+\text{integer}, 
\label{NonDegeneracyFluxes}
\end{equation}
which still have the particle-hole symmetry. Let us assume the edge states to be always non-degenerate near $E=0$, then quantum Hall effect implies that, as the magnetic flux continuously changes by $1$, the edge spectrum must shift up or down by one level, depending on the sign of the Hall conductance. The particle-hole symmetry then guarantees that when the flux is as in Eq.~\ref{NonDegeneracyFluxes}, the Fermi energy $E=0$ is in the middle of two edge state levels, and there is no discontinuity in $\Delta Q$. Using quantum Hall effect again, we then obtain
\begin{align}
\Delta Q&=-\sgn(\lambda)\left( \Phi-\frac{1}{4}n_\Omega\sgn(t)\sgn(\lambda) \right)\mod \mathbb{Z}\nonumber\\
&=-\sgn(\lambda)\Phi+\frac{1}{4}n_\Omega\sgn(t)\mod\mathbb{Z}, 
\end{align}
except when there are degeneracies between bound states and edge states at the Fermi energy. This agrees with our numerical calculation. If we further assume that $\Delta Q$ only has discontinuities at the fluxes in Eq.~\ref{DegeneracyFluxes}, we can get a more precise prediction: $\Delta Q=0$ when $\Phi=\left( \frac{1}{4}n_\Omega -\frac{1}{2}\right)\sgn(t)\sgn(\lambda)\mod\mathbb{Z}$ and otherwise
\begin{equation}
\Delta Q=-\sgn(\lambda)\Phi+\frac{1}{4}n_\Omega\sgn(t)+k\in(-\frac{1}{2},\frac{1}{2}).  
\end{equation}

\subsubsection{Derivation of the ground state spin change}\label{DerivingGndStateSpinChange}
Now we give a derivation of $\Delta s(m)$ in the large system size limit. When $N$ is large, the phenomenon of states being transported between the upper and lower bands is nothing but quantum Hall effect. If we zoom in to the neighborhood of the rotation center where the monopole flux $\Phi=m\Phi_0$ is injected, it looks like we are on an infinite plane with a magnetic flux $-\Phi$ threaded only through one plaquette. Then from quantum Hall effect, we know that when the flux $\Phi$ is continuously tuned, there will be bound states going between the upper and lower bands. The observed change of ground state spin is purely due to these transported bound states, so all we need is to understand their rotation symmetry property. 

Let us first choose a convenient gauge for theoretical analysis. In our original $C_6$ symmetric gauge $\vec{A}(\vec{r})$, hopping terms in the Hamiltonian are modified according to the Peierls substitution: 
\begin{equation}
\tau\ket{\vec{r}'}\bra{\vec{r}}\mapsto \tau\rme^{\rmi\theta}\ket{\vec{r}'}\bra{\vec{r}},~~~\text{with}~\theta=-2\pi\int_{\vec{r}}^{\vec{r}'}\vec{A}\cdot\rmd\vec{l}, 
\end{equation}
where the integral path is the straight line between $\vec{r}$ and $\vec{r}'$, the electron charge is taken to be $-e$ and the unit of magnetic flux is $\Phi_0=h/e$. We will now push all the gauge potential to a branch cut at the positive $x$-axis. We define a new orbital basis as 
\begin{equation}
\ket{\tilde{\vec{r}}}=\exp\left( -\rmi(2\pi)\int_{\vec{r}_0}^{\vec{r}}\vec{A}\cdot\rmd\vec{l} \right)\ket{\vec{r}}, 
\end{equation}
where $\vec{r}_0$ is an arbitrary fixed point and the integral path should not cross the branch cut. In this basis, all hopping terms in the Hamiltonian come back to their original values except for those crossing through the branch cut. If a hopping bond crosses the positive $x$-axis counterclockwise/clockwise, there is an additional phase factor $\exp(-2\rmi\pi(\pm\Phi))$. Including this factor is equivalent to imposing the angular periodicity condition $\ket{\widetilde{2\pi+\varphi}}=\exp(-2\rmi\pi\Phi)\ket{\tilde{\varphi}}$, which is exactly how we realized magnetic fluxes in the disclination calculation. The $C_6$ rotation generator $\hat C_6$ acts on this basis as 
\begin{equation}
\hat C_6\ket{\tilde{\varphi}}=\rme^{\rmi\pi\Phi/3}\widetilde{\ket{\varphi+\pi/3}}, 
\label{TorusTildeBasisRotationRep}
\end{equation} 
which is also the same as what we used before. 

Consider a monopole charge $m=m_0+1/2$ where $m_0$ is an integer. Near the rotation center, the Hamiltonian can only see a $\pi$-flux, therefore we know that there is a zero energy bound state $\ket{\Psi}$ whose wave function $\tilde\Psi(\vec{r})$ in the new orbital basis has the following periodicity: 
\begin{equation}
\tilde\Psi(\varphi+\pi/3)=-\sgn(t)\sgn(\lambda)\rmi\tilde\Psi(\varphi). 
\label{TorusBoundStateRotationProperty}
\end{equation} 
All other bound states, if any, come in dual pairs under the particle-hole symmetry $\mathcal{P}$. The rotation eigenvalue of this state $\ket{\Psi}$ is precisely responsible\footnote{This statement is definitely true if all bound state levels move monotonically as we observed in numerics, but one can in fact prove it even without this assumption using the particle-hole duality and its commutation relation with the rotation symmetry. We will not go into such boring details here. } for the ground state spin change as we tune the monopole charge from $m_0$ to $m_0+1$. From Eq.~\ref{TorusTildeBasisRotationRep} and \ref{TorusBoundStateRotationProperty}, we have 
\begin{align}
\hat C_6\ket{\Psi}&=\rmi\sgn(t)\sgn(\lambda)\rme^{\rmi\pi\Phi/3}\ket{\Psi}. 
\end{align}
Putting in $\Phi=-m_0-1/2$, we find the spin of this state: 
\begin{equation}
s_{\ket{\Psi}}(m_0)=m_0+\frac{1}{2}(1-3\sgn(t)\sgn(\lambda)), 
\end{equation} 
defined by $\hat C_6\ket{\Psi}=\exp(-\rmi s_{\ket{\Psi}}\pi/3)\ket{\Psi}$. We can now derive the ground state spin change: 
\begin{align}
s(m+1)-s(m)&=-\sgn(\lambda) s_{\ket{\Psi}}(m)\\
\Rightarrow\Delta s(m)&=-\frac{m^2}{2}\sgn(\lambda)+\frac{3m}{2}\sgn(t), 
\end{align}
where $m$ is an integer. This coincides with our previous empirical formula. 
\subsection{Continuum theory approach}\label{ResponsesContinuumApproach}
Both types of topological responses of the Haldane model can also be computed with a continuum theory approach. Since the techniques we used are already developed in Ref.~\onlinecite{graphenecone,haldanemodelcone}, we will only state our result here and leave the details to Appendix \ref{HaldaneContinuumTheory}. 

We found that a doubled Haldane model disclination system splits into two sectors, each being equivalent to a Haldane model with the same $t,\lambda$ but \emph{without} disclination. The two sectors have effective magnetic fluxes
\begin{equation}
	\Phi_\pm=\Phi+\frac{1}{6}n_\Omega L_z\pm\frac{1}{4}n_\Omega\gamma
\end{equation}
with $\gamma\equiv\sgn(t)\sgn(\lambda)$, and their boundary conditions represent infinite mass Haldane models in the disclination hole with the \emph{same} and \emph{opposite} signs of mass, respectively. We can now compute the charge accumulation near the disclination hole boundary: only the sector with $\Phi_-$ contributes and we have
\begin{align}
	\text{charge}&=\frac{1}{2}\times 2\times(-\sgn(\lambda))\Phi_-+\text{integer}\nonumber\\
	&=-\sgn(\lambda)\left(\Phi+\frac{1}{6}n_\Omega L_z \right)+\frac{1}{4}n_\Omega\sgn(t)\nonumber\\
	&+\text{integer}, 
\end{align}
which correctly reproduces the disclination charge formula. 

We also computed the wave function of the zero energy disclination bound state at $n_\Omega=0,\Phi=1/2$ and found that its rotation property is the same as the result in Sec.~\ref{ZeroEnergyStatesSection}. This leads to the formula for the monopole induced ground state spin change. 

\subsection{A Chern-Simons theory description}
In this subsection, we present an effective Chern-Simons theory  for the $C_6$ symmetric Haldane model, and we will see that it unifies our previous results on disclinations and torus monopole fluxes. 

The $C_6$ symmetric Haldane model has a charge $U(1)$ symmetry coupled to the electromagnetic gauge potential $A\equiv A_\mu\rmd x^\mu$ and discrete lattice symmetries, namely translations and the six-fold rotation. We assume that, in a low energy continuum theory, the $C_6$ rotation becomes an emergent $U(1)$ symmetry which can be coupled to a probe gauge field $B\equiv B_\mu\rmd x^\mu$. We then write down the following Chern-Simons Lagrangian: 
\begin{equation}
\mathcal{L}_{CS}~d^3x=-\left(\frac{C}{4\pi}A\wedge \rmd A+\frac{S}{2\pi} A\wedge \rmd B\right). 
\label{LCS}
\end{equation}
The first term is nothing but the usual quantum Hall effect with $C=-\sgn(\lambda)$ being the Chern number, and the second term characterizes the coupling of $A$ and $B$ fields. The unit for magnetic flux is taken to be $\hbar/e=\Phi_0/(2\pi)$ here. 

The above effective action is the result after integrating out all the dynamical degrees of freedom. In the initial Lagrangian, the electron number current $J_A^\mu$ is coupled to $A_\mu$ via the term $-qJ^\mu_A A_\mu$, where we have taken the $(+,-,-)$ metric convention and $q=-1$ is the electron charge in the unit of $e$. This implies the equation 
\begin{equation}
J_A^\mu=\frac{\delta\mathcal{L}_{CS}}{\delta A_\mu}=-\left(\frac{C}{2\pi}\epsilon^{\mu\lambda\nu}\partial_\lambda A_\nu+\frac{S}{2\pi}\epsilon^{\mu\lambda\nu}\partial_\lambda B_\nu\right). 
\end{equation}
The electron number $N_e$ in a spatial region $\Sigma$ is thus given by 
\begin{equation}
N_e\equiv\int_\Sigma\rmd^2 xJ_A^0=C\frac{\Phi_A}{2\pi}+S\frac{\Phi_B}{2\pi},
\end{equation}
where $\Phi_A$ and $\Phi_B$ are fluxes for $A$ and $B$ fields, respectively. In the case of disclination, $N_e$ is nothing but $\Delta Q$ and we take $\Phi_B=\pi n_\Omega/3$, we then find that
\begin{equation}
S=-\sgn(\lambda)L_z+\frac{3}{2}\sgn(t). 
\end{equation}
Now consider the ground state spin change due to torus monopole fluxes. Suppose a current $J^\mu_B$ is coupled to $B_\mu$ via the term $J^\mu_B B_\mu$, we then similarly have
\begin{equation}
\int_{T^2}\rmd^2x J_B^0=Sm, 
\end{equation}
where $m$ is the monopole charge. This correctly reproduces the linear terms in $\Delta s(m)$, but the other quadratic term $Cm^2/2$ is missed. 

To understand where this quadratic term comes from, we consider a slightly different situation: instead of considering a monopole flux uniformly distributed on the whole torus, we confine the flux into a tube small compared to the system size, and then thread the flux tube symmetrically around the rotation center. Indeed, we expect our effective model should only apply to the case of tightly confined magnetic flux, because in order to promote the $C_6$ rotation symmetry to a continuous one, we have to zoom in and forget about the global geometry of the torus. Now suppose the flux $\Phi=2m\pi$ is uniformly distributed in a circular disk, and the trapped electric charge $-Cm$ due to the quantum Hall effect is also uniformly distributed in the disk, then using Gauss's law, one can easily compute the angular momentum of the electromagnetic field as (in the SI units): 
\begin{align}
	L_z^\text{EM}&=\int_{0}^{R}\rmd \rho\int\rmd z~2\pi\rho^2\left( -\varepsilon_0 E_\rho B_z \right)=\frac{Cm^2}{2}\hbar, 
\end{align}
where $R$ is the disk radius and $B_z=\Phi/(\pi R^2)$ is the usual magnetic field (not the new $U(1)$ gauge field). This is exactly what we need. 

The mutual Chern-Simons term in Eq.~\ref{LCS} was in fact already studied in Ref.~\onlinecite{Wen92}. By comparing to their result, the probe gauge field $B$ should be identified with the connection 1-form $\omega$ of (the tangent bundle of) the space manifold. To see this more clearly, consider putting a few disclinations on a closed manifold, then calculating the Euler characteristic gives us the following relation: 
\begin{align}
	\sum_n\frac{1}{12}nF_n=1-g, 
\end{align}
where $F_n$ is the number of $(6-n)$-gon disclinations on this manifold and $g$ is the genus. For example, we can put 12 pentagon disclinations on a sphere, which corresponds to $F_1=12$,  $F_n=0$ for all other $n$, and $g=0$. By the Gauss-Bonnet theorem, the number of curvature quanta $(2\pi)^{-1}\int\rmd\omega$ on the manifold is $2(1-g)$, therefore each $(6-n)$-gon disclination corresponds to $n/6$ flux quantum, the same flux of the $B$ field. 

\subsection{Effect of a larger hole}\label{LargerHole}
Let us now briefly discuss the effect of a larger hole at the rotation center of disclination and torus geometries for the Haldane model, based on the lattice theory we developed previously. 

The key to solving the disclination or torus monopole flux problem is the bound state properties. Thus let us first answer the question: what will happen to the bound states if we remove the central plaquette of a disclination? Recall that for $n_\Omega=0$, $\Phi=1/2$ and a minimal disclination hole (the central plaquette is not removed), there exists a zero energy bound state $\ket{\Psi}$ whose wave function satisfies the following $\pi/3$ periodicity:  
\begin{equation}
\Psi(\varphi+\pi/3)=-\sgn(t)\sgn(\lambda)\rmi\Psi(\varphi), 
\end{equation}
where we have used the branch cut gauge. All other zero energy bound states, if any, must come in particle-hole dual pairs. Now let us continuously decouple the central plaquette from the rest of the system in a way which preserves the particle-hole symmetry. It is easy to see that, on the decoupled plaquette, there are two zero energy states $\ket{\Psi_{0,\pm}}$ satisfying
\begin{equation}
\Psi_{0,\pm}(\varphi+\pi/3)=\pm\rmi\Psi_{0,\pm}(\varphi). 
\end{equation}
We then immediately know that, in the rest system, there exists a zero energy bound state $\ket{\Psi_r}$, such that 
\begin{equation}
\Psi_r(\varphi+\pi/3)=\sgn(t)\sgn(\lambda)\rmi\Psi_r(\varphi), 
\end{equation}
opposite to that of $\ket{\Psi}$, and again all other zero energy bound states come in dual pairs. With this result, and using the strategy in Sec.~\ref{DerivingGndStateSpinChange}, one can show that the new formula for the ground state spin change due to torus monopole fluxes becomes
\begin{equation}
\Delta s_r(m)=-\frac{3m}{2}\sgn(t)-\frac{m^2}{2}\sgn(\lambda)-mL_z\sgn(\lambda), 
\end{equation}
i.e. the sign of the $\sgn(t)$ term is reversed. Does the same thing happens for the disclination charge? Note that when we are continuously decoupling the central plaquette, the disclination charge does not change up to an integer, since it is quantized. Now if we remove the decoupled plaquette, and redefine the disclination charge by also dropping the half-filling background on this removed plaquette, then the remained disclination charge becomes $\Delta Q+(6-n_\Omega)/2+\text{integer}$, where we used the fact that the removed plaquette can only takes away integer number of electrons. We then obtain the new disclination charge formula: 
\begin{equation}
\Delta Q_r=-\sgn(\lambda)\left( \Phi+\frac{n_\Omega}{6}L_z \right)-\frac{1}{4}n_\Omega\sgn(t)+\text{integer}. 
\end{equation}
Similarly, the $\sgn(t)$ term gets its sign reversed, consistent with our expectation that $\Delta Q_r$ up to an integer is proportional to the linear part of $\Delta s_r$. 

\section{Wannier representation}\label{WannierRep}
\label{Wannier}
Having investigated the topological response of the shift insulator model in detail in the previous section, we now move to the question of Wannier representability. Recall that a Wannier basis is defined as a basis of localized states that span the same space as the occupied Bloch bands \cite{Marzari12}. They can be obtained from the Fourier transform of the Bloch-like states $|\psi_{n,\bk} \rangle$, which are eigenstates of the flattened Hamiltonian $Q_\bk  = \H_\bk / \sqrt{\H_\bk^2}$, using the expression
\beq
\label{W}
| n,\bR \rangle = \frac{A}{(2\pi)^2} \int_{\rm BZ} d\bk \, e^{-i \bk \cdot \bR} | \psi_{n,\bk} \rangle,
\eeq
where $A$ is the unit cell area and $n$ is a (filled) band index. The Bloch-like states $| \psi_{n,\bk} \rangle$ are {\it not} eigenstates of the Hamiltonian $\H_\bk$, but they span the subspace of filled states. The states $| n,\bR \rangle$ are guaranteed to be exponentially localized around lattice site $\bR$ provided that the Bloch-like states $| \psi_{n,\bk} \rangle$ are smooth and periodic in $\bk$ under reciprocal lattice translations $| \psi_{n,\bk + \bG} \rangle = | \psi_{n,\bk} \rangle$ \cite{Brouder07}(We note that we are using here the periodic Hamiltonian $\H_\bk$ rather than $h_\bk$ for which the periodicity of Bloch states has a different form). 

We note that there is a large gauge freedom in defining the Wannier basis since we can always perform the gauge transformation
\beq
\label{V}
|\psi_{n,\bk}\rangle \rightarrow |\tilde \psi_{n,\bk}\rangle  = \sum_m u_{mn}(\bk) |\psi_{m,\bk} \rangle,
\eeq
with $u(\bk)$ being an $N \times N$ $\bk$-dependent unitary matrix ($N$ here is the number of occupied bands). In a Chern insulator, it is generally impossible to choose a smooth periodic gauge \cite{Thonhauser06} leading to an obstruction to constructing a Wannier basis. When the total Chern number vanishes, e.g. in the presence of time-reversal symmetry, it is always possible to find a smooth periodic gauge and as a result to construct localized Wannier functions \cite{Brouder07}. 

The standard procedure to construct Wannier states uses the so-called projection method \cite{Marzari97,Soluyanov11, Marzari12}. The procedure starts with a trial basis of states $|\tau_{n,\bR} \rangle$, $n = 1,\dots,N$ localized at lattice site $\bR$, which can be transformed to $\bk$-space as
\beq
|\tau_{n,\bk} \rangle = \frac{1}{\sqrt{N}} \sum_\bR e^{i \bk \cdot \bR} |\tau_{n,\bR} \rangle.
\eeq
The next step is to project these states onto the occupied bands by defining the states
\beq
|\chi_{n,\bk} \rangle = P_\bk |\tau_{n,\bk} \rangle.
\eeq
The basis $|\chi_{n,\bk} \rangle$ is not necessarily orthonormal, thus the Löwdin orthogonalization procedure is used to get an orthonormal basis out of it. This is done by defining the overlap matrix
\beq
S_{mn}(\bk) = \langle \chi_{m,\bk} | \chi_{n,\bk} \rangle.
\eeq
If this matrix is invertible for every $\bk$, we can define the orthonormal basis
\beq
|\psi_{n,\bk}\rangle = \sum_m [S(\bk)^{-1/2}]_{mn} |\chi_{m,\bk} \rangle,
\eeq
which will be smooth and periodic in $\bk$ and can thus be used to construct the Wannier states using (\ref{W}). If the Chern number does not vanish, there is no choice of the trial basis $|\tau_{n,\bk} \rangle $ for which the overlap matrix is invertible at every $\bk$. This implies that the projection method will fail for any choice of trial basis which encodes the topological obstruction in finding localized Wannier states.

For a symmetry-protected topological insulator, localized Wannier states can always be constructed (due to vanishing Chern number) but they will necessarily break the protecting symmetry. For example, we can only find localized Wannier states for the 2D time-reversal invariant topological insulator if the trial basis used in the projection method is not time-reversal symmetric \cite{Soluyanov11}. The resulting Wannier states centered at a given site will not be Kramers partners thereby breaking time-reversal symmetry. We note here that the full set of Wannier states is still time-reversal symmetric (this follows from the fact that this set spans the manifestly symmetric space of filled bands) which means that time-reversal acts non-locally on the set of Wannier states. 

The same consideration applies for TCIs where the existence of an obstruction to finding a symmetric Wannier basis can be taken as the defining feature of a topologically non-trivial band insulator \cite{Po17, Bradlyn17}. Within this approach, a band insulator is trivial if the projector onto the filled bands can be symmetrically deformed to the projector onto a set of symmetric localized orbitals i.e. if it is adiabatically deformable to an atomic insulator. In this case, these localized orbitals can be chosen as the trial basis $|\tau_{n,\bk} \rangle $ for the projection procedure explained above, yielding symmetric localized Wannier states. Here, as in the case with time-reversal symmetry, we have to restrict ourselves to atomic insulators where the symmetry acts naturally. The natural action action of crystalline symmetry group on atomic orbitals is specified by a Wyckoff position $x$ and an irrep for the little group $g_x$ \cite{Wyckoff}. Recall that a Wyckoff position $x$ denotes a minimal set of spatial positions that is closed under the action of the crystalline symmetries with the little group $g_x$ denoting the group of symmetries leaving every point in the position $x$ invariant. Using such approach, all possible symmetric atomic orbitals can be enumerated and the question of the existence of a symmetric Wannier representation boils down to whether the projector onto filled bands can be deformed onto a projector of localized orbitals contained in this set \cite{Po17, Bradlyn17}.

\subsection{Wannier obstruction}
\label{Obstruction}
 We now ask the question whether the shift insulator model introduced in Sec.~\ref{Model} admits a symmetric Wannier representation. The symmetry group of the model is the wallpaper group $p6$, for which there are four Wyckoff positions illustrated schematically in Fig.~\ref{Wyckoff}. The first one, denoted by $a$, corresponds to the hexagon center which is invariant under $g_a = C_6$ and thus have Wyckoff multiplicity of 1. The second position, denoted by $b$, corresponds to the positions of the A and B sublattices which are invariant under $g_b = C_3$ but transform into each other under $C_2$ and $C_6$ thus having a Wyckoff multiplicity of 2. The third position, denoted by $c$, corresponds to the edges of the hexagon (which form a Kagome lattice) which are invariant under $g_c = C_2$ but transform into each other under $C_3$ and $C_6$ leading to a multiplicity of 3. The last position consists of any generic point and its orbit under the action of $C_6$ (with multiplicity 6). 
 
 \begin{figure}[h]
\center
\includegraphics[width=0.5\columnwidth]{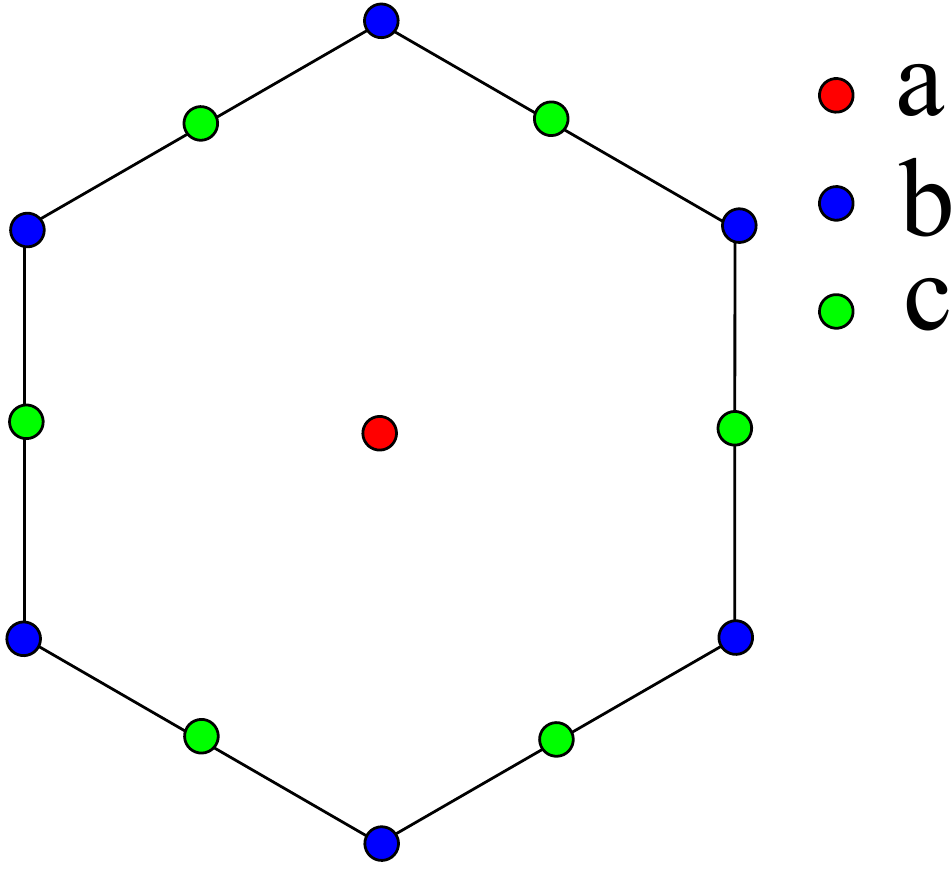}
\caption{Illustration of the Wyckoff positions $a$, $b$, and $c$ for wallpaper group $p6$ corresponding to the hexagon center, corner and edge, respectively.}
\label{Wyckoff}
\end{figure}

A necessary, but not sufficient, condition for the existence of a symmetric Wannier representation in a given model is the existence of an atomic insulator with the same symmetry irreps at high symmetry momenta \cite{Po17, Watanabe17}.
For the shift insulator model, such irreps depend on the signs of the parameters $t$ and $\lambda$ and are given in Table~\ref{Irreps}. There is one $C_6$ invariant point $\Gamma$, two $C_3$ invariant points $K$ and $K'$ related by $C_2$ and three $C_2$ invariant points $M$, $M'$ and $M''$ related by $C_3$. We can see from the table that the sign of $t$ controls the symmetry irreps at $\Gamma$ and $M$ while the sign of $\lambda$ controls the irreps at $K$.

\begin{center}
\begin{table}[t]
\caption{Symmetry irreps at the high symmetry momenta for different signs of the model parameters $t$ and $\lambda$.}
\bgroup
\setlength{\tabcolsep}{0.25 em}
\setlength\extrarowheight{0.25em}
\begin{tabular}{c|c|c|c}
\hline \hline
$(\sgn(t),\sgn(\lambda))$ & $\Gamma$ & $K$ & $M$\\
\hline 
$(+,-)$  & $e^{\pm i \pi/3}$ & $e^{\pm 2i \pi/3}$ & $1,1$\\
$(-,-)$  & $e^{\pm 2i \pi/3}$ & $e^{\pm 2i \pi/3}$ & $-1,-1$ \\
$(+,+)$  & $e^{\pm i \pi/3}$ & $1,1$ & $1,1$ \\
$(-,+)$  & $e^{\pm 2i \pi/3}$ & $1,1$ & $-1,-1$ \\
\hline \hline
\end{tabular}
\egroup
\label{Irreps}
\end{table}
\end{center}

Since our model has a filling of two per unit cell, there are only two possible atomic insulators which can be equivalent to it: the first one is built by taking one orbital at position $b$ (which has a multiplicity of 2 leading to a filling of two per unit cell) and the second by taking two orbitals at position $a$. In the first case, the little group $g_b$ is $C_3$ and the orbitals can have angular momentum $l_b = 0,1,2 \Mod 3$. It can be easily shown that regardless of the value of $l$, the symmetry irreps at the $M$ point are $\pm 1$ which obviously does not match with the irreps of our model given in Table~\ref{Irreps} for any value of $t$ and $\lambda$. For the second case, the little group for orbitals at position $a$ is $C_6$ and the angular momenta of the orbitals take the values $l_{a} = 0,\dots, 5 \Mod 6$. In order to match the irreps of the model at the $\Gamma$ point, the angular momenta $l_a$ of the two orbitals should be $\pm 1$ for positive $t$ and $\pm 2$ for negative $t$. In the former case, the resulting irreps at $M$ are $-1,-1$ while in the latter case, they are $1,1$ which both do not match the irreps of the model shown in Table~\ref{Irreps}.

\subsection{Is the obstruction stable or fragile?}
We now ask the question whether the obstruction we found in the previous section is stable or fragile. Recall that a Wannier obstruction is fragile if it can be removed by adding some trivial (atomic) degrees of freedom \cite{Po17fragile}. This means that the model can be smoothly deformed to a linear superposition of atomic insulators with some coefficients being negative. We will show that this is indeed the case for our shift insulator model. Before delving into technical details, let us first try to see why we expect this to be the case based on what is already known. The symmetry indicators, which indicate that the representation content of a given band structure differ from any superposition of atomic insulators (which includes fragile phases), were computed in Ref.~\onlinecite{Po17} for classes AI and AII. Later on, it was shown that all the phases corresponding to non-trivial indicators in class AI are semimetals \cite{Fang17semimetals}. This means that stable Wannier obstruction for insulators in class AI cannot be encoded in symmetry representation. Hence, it is possible to add some atomic orbitals to our model so that the representations at high symmetry momenta correspond to those of an atomic insulator. In general, this is not enough to show that there is no Wannier obstruction and it could be the case that the obstruction persists even if we add some atomic orbitals to ``fix'' the representations at high symmetry momenta \footnote{For example, an $n$-fold symmetric Chern insulator with Chern number $n$ has a stable Wannier obstruction which is not encoded in its representation content which is identical to the same representation content of the trivial insulator \cite{Fang12}.}. In the following, we will show that this is not the case and the obstruction can be removed upon the addition of some atomic orbitals. Interestingly, however, we will find that the number and type of orbitals we need to add depends on the signs of the model parameters $t$ and $\lambda$.

In the following, we will use the projection method \cite{Soluyanov11, Marzari12} explained in the previous section to show that we can construct Wannier states once we add a properly chosen set of atomic orbitals. The criterion for choosing this set is that, when combined to the symmetry irreps of the model, the resulting irreps should correspond to yet another set of atomic orbitalts. In other words, the algorithm for establishing the fragility of the obstruction relies on first choosing $n$ atomic orbitals to add to the model then trying to guess a basis of $n+2$ atomic orbitals which reproduces the same symmetry irreps as the model plus the extra orbitals. Once this basis is found, we can use it as a trial basis for the projection procedure to construct Wannier states. 

We start by considering the case $(\sgn(t),\sgn(\lambda)) = (+,-)$ and make the observation that an atomic insulator whose orbitals are localized at the $c$ position (3 orbitals per unit cell) and have angular momentum $l_c=1$ has the symmetry irreps $(-1,e^{\pm i \pi/3})$ at $\Gamma$, $(1,e^{\pm 2i \pi/3})$ at $K$, and $(-1,1,1)$ at $M$ which differ from the irreps of our model by the irreps $-1$, $1$, and $-1$ at the $\Gamma$, $K$, and $M$ points respectively. This can be easily reproduced by adding an atomic orbital at the hexagon center (position $a$) with angular momentum $l_a=3$. Thus, upon adding such an atomic orbital, we expect our model to become Wannier representable. 

The trial basis should be chosen such that the charge center lies on the hexagon edges ($c$ position). This means that the total weight of the ($p_\pm$) orbitals on the A and B sublattice sites lying at the end points of such an edge has to be the same. In addition, the weight of the orbitals on the $a$ position (hexagon center) for the two unit cells sharing such an edge should also be the same. What remains is to choose the relative weight of these two and the phases of the orbitals to produce the desired symmetry content in momentum space, which can be done by choosing the basis
\begin{gather}
|\tau'_{1,\bk} \rangle = (e^{i \bk \cdot \bt_A}, e^{i \bk \cdot \bt_A}, e^{i \bk \cdot \bt_B}, e^{i \bk \cdot \bt_B}, 1 + e^{i \bk \cdot (\bt_A + \bt_B)})^T, \nonumber \\ |\tau'_{2,\bk} \rangle =  \left(\begin{array}{cc} U_6 & 0 \\ 0 & -1 \end{array} \right) |\tau'_{1,O_6\bk} \rangle,  \nonumber \\ 
|\tau'_{3,\bk} \rangle = \left(\begin{array}{cc} U_6 & 0 \\ 0 & -1 \end{array} \right) |\tau'_{2,O_6\bk} \rangle.
\label{tpp}
\end{gather}
We notice that, like the eigenstates of the original Hamiltonian $h_\bk$, the states $|\tau'_\bk \rangle$ are not periodic in $\bk$ but satisfy instead $|\tau'_{\bk + \bG} \rangle = \tilde V_\bG |\tau'_{\bk} \rangle$. Here, $\tilde V_\bG$ is defined by extending $V_\bG$ defined in (\ref{VG}) to include an extra orbital at the center of the unit cell. It is given explicitly by
\beq
\tilde V_\bG = \left( \begin{array}{cc} V_\bG & 0 \\ 0 & 1 \end{array} \right).
\eeq
To obtain the Wannier states using (\ref{W}), we perform a transformation similar to (\ref{VhV}) to get the periodic basis states $|\tau_{\bk} \rangle = \tilde V_\bk^\dagger |\tau'_{\bk} \rangle$. The overlap matrix is then given by $S_{mn}(\bk) = \langle \tau_{m,\bk}| P_\bk | \tau_{n,\bk} \rangle = \langle \tau'_{m,\bk}| V_\bk P_\bk V^\dagger_\bk| \tau'_{n,\bk} \rangle$ which can be written in terms of the projection operator of the original Hamiltonian as
\beq
S_{mn}(\bk) = \langle \tau'_{m,\bk}| p_\bk | \tau'_{n,\bk} \rangle, \quad p_\bk = \frac{1}{2}(1 - h_\bk / \sqrt{h_\bk^2}).
\eeq
This implies that the overlap matrix does not change under the unitary transformation $\tilde V_\bk$ as expected.

The determinant of the overlap matrix $\det S(\bk)$ is shown in Fig.~\ref{DS} and we can clearly see that it is non-vanishing everywhere, thereby establishing that there is no obstruction to constructing Wannier states from the basis states $|\tau_{n,\bk} \rangle$. The resulting symmetric Wannier states are shown in Fig.~\ref{Wannier1} and we can see that the charge center lies at the $c$ position and that the three Wannier states are manifestly related by rotation. 

\begin{figure}[h]
\center
\includegraphics[width=0.8\columnwidth]{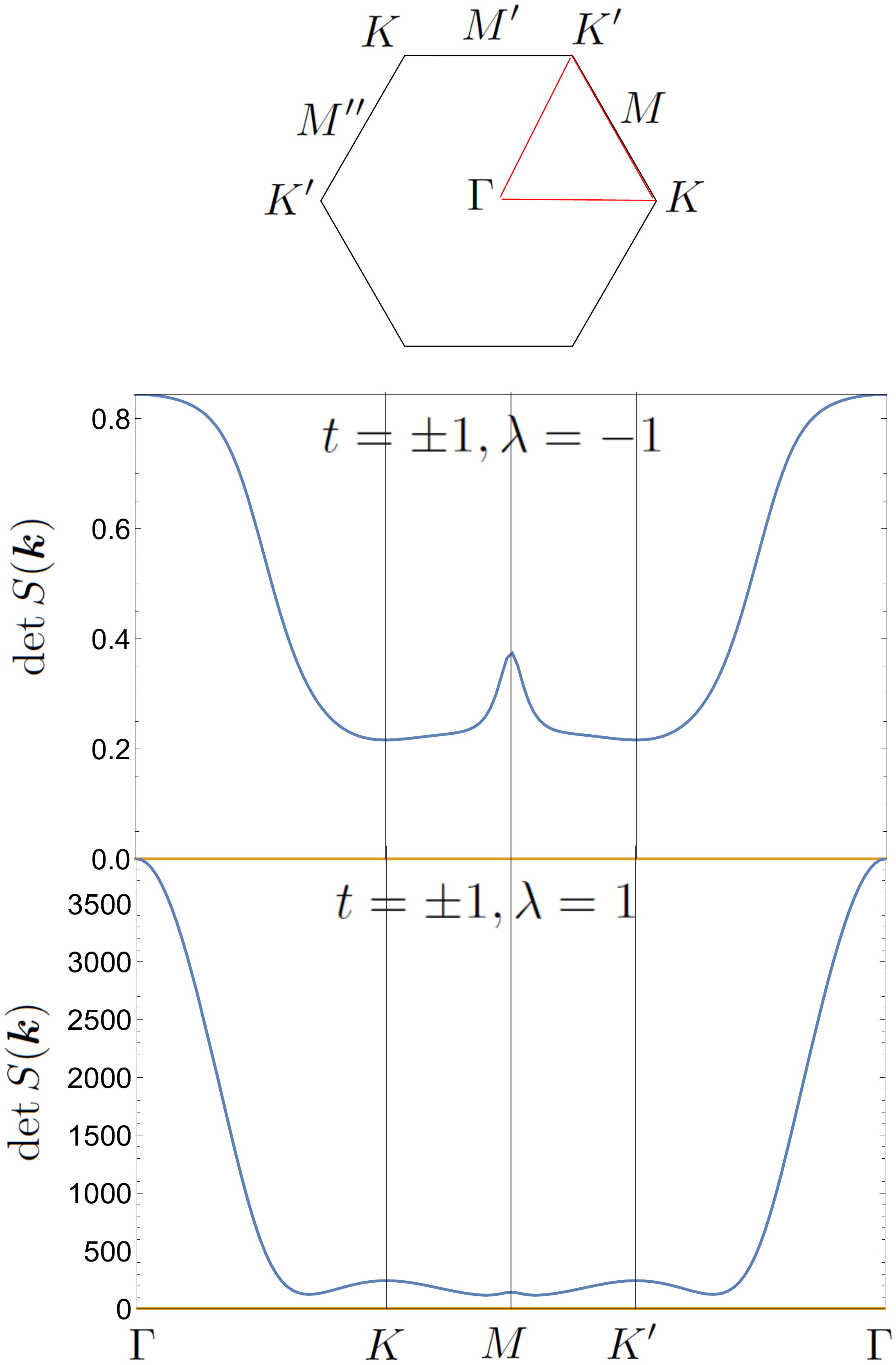}
\caption{The determinant of the overlap matrix $S(\bk)$ for the basis choices given in Eqs.~\ref{tpp}, \ref{tmp}, \ref{tpm} and \ref{tmm} (Note that the curves for $t=\pm 1$ are identical).}
\label{DS}
\end{figure}

 \begin{figure*}[t]
\center
\includegraphics[height=0.4\columnwidth]{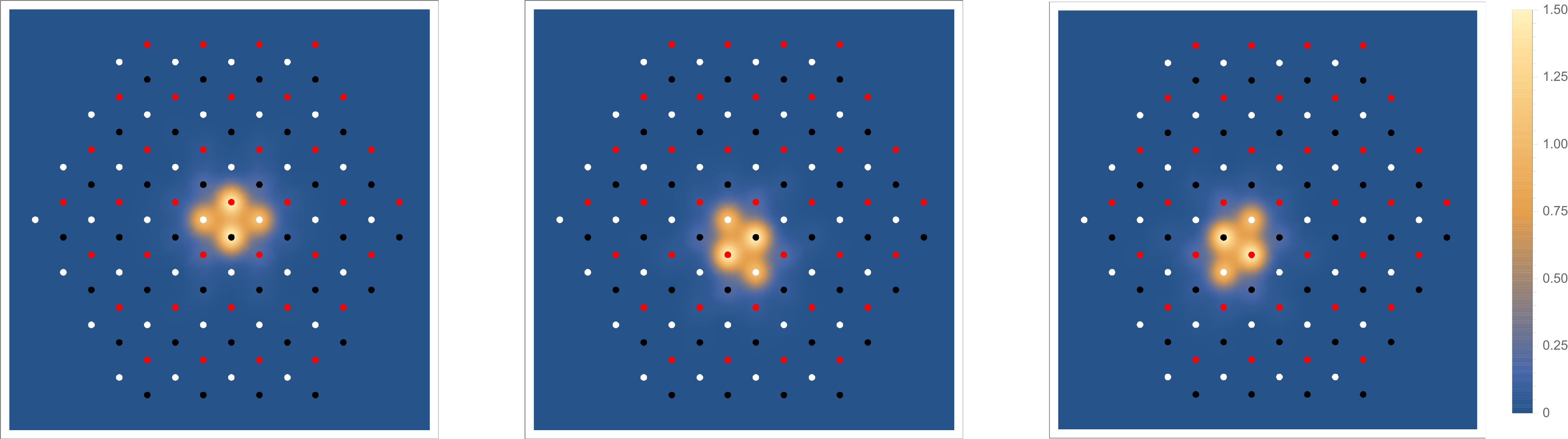}
\caption{The weight of the Wannier function $|W_{n,0}(\br)|^2$ at different sites for the shift insulator model with $\lambda<0$ after the addition of an atomic orbital at the center of the unit cell ($a$ position). The white, red and black dots denote the center of the unit cell, A and B sublattices respectively. The orbital character of the orbitals is not shown and the weight of each position is multiplied by a Gaussian function centered around the corresponding position for clarity. We can see that the three Wannier states are centered at the hexagon edges ($c$ position) and are related to each other by rotation.}
\label{Wannier1}
\end{figure*}

We next consider the case  $(\sgn(t),\sgn(\lambda)) = (-,-)$ which is very similar to the $(+,-)$ case with the main difference being that the orbital which should be added at the $a$ position has angular momentum $l_a=0$. The resulting irreps are $(1,e^{\pm 2i \pi/3})$ at $\Gamma$, $(1,e^{2\pm i \pi/3})$ at $K$, and $(1,-1,-1)$ at $M$ which correspond to an atomic insulator with three orbitals at the $c$ position with angular momentum $l_c=0$. We can then write the basis $|\tau'_{n,\bk} \rangle$ similar to (\ref{tpp}) as
\begin{gather}
|\tau'_{1,\bk} \rangle = (e^{i \bk \cdot \bt_A}, e^{i \bk \cdot \bt_A}, -e^{i \bk \cdot \bt_B}, -e^{i \bk \cdot \bt_B}, 1 + e^{i \bk \cdot (\bt_A + \bt_B)})^T, \nonumber \\ |\tau'_{2,\bk} \rangle = \left(\begin{array}{cc} U_6 & 0 \\ 0 & 1 \end{array} \right) |\tau'_{1,O_6\bk} \rangle,  \nonumber \\  
|\tau'_{3,\bk} \rangle =  \left(\begin{array}{cc} U_6 & 0 \\ 0 & 1 \end{array} \right) |\tau'_{2,O_6\bk} \rangle.
\label{tmp}
\end{gather}
The determinant of the resulting overlap matrix as well as weights of the Wannier states on different sites turns out to be the same as in the $(+,+)$ given in Fig.~\ref{DS} and \ref{Wannier1} respectively (note that the orbital character, which differs in both cases, is not shown in these plots).

The case of $(\sgn(t),\sgn(\lambda)) = (+,+)$ is different since we need to add two, rather than one, atomic orbitals to the model to be able to find a Wannier representation. The two atomic orbitals are placed at the A and B sublattices (Wyckoff position $b$) and correspond to angular momentum $l_b=0$. They yield the irreps $\pm 1$ at $\Gamma$, $e^{2\pm i \pi/3}$ at $K$, and $\pm 1$ at $M$ which when combined with the irreps of our model correspond to an atomic insulator with one $l_a=0$ orbital at position $a$ and three $l_c=1$ orbitals at position $c$. The corresponding basis states are
\begin{gather}
|\tau'_{1,\bk} \rangle = (e^{i \bk \cdot \bt_A}, e^{i \bk \cdot \bt_A}, e^{i \bk \cdot \bt_B}, e^{i \bk \cdot \bt_B}, e^{i \bk \cdot \bt_A},-e^{i \bk \cdot \bt_B})^T, \nonumber \\ 
|\tau'_{2,\bk} \rangle =\left(\begin{array}{ccc} U_6 & 0 & 0 \\ 0 & 0 & 1 \\ 0 & 1 & 0 \end{array} \right) |\tau'_{1,O_6\bk} \rangle,  \nonumber \\  
|\tau'_{3,\bk} \rangle = \left(\begin{array}{ccc} U_6 & 0 & 0 \\ 0 & 0 & 1 \\ 0 & 1 & 0 \end{array} \right) |\tau'_{2,O_6\bk} \rangle, \nonumber \\
|\tau'_{4,\bk} \rangle = (0, e^{i \pi/3} e^{i \bk \cdot \bt_A} - e^{-i \bk \cdot \bt_B} + e^{-i \pi/3} e^{ i \bk \cdot (-\bt_A + \bt_B)}, \nonumber \\ \qquad \qquad 0, e^{i \bk \cdot \bt_B} + e^{-2i \pi/3} e^{-i \bk \cdot \bt_A} + e^{2i \pi/3} e^{ i \bk \cdot (\bt_A - \bt_B)},\nonumber \\ \qquad e^{i \bk \cdot \bt_A} + e^{-i \bk \cdot \bt_B} + e^{ i \bk \cdot (-\bt_A + \bt_B)}, \nonumber \\ \qquad  e^{i \bk \cdot \bt_B} + e^{-i \bk \cdot \bt_A} + e^{ i \bk \cdot (\bt_A - \bt_B)})^T.
\label{tpm}
\end{gather}
The determinant of the resulting overlap matrix is shown in Fig.~\ref{DS} and the corresponding Wannier states are shown in Fig.~\ref{Wannier2}, where we can see three rotation-related states centered at position $c$ and one rotationally symmetric state centered at $a$.

 \begin{figure*}[t]
\center
\includegraphics[height=0.4\columnwidth]{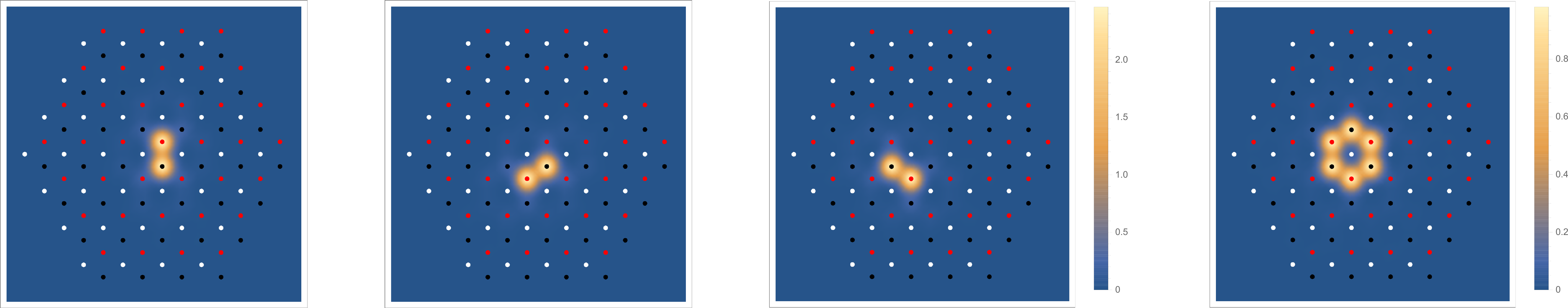}
\caption{The Wannier functions $|W_{n,0}(\br)|^2$ for the shift insulator model with $\lambda>0$ after the addition of two atomic orbitals at the A and B sublattices ($b$ position). As in Fig.~\ref{Wannier1}, the white, red and black dots denote the center of the unit cell, A and B sublattices, respectively and we use a Gaussian function centered around the corresponding position for clarity. Here, we have three rotation-related states centered at the hexagon edges ($c$ position) and one rotationally symmetric state centered at the hexagon center ($a$ position).}
\label{Wannier2}
\end{figure*}

 For $(-,+)$, the situation is quite similar to the $(+,+)$ case. Here, we also have to add two atomic orbitals with angular momentum $l=0$ at position $b$ to be able to find a Wannier representation. The resulting atomic insulator is, however, slightly different from the $(+,+)$ case and it corresponds to a single $l_a=3$ orbital at position $a$ and three $l_c=0$ orbitals at position $c$. The basis states are 
\begin{gather}
|\tau'_{1,\bk} \rangle = (e^{i \bk \cdot \bt_A}, e^{i \bk \cdot \bt_A}, -e^{i \bk \cdot \bt_B}, -e^{i \bk \cdot \bt_B}, e^{i \bk \cdot \bt_A}, e^{i \bk \cdot \bt_B})^T, \nonumber \\ 
|\tau'_{2,\bk} \rangle = \left(\begin{array}{ccc} U_6 & 0 & 0 \\ 0 & 0 & 1 \\ 0 & 1 & 0 \end{array} \right) |\tau'_{1,O_6\bk} \rangle,  \nonumber \\ 
|\tau'_{3,\bk} \rangle =  \left(\begin{array}{ccc} U_6 & 0 & 0 \\ 0 & 0 & 1 \\ 0 & 1 & 0 \end{array} \right) |\tau'_{2,O_6\bk} \rangle, \nonumber \\
|\tau'_{4,\bk} \rangle = (0, e^{i \pi/3} e^{i \bk \cdot \bt_A} - e^{-i \bk \cdot \bt_B} + e^{-i \pi/3} e^{ i \bk \cdot (-\bt_A + \bt_B)}, \nonumber \\ \qquad \qquad 0 ,-e^{i \bk \cdot \bt_B} - e^{-2i \pi/3} e^{-i \bk \cdot \bt_A} - e^{2i \pi/3} e^{ i \bk \cdot (\bt_A - \bt_B)},\nonumber \\ \qquad \qquad e^{i \bk \cdot \bt_A} + e^{-i \bk \cdot \bt_B} + e^{ i \bk \cdot (-\bt_A + \bt_B)}, \nonumber \\ \qquad -e^{i \bk \cdot \bt_B} - e^{-i \bk \cdot \bt_A} - e^{ i \bk \cdot (\bt_A - \bt_B)})^T,
\label{tmm}
\end{gather}
The determinant of the resulting overlap matrix and the Wannier states are the same as in the $(+,+)$ case.

The expression of the shift insulator in terms of atomic orbitals is provided in Table~\ref{AtomicRepresentations} for different values of $t$ and $\lambda$. For comparison, we also consider the shift insulator model built using $p_z$ orbitals ($L_z = 0$) rather than $p_\pm$ orbitals.

It is worth noting that our model bears some similarity to inversion-protected TCIs studied in Ref.~\cite{Alexandradinata14}. Such models are characterized by Wilson loop winding which can be used to show they are built by stacking Chern insulators with opposite Chern numbers. Recently, it was shown that this Wilson loop winding is a signature of fragile topology of these models \cite{Alexandradinata18}. This is consistent with our conclusions which are obtained using very different methods.

\subsection{Topological response from the atomic representation}
\label{TopologicalAtomic}
\begin{center}
\begin{table*}[t]
\caption{Realization of the model for different values of $L_z$, $t$ and $\lambda$ in terms of atomic insulators. An atomic insulator is indicated by $x_l$, where $x$ is the Wyckoff position $a$, $b$, or $c$ (shown in Fig.~\ref{Wyckoff}) and $l$ is the angular momentum. The charge trapped in a disclination $\Delta Q$ and the angular momentum response to the insertion of a flux of $m$ are computed using the expressions: $\Delta Q = n_\Omega [-\sgn(\lambda) L_z/3 + \sgn(t) /2]$ and $\Delta s = m[3 \sgn t - 2 L_z \sgn \lambda]$ derived in Sec.~\ref{Response}. These results are reproduced by the atomic representation by by assigning a disclination charge of $n_\Omega/6$ and flux angular momentum of $m$ for the atomic insulator in Wyckoff position $a$ (and 0 charge and angular momentum for positions $b$ and $c$). The last column indicates the interacting invariant discussed in Sec.~\ref{Interaction}.}
\bgroup
\setlength{\tabcolsep}{0.5 em}
\setlength\extrarowheight{0.4em}
\begin{tabular}{c|c|c|c|c}
\hline \hline
$(L_z,\sgn(t),\sgn(\lambda))$ & Atomic representation & $\Delta Q \!\! \mod 1$ & $\Delta s(m) \!\! \mod 6$ & $\nu \in \Z_{12} \times \Z_3 \times \Z_4$\\
\hline 
$(1,+,-)$  & $c_{1} - a_{3}$ & $5n_\Omega/6 \sim -n_\Omega/6$ & $5m \sim -m$ & $(5,0,1)$\\
$(1,-,-)$  & $c_{0} - a_{0}$ & $-n_\Omega/6$ & $-m$ & $(-1,0,1)$\\
$(1,+,+)$  & $c_{1} + a_{0} - b_{0}$ & $n_\Omega/6 $ & $m$ & $(1,-1,-1)$\\
$(1,-,+)$  & $c_{0} + a_{3} - b_{0}$ & $-5 n_\Omega/6 \sim n_\Omega/6$ & $-5m \sim m$ & $(-5,-1,1)$\\
$(0,+,\pm)$  & $c_{0} + a_{1} + a_0 + a_{-1} - b_{1} - b_{-1}$ & $n_\Omega/2 \sim -n_\Omega/2$ & $3m \sim -3m$ & (3,-2,1)\\
$(0,-,\pm)$  &  $c_{1} + a_{2} + a_3 + a_{-2} - b_{1} - b_{-1}$ & $-n_\Omega/2$ & $-3m$ & $(-3,-2,-1)$\\
\hline \hline
\end{tabular}
\egroup
\label{AtomicRepresentations}
\end{table*}
\end{center}

We now show how the topological response obtained in Sec.~\ref{Response} can be derived from the atomic description we just obtained. A fragile phase can be written as a superposition of atomic insulators with integer coefficients. Each atomic insulator is specified by a Wyckoff position $x$ and a representation of the site symmetry group of this position. For the wallpaper group $p6$ relevant to the shift insulator, the representations are labeled by integers $l=0,\dots, N_x$ where $C_{N_x}$ is the symmetry group of the Wyckoff position $x$. $N_x$ is given explicitly by  $6,3,2$ for positions $x=a,b,c$ respectively. Any fragile phase can then by written as
\beq
\chi = \sum_{x=a,b,c} \sum_{l=0}^{N_x-1} c_{x,l} x_l.
\label{AExp}
\eeq 
We note here that the general Wyckoff position $d$ is not considered since we can always symmetrically bring the orbitals in position $d$ to any of the other three positions. We also note that the representation (\ref{AExp}) is generally not unique due to the fact that any set of $N_x$ orbitals at Wyckoff position $x$ with all possible values of angular momenta $l=0,\dots,N_x-1$ can be symmetrically deformed away using the general position $d$ and brought to any of the other Wyckoff positions.

To understand the response to a disclination, we notice that any atomic orbital far away from the $C_6$ rotation center comes with a group of six orbitals which transform into each other under the rotation symmetry. If we cut off a $60\degree$ wedge, these six orbitals can be smoothly glued back into five orbitals. On the other hand, for orbitals located near the rotation center, we do not know what would happen in general: some orbitals cannot be glued back after we cut a wedge, while some new orbitals may appear. Since we do not care about integer charges, we can imagine first symmetrically unoccupy all the orbitals near the rotation center, then the net number of electrons (after subtracting off the background) near the rotation center becomes $-n_a+6k$ where $k\in\mathbb{Z}$ and $n_a$ is given by
\beq
n_a = \sum_{l=0}^5 c_{a,l}.
\eeq
Now if we cut off a wedge, the net number of electrons near the pentagon disclination ($n_\Omega=1$) is $n_a/6$ up to an integer, and this is nothing but the disclination charge we are looking for. More generally, the $(6-n_\Omega)$-gon disclination charge is given by $(n_\Omega/6)n_a$. We see from Table~\ref{AtomicRepresentations} that this precisely corresponds to the charge obtained in Sec.~\ref{Disclination}. We also notice that $n_a$ is only defined modulo 6 due to the aforementioned ambiguity in the representation (\ref{AExp}), which is consistent with the fact that the disclination charge is defined modulo an integer.

The response to the threading of a monopole flux can be understood analogously by noting that only the $a$ position orbitals at the rotation center are affected by the threading of a flux quantum, assuming for simplicity that all orbitals have tiny spread. The threading of $m$ flux quanta changes the orbital wavefunctions as
\beq
\psi(r,\varphi) \rightarrow \psi(r,\varphi) e^{i m \varphi}.
\eeq
This is associated with a change of $m$ in angular momentum. Thus, the resulting total change in angular momentum of the ground state is again determined by the number of atomic orbitals at position $a$ and given by $m n_a$. From this, we can see that topological response for the shift insulator model is captured by the Chern-Simons theory (\ref{LCS}) with $S$ given by $n_a$ modulo 6.

Following the discussion of this section, it is worth noting that the response of the shift insulator given in Sec.~\ref{Response} and summarized in Table~\ref{AtomicRepresentations} is indeed a signature of fragile topology. The reason is that the only two possibilities for an atomic insulator with two filled bands in the symmetry group $p6$ corresponds to either two orbitals at position $a$ for which $n_a = 2$ or one orbital at position $b$ for which $n_a = 0$. The shift insulator corresponds to $n_a = \pm 1 \Mod 6$ (for $L_z = 1$) or $n_a = 3 \Mod 6$ (for $L_z = 0$) which is inconsistent with either atomic insulator, thus implying that the charge response to disclination or angular momentum response to flux threading can indeed be used to detect the fragile topology of our model.

\subsection{Discussion of Topological Responses}
\label{DTR}
The robust topological responses of the shift insulator raises the following puzzle: how is this seemingly absolute attribute reconciled with the observation that the topology of atomic insulators is only relative, and the shift insulator can be reduced to a collection of atomic insulators (with potentially negative coefficients)? 

A simpler setting where this question appears is in inversion symmetric 1D spinless insulators \cite{Turner12}, where two insulators with relative polarization $\pi$ are allowed. These correspond to the charge densities centered on one or the other inversion center in real space. A priori it is not clear which one is to be assigned zero polarization. This ambiguity is settled by specifying additional information beyond the band structure, the location of the ionic cores (or sites of a tight binding model) which ensure charge neutrality, required to define polarization. Zero polarization  is then assigned to the band structure in which the electronic charge center coincides with the atomic cores (or tight binding sites). 

In a similar way, our tight binding sites of the shift insulator provide a fixed reference state relative to which the atomic insulators are compared. When the charge centers of the band structure coincide with the  tight binding sites, the responses we define are trivial - for example, a disclination will not carry fractional charge. Thus, effectively a relative topology has been extracted from these responses.

\section{TCIs protected by point group symmetries: general discussion}
\label{Layer}
In Secs.~\ref{Model}, \ref{Response}, and \ref{Wannier}, we considered a specific model for a TCI protected by rotation symmetry which does not possess any surface states. We have shown that, despite the apparently non-trivial features of the model, it can be described in terms of a superposition of atomic orbitals provided we allow for negative coefficients, making it an example of fragile topology. In this section, we consider the general problem of TCIs protected by point group symmetries. We will show that those TCIs without surface states all share the features of the shift insulators in that their topology is at most fragile. This is established by showing that they can all be understood in terms of symmetrically repeating a 0D unit in a manner similar to the layer construction of Refs.~\onlinecite{Hermele17, Huang17,Fulga16}. We begin the section by a general review of the layer construction before presenting the proof that every point-group-protected TCI without surface states is either an atomic insulator or a fragile phase.

\subsection{Review of the layer construction}

The layer construction was used in Refs.~\onlinecite{Hermele17, Huang17} to reduce the problem of classifying symmetry protected topological phases (SPTs) protected by spatial symmetries to that of classifying SPTs protected by internal symmetries in lower dimensions. For simplicity, we will restrict our discussion here to TCIs protected by point group symmetries. It should be noted, however, that the validity of the layer construction is not limited to point groups and it was employed successfully to understand TCIs protected by other spatial symmetries such as glides, screws or translations \cite{Ezawa16, SongFang17}. 

The layer construction relies on the observation that a spatial point group symmetry $g$ in $d$ dimensions generally divides the space into several $d$-dimensional symmetry-related regions separated by a $(d-1)$-dimensional symmetry-invariant region. Since the symmetry $g$ does not act within any of these $d$-dimensional regions, the phase can be trivialized there (since an SPT can always be trivialized when the symmetry is broken). This means that the non-trivial topology of the phase has to be encoded in the $(d-1)$-dimensional symmetry-invariant region where the symmetry $g$ acts either as an internal symmetry leaving every point invariant (e.g. mirror symmetry acting within the mirror plane) or as a spatial symmetry mapping different points to each other (e.g. inversion symmetry acting within a plane containing the inversion center). In the latter case, the symmetry-invariant region can either host an internal-symmetry-protected TI, which is just compatible with the spatial symmetry\footnote{The reason we need to check for compatibility is that sometimes spatial symmetries may impose some constraints on the strong topological invariants and rule out some strong phases.}, or a TCI in $(d-1)$-dimensions for which the dimensional reduction can be implemented further. Repeating this procedure is guaranteed to reduce the problem of classifying TCIs to that of classifying TI protected by internal symmetries in a lower dimension. Once this is done, the internal symmetries can be used to transform the Hamiltonian into a block-diagonal form where each block belongs to one of the 10 Altland-Zinrbauer classes (including only the fundamental symmetries such as time-reversal or $U(1)$ charge conservation) whose topological classification is well-understood \cite{Kitaev09, Schnyder09, Ryu10}.

To understand how this works in more concrete terms, let us consider TCIs in class AII protected by mirror or inversion symmetry in 3D. For the case of mirror, the 3D space is split in two regions by the mirror plane on which mirror symmetry acts as an internal $\Z_2$ symmetry. This symmetry can be used to transform the Hamiltonian into a block-diagonal form where the two blocks are related by time-reversal symmetry. Each block separately may break time reversal symmetry (class A), hence, we can assign to it a non-vanishing Chern number $C$ corresponding to a quantum Hall or Chern insulator phase (the total Chern number of the two blocks is still zero). The resulting TCI is characterized by mirror Chern number $C$ in the corresponding mirror plane.

For inversion, the 3D space can be split by any plane containing the inversion center. Such a plane, on which inversion acts as two-fold rotation, can host either a 2D time-reversal-invariant TI (which is compatible with two-fold rotation), or a 2D TCI protected by twofold rotation. The former corresponds to a second-order TI \cite{Fang17, Khalaf17, Khalaf18} whereas the latter can be understood by splitting the 2D plane using a line which contains the inversion center and investigating possible phases on this line. Due to the absence of TIs (protected by internal symmetries) in 1D, the only possible topological phase on this line is an inversion-protected TCI which can be similarly understood by dividing it in two pieces around the inversion center. By placing a 0D TI (which is just an atomic insulator with a given number of filled bands) at the inversion center, we can obtain a non-trivial TCI (since the classification of TIs protected by $U(1)$ charge conservation in 0D (class A, AI, and AII) is $\Z$ \cite{Schnyder09, Kitaev09, Ryu10}). The 1D TCI obtained is nothing but the inversion-protected Su-Schrieffer-Heeger (SSH) chain and the corresponding 3D TCI is built by repeating the SSH chain symmetrically to fill the 3D space.

The previous dimensional reduction argument provides a map between TCIs in a given dimension and TIs protected by internal symmetries in lower dimensions. Such map provides a hierarchy of TCIs in a given dimension and symmetry class which distinguishes them according to the dimension of the TI used in their construction. Since TIs protected by internal symmetries in $d$ dimensions always exhibit $(d-1)$-dimensional surface states, such hierarchy of TCIs distinguishes them according to the dimensionality of their surface states. More specifically, a TCI which maps via the dimensional reduction procedure to a $d$-dimensional TI will host gapless $(d-1)$-dimensional surface states on a generic symmetry-compatible surface\footnote{Following Ref.~\onlinecite{Trifunovic18}, we call a surface symmetry-compatible if it does not have any symmetry-invariant face but preserves the symmetry as a whole}. Such hierarchy was explored in a recent work by Trifunovic and Brouwer \cite{Trifunovic18} which introduced the sequence of subgroups $K^{(d)} \subseteq K^{(d-1)} \dots K^{(1)} \subseteq K^{(0)}$. Here, $K^{(n)}$ denotes the group of phases for which any symmetry-compatible surface is gapped except for a region whose dimension is at most $d-n-1$ and $K^{(0)} = K$ is the standard K-group \cite{Kitaev09, Morimoto13, Chiu16, Lu14, Shiozaki14}. The set of phases with $(d-n-1)$-dimensional surface states on a general symmetry-compatible surface is then given by $K^{(n)}/K^{(n+1)}$ which correspond in our discussion to TCIs built by layering $(d-n)$-dimensional TIs.

The dimensional reduction map described above can also be reversed to obtain a TCI from a lower-dimensional symmetry-compatible TI \cite{Ezawa16, Isobe15, Hermele17, Huang17,Fulga16}. To do this, we start with a lower-dimensional TI, for example a quantum Hall layer, and repeatedly adjoin it with a 2D layer and its copies under the symmetry until the whole 3D space is filled. This procedure, however, does not provide a one-to-one correspondence between TCIs and lower-dimensional TIs due to the fact that two lower-dimensional TIs related by the adjoining operation (adding a ``layer'' and its copies under the symmetry) lead to identical higher-dimensional TCIs. To obtain a one-to-one correspondence, lower-dimensional phases related by the adjoining procedure should be identified \cite{Hermele17, Huang17}.

It is instructive to consider an illustrative example of the hierarchy of TCIs explained above. Consider TCIs without time-reversal symmetry (class A) protected by inversion in odd dimensions $d=2m+1$ whose classification is  $\Z$ \cite{Lu14, Shiozaki14}. A simple analysis shows that the odd elements in $\Z$ correspond to second-order TIs built by layering a quantum Hall system in $d=2m$ dimensions which hosts $(2m-1)$-dimensional surface states. If we add two such phases, there are two possible inversion-odd mass terms which can be used to gap out these surface states leaving a $(2m-3)$-dimensional gapless region. The resulting phase can then by identified with TCIs built by layering quantum Hall states in $d=2m-2$ dimensions. This pattern continues until we get the TCIs without any surface states corresponding to elements which are multiples of $2^m \in \Z$. Such hierarchy is summarized by the equation below
\beq
n \in \Z = \begin{cases}
1 \mod 2, &: (2m-1) \text{D surface states},\\
2 \mod 4, &: (2m-3) \text{D surface states},\\
\dots & \dots \\
2^{m-1} \mod 2^m, &: 1\text{D surface states},\\
0 \mod 2^m, &: \text{no surface states},
\end{cases}
\label{Hierarchy}
\eeq
and it corresponds to the sequence of subgroups given by $K^{(2l)} = K^{(2l+1)} = 2^l \Z$ in the language of Ref.~\onlinecite{Trifunovic18}.

\subsection{TCIs without surfaces states}
Having explained the general dimensional mapping, let us now focus on the subgroup of point-group-protected TCIs which do not have any surface states. The dimensional reduction mapping can be used to reduce the study of these TCIs in $d$ dimensions to TIs protected by internal symmetries in a lower dimension $\delta < d$. Since the latter always posses surface states with dimension $\delta - 1$, the absence of surface states implies that $\delta = 0$. In other words, TCIs without surface states are mapped to 0D TI protected by internal symmetries which are necessarily atomic insulators. This, however, does not imply that these phases admit a Wannier representation. In fact, we know that fragile phases do not have any surface states while at the same time not admitting a Wannier representation. How can this then be reconciled with the statement that they are mapped to atomic insulators within the dimensional reduction procedure?

 \begin{figure}[t]
\center
\includegraphics[width=\columnwidth]{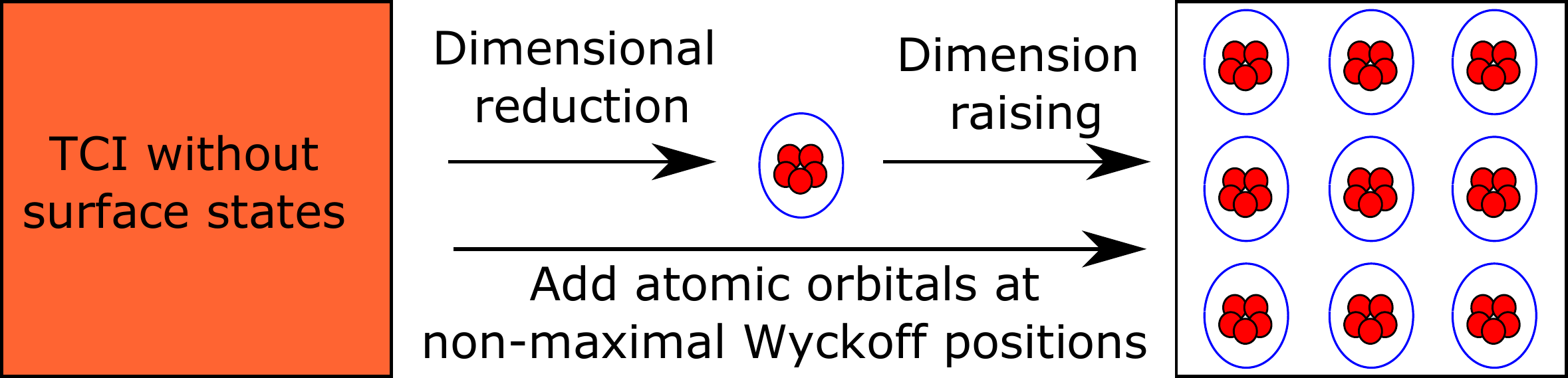}
\caption{Schematic illustration of the proof that all TCIs without surface states admit a Wannier representation possible after the addition of some atomic degrees of freedom (DOFs). Using the layer construction, we shown in the main text that any $d$-dimensional TCI without surface states protected by point group symmetries can be mapped to a Wannier representable phase by first using the dimensional reduction to a 0D atomic insulator then using the dimensional raising or layering procedure. The mapping involves the addition of some atomic orbitals at non-maximal Wyckoff positions.}
\label{Proof}
\end{figure}

To resolve this apparent paradox, we should note that the dimensional reduction procedure generally breaks translational symmetry since it singles out a symmetric lower-dimensional region. This implies that a TCI which reduces to an atomic insulator within the dimensional reduction procedure is not necessarily Wannier representable since a Wannier representation necessarily requires the existence of translational symmetry. Translational symmetry can, however, be restored by following the dimensional reduction procedure by a dimension-raising procedure i.e. by repeating the 0D unit cell in a manner that is consistent with both translation and the other symmetries to construct a higher-dimensional TCI. The resulting TCI is, however, not completely equivalent to the original TCI due to the many-to-one nature of the dimension-raising map. The two systems are only equivalent up to the adjoining procedure which adds some atomic orbitals off high-symmetry points, lines, or planes and their images under the symmetry. This implies that a TCI without surface states is equivalent under adjoining to a Wannier representable phase i.e. it has at most a fragile obstruction to a Wannier representation. Such procedure is illustrated schematically in Fig.~\ref{Proof}.

One aspect that is not entirely clear in the correspondence between TCIs without surface states and fragile phases is that the adjoining operation seems more restricted that the type of operation allowed to make a fragile phase Wannier representable. To be more precise, we are allowed to add {\it any} atomic insulator to make a fragile phase Wannier representable but the adjoining operation is restricted to adding atomic orbitals in non-maximal Wyckoff positions (a maximal Wyckoff position is one which cannot be deformed into a more symmetric position). However, this restriction does not lead to any extra obstructions since the general Wyckoff position (that consists of a point without any symmetry and its images under the symmetry), which is allowed by the adjoining procedure, contains all possible symmetry representations at all Wyckoff positions \cite{Po17}. 

More explicitly, imagine we have a fragile phase which can be written as $\sum_{x,\alpha_x} n_x^{\alpha_x} x^{\alpha_x}$ where $x^{\alpha_x}$ corresponds to an atomic insulator with orbital centered at the Wyckoff position $x$ in the symmetry irrep $\alpha_x$. The absence of an atomic representation would be reflected by the fact that some of the integers $n_x^{\alpha_x}$ in this expression are negative. If we now add $N$ orbitals in the general Wyckoff position, $n_x^{\alpha_x}$ increases by $N|G|/|G_x| d_x^{\alpha_x}$ \cite{Po17}, where $|G|$ is the order of the full symmetry group, $|G_x|$ is the order of the little group which leaves position $x$ invariant and $d_x^{\alpha_x}$ is the dimension of the irrep $\alpha_x$. Thus, for large enough $N$, we can always make all the coefficients $n_x^{\alpha_x}$ positive.

As an example, consider the shift insulator for $t<0$ and $\lambda<0$. As shown in Table~\ref{AtomicRepresentations}, this is a fragile TCI which is equivalent to the difference between an atomic insulator with $l=0$ orbitals at the $c$ position and one with $l=0$ orbitals at the $a$ position which we write as $\H_{--} \sim c_0 - a_0$. Under the dimensional reduction then raising, this model maps to an atomic insulator which differs by the addition of an atomic orbitals at a general position and its copies under sixfold rotation. At position $a$, this corresponds to adding an orbital with each angular momentum $l=0,\dots,5$. This can be written as $\H_{--} + \sum_{l=0}^{l=5} a_l \sim c_1 + \sum_{l=1}^{l=5} a_l$ which shows that the general Wyckoff position is enough to make the model Wannier representable.

\section{Classification of rotation protected TCIs in 2D}
\label{Rotation}
In this section, we use the layer construction to classify TCIs protected by rotation symmetry in 2D. Following the discussion of the previous section, 2D TCIs protected by $n$-fold rotation can be understood by dividing the 2D plane into $n$ symmetry-related patches separated by a symmetry-invariant region which in this case has the form of a collection of lines intersecting at the origin (see Fig.~\ref{LayerC6} for the case $n=6$). Due to the absence of strong TIs in one spatial dimension for electronic systems (one needs chiral or particle-hole symmetry for that), this 1D region cannot really host any strong TI and it can only host a rotation-protected TCI which is captured by reducing the problem to the (0D) rotation center. The problem then reduces to classifying 0D systems with $U(1)$ charge conservation protected by $\Z_n$ internal symmetry under the adjoining procedure which adds to the system an arbitrary atomic insulator and its copies under the symmetry. The reduction to 0D enables us to obtain the full classification of 2D TCIs protected by rotation as well as the interaction-induced reduction of such classification. The same results can also be used to obtain the classification of 3D TCIs without surface states protected by inversion or roto-inversion.

\begin{figure}[h]
\center
\includegraphics[width=0.55\columnwidth]{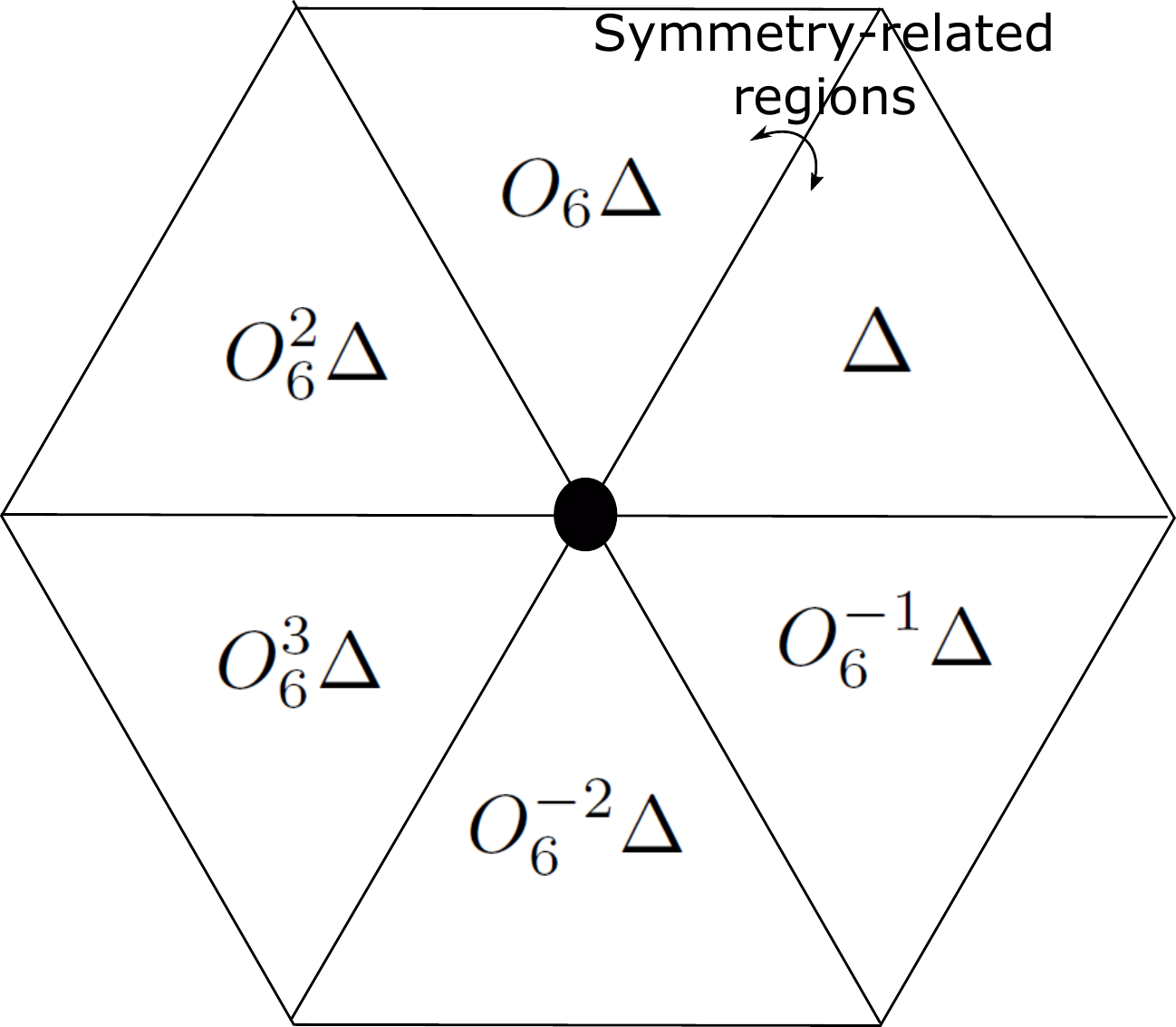}
\caption{Illustration of the dimensional reduction procedure for TCIs protected by 6-fold rotation symmetry. The 2D plane is divided into six symmetry-related regions separated by three lines at angle $\pi/3$ from each other. The three lines form a 1D symmetry-invariant region for which the construction can be repeated to reduce the problem to the 0D problem with internal $\Z_N$ symmetry at the rotation center.}
\label{LayerC6}
\end{figure}

We note that to obtain the classification of point-group protected phases in a given wallpaper or space group, we should in addition take into account the presence of translation which leads to the presence of several inequivalent high symmetry points within the unit cell. Since the layer construction explicitly breaks translation by singling out one of these high symmetry point, the full classification of point-group protected TCIs in the presence of translation requires repeating the dimensional reduction procedure for each of these high symmetry points.

\subsection{Non-interacting classification}

In the following, we will consider the possibilities of spinless, spinful or no time-reversal symmetry corresponding to symmetry classes AI, AII, and A, respectively. The rotation symmetry satisfies $C_N^N=1$ for the spinless case and $C_N^N =-1$ for the spinful case. However, since the problem is also relevant to inversion symmetry in 3D (for which $C_N^N = 1$ even for spinful electrons), we will consider a general setting of $C_N^N = (-1)^\eta$, $\eta = 0, 1$ for each of the three classes A, AI and AII.

Let us start by considering class A with $C_N^N=1$ ($\eta = 0$). The problem reduces to considering a 0D Hamiltonian with $U(1)$ charge conservation and an internal $\Z_n$ symmetry $C_N$. The eigenvalues of $C_N$ are the roots of unity $e^{-il \frac{2\pi}{N}}$ with the angular momentum $l=0,\dots,N-1$. For each $l$, we can define the number of particles in the ground state with angular momentum $l$ which we denote by $n_l$. In the absence of any extra constraints, such phase is described by $N$ integers $n_l$. The adjoining operation amounts to adding a Hamiltonian $\H = \H_0$ and its copies under rotation $\H_i = C_N^i \H (C_N^\dagger)^i$, $i = 1,\dots,N-1$. For each eigenstate $\psi_0$ of $\H_0$, we can construct an eigenstate $\psi_i = C_N^i \psi_0$ for $\H_i$, leading to a total of $N$ eigenstates with the same eigenvalues\footnote{Here, we imagine $\H_i$ to be slightly off the rotation center so that their eigenstates are always orthogonal}. We now define the states $\phi_l$ via the discrete Fourier transform as
\beq
\phi_l = \sum_{j=0}^{N-1} e^{i j l \frac{2\pi}{N}} \psi_j = \sum_{j=0}^{N-1} e^{i j l \frac{2\pi}{N}} C_N^j \psi_0.
\eeq
It is easy to see that $\phi_l$ is an eigenstate of $C_N$ corresponding to angular momentum $l$. This means that the adjoining operation adds the same set of eigenvalues to {\it all} angular momentum sectors $n_l \rightarrow n_l + n$. In other words, the configurations $\{n_l\}$ and $\{n_l + n\}$ are identified for any integer $n$. The resulting TCI classification can then by obtained by fixing $n_0$ leading to $\Z^N /\Z \simeq \Z^{N-1}$. The situation is identical for $C_N^N=-1$ ($\eta = 1$) where the angular momenta $l$ are half-integer. The same argument can be repeated leading also to $\Z^{N-1}$.

Spinless time-reversal symmetry satisfying $\T^2 = 1$ (class AI) flips the angular momentum $l \rightarrow -l$, imposing the constraint $n_l = n_{-l}$. For $\eta = 0$ with $N$ even, the constraint does not affect the two $\T$-invariant angular momenta $l=0$ and $l=N/2$ leading to $N/2+1$ independent integers which reduces to $N/2$ upon the identification $n_l \rightarrow n_l + n$ leading to $\Z^{N/2}$. For $N$ odd, there is only one $\T$-invariant angular momentum $l=0$ leading to $(N+1)/2$ integers which reduce to $(N-1)/2$ after the identification $n_l \rightarrow n_l + n$ leading to $\Z^{(N-1)/2}$. The case of $\eta = 1$ can be analyzed similarly. Here, there is only one $\T$-invariant angular momentum for odd $N$ ($l=N/2$) and no $\T$-invariant angular momenta for even $N$ leading to $\Z^{(N-1)/2}$ and $\Z^{N/2-1}$ respectively.

In addition to flipping the angular momentum, spinful time-reversal symmetry ($\T^2 = -1$) also imposes Kramers degeneracy at $\T$-invariant angular momenta. Furthermore, the adjoining operation involves the addition of a Kramers pair of states to each angular momentum $n_l \rightarrow n_l + 2n$. As a result, we obtain respectively the classifications $(2\Z)^2 \times \Z^{N/2-1}/2\Z \simeq 2\Z \times \Z^{N/2-1}$ or $2\Z \times \Z^{(N-1)/2} / 2\Z \simeq \Z^{(N-1)/2}$ for the case of $\eta=0$ for $N$ even or odd. For $\eta=1$, we get respectively $\Z^{N/2}/2\Z \simeq \Z^{N/2-1} \times \Z_2$ or $2\Z \times \Z^{(N-1)/2} / 2\Z \simeq \Z^{(N-1)/2}$ for even or odd $N$. The appearance of the $\Z_2$ factor for $\eta=1$ for $N$ even may be puzzling at the beginning, but we can understand it, for example, by considering the case of spinful two-fold rotation $N=2$, $\eta = 1$. In this case, the eigenvalues are $\pm i$ and they appear in pairs related by time-reversal (Kramers pairs). The adjoining operation can remove any even number of Kramers pairs by combining them leaving two distinct possibilities depending on the parity of the number of Kramers pairs. The results for class A, AI and AII for $C_N^N = \pm 1$ are summarized in Table~\ref{Classification}.

\begin{center}
\begin{table*}[t]
\caption{Classification of non-interacting 0D TIs protected by internal $\Z_N$ symmetry ($C_N^N=(-1)^\eta=\pm 1$) under the adjoining operation for classes A, AI and AII and its reduction in the presence of interaction. This captures the classification of 2D TCIs protected by rotation symmetry for spinful electrons ($\eta=1$) with (AII) or without (A) time-reversal symmetry and for spinless electrons ($\eta = 0$) with (AI) or without (A) time-reversal symmetry. We can also read from here the classification of 3D TCIs without surface states protected by inversion ($N=2$, $\eta = 0$) or roto-inversion for spinful electrons ($\eta=1$) with (AII) or without (A) time-reversal symmetry and for spinless electrons ($\eta = 0$) with (AI) or without (A) time-reversal symmetry}
\bgroup
\setlength{\tabcolsep}{1 em}
\setlength\extrarowheight{0.3em}
\begin{tabular}{c|c|c|c|c}
\hline \hline
$\eta$ & $N$ & A & AI & AII\\
\hline 
0 & even & $\Z^{N-1} \rightarrow \Z_{2N} \times \Z_{N/2}$ & $\Z^{N/2} \rightarrow \Z_{2N}$ & $2\Z \times \Z^{N/2-1} \rightarrow \Z_{N}$\\

0 & odd & $\Z^{N-1} \rightarrow \Z_{N} \times \Z_{N}$ & $\Z^{(N-1)/2} \rightarrow \Z_{N}$ & $\Z^{(N-1)/2} \rightarrow \Z_{N}$ \\

1 & even & $\Z^{N-1} \rightarrow \Z_{2N} \times \Z_{N/2}$ & $\Z^{N/2-1} \rightarrow \Z_{N/2}$ & $\Z^{N/2-1} \times \Z_2 \rightarrow \Z_{N}$ \\

1 & odd & $\Z^{N-1} \rightarrow \Z_{N} \times \Z_{N}$ & $\Z^{(N-1)/2} \rightarrow \Z_{2N}$ & $\Z^{(N-1)/2} \rightarrow \Z_{N}$ \\
\hline \hline
\end{tabular}
\egroup
\label{Classification}
\end{table*}
\end{center}

\subsection{Reduction of the classification in the presence of interactions}
\label{Interaction}
Let us now discuss the stability of the rotation protected TCIs considered in the previous section to the addition of  interactions. We restrict ourselves to local interactions and only consider the {\it reduction} of the non-interacting classification in the presence of interaction rather than any potential interaction-induced phases. In this case, all the phases considered are adiabatically connected to a non-interacting phase and are, therefore, short-range-entangled with a finite correlation length $\xi$. This means that the reduction to 0D can be performed as in the previous section. 

The interaction is required to be symmetric under $U(1)$ charge conservation and $C_N$. This means that the total particle number $N_t = \sum_l n_l$ and angular momentum $L_t = \sum_l l n_l \,\Mod N$ cannot be changed by the interaction. The filling for the individual angular momentum states $n_l$, however, can generally be altered. In fact, it is always possible to write an interaction term to transform any configuration $\{n_l\}$ to any other configuration $\{n_l'\}$ as long as the two configurations have the same total particle number $N_t = N_t'$ and total angular momentum $L_t = L_t'$. In order to obtain the interacting classification, we need in addition to take into account the action of the adjoining operation on $N_t$ and $L_t$. 

The adjoining operator transforms $n_l$ by adding a fixed integer $n$ to all of them $n_l \rightarrow n_l + n$. As a result, the total particle number and angular momentum change as
\beq
N_t \rightarrow N_t + n N, \qquad L_t \rightarrow L_t + n L_0.
\eeq
Here, $n$ is an arbitrary integer for classes A and AI and an even integer for class AII and $L_0$ is the angular momentum associated with the addition of a single state for each angular momentum. It is given by
\beq
L_0 = \sum_{l=0}^{N-1} l + \frac{\eta}{2} \!\!\! \mod N = \begin{cases} N/2 &: N + \eta \text{ even,} \\ 0 &: N + \eta \text{ odd.} \end{cases} 
\label{L0}
\eeq

In class A, where there is no extra constraint stemming from time-reversal symmetry, the interacting classification can be obtained as follows. For $\eta=0$ and even $N$, we can form the integer $N_t - 2L_t$ which is invariant under adjoining. Moreover, this integer will be invariant modulo $2N$ under interactions since $N_t$ is conserved and $L_t$ is conserved modulo $N$. In addition to this invariant, we have $L_t$ which is an integer defined only modulo $N/2$ due to adjoining. One can verify that these two invariants are sufficient to distinguish all distinct phases, and also all integer values of them are allowed, therefore the resulting classification is $\Z_{2N} \times \Z_{N/2}$ (notice that $N_t \mod N$ is completely fixed by these two integers). 
For $\eta=0$ and odd $N$, $L_t$ will be invariant under adjoining thus we can use $N_t \mod N$ and $L_t \mod N$ to specify the phase leading to $\Z_N \times \Z_N$. In summary, we have the following invariants for $\eta=0$: 
\beq
\label{ReductionA}
\S_{\rm A}^0 =
\begin{cases} 
(N_t - 2L_t \Mod 2N, L_t \Mod N/2) &: N\text{ even,} \\
(N_t \Mod N, L_t \Mod N) &: N\text{ odd.} \end{cases}
\eeq
We note that these two invariant pairs are also applicable to the more general cases $N+\eta=\text{even}$ and $N+\eta=\text{odd}$, respectively, which will be useful later when we are considering class AI. 
Next consider cases with $\eta=1$. We define a shifted total angular momentum $\tilde L_t\equiv L_t+N_t/2$ which is also conserved modulo $N$. The introduction of $\tilde L_t$ is in fact equivalent to a redefinition of the $C_N$ symmetry. One can easily check that for a given $N$, $\tilde L_t$ behaves in exactly the same way as $L_t$ with $\eta=0$: they both take integer values, and change by $N/2$ (or 0) under adjoining when $N$ is even (or odd), thus we get identical classifications as in the $\eta=0$ cases. The corresponding invariants $\S_{\rm A}^1$ take the same form as in \eqref{ReductionA} but with $L_t$ replaced by $\tilde L_t$. 

In class AI, there is the extra constraint that $n_l = n_{-l}$. This implies that the total angular momentum is a multiple of $L_0$ given in \eqref{L0}. For odd $N+\eta$, this means that the total angular momentum $L_t$ vanishes leading to the vanishing of the second component of the invariant pair $(N_t \Mod N, L_t \Mod N)$. In addition, $\eta=1$ for even $N$ implies that the total particle number is even, thus we get a $\Z_{N/2}$ classification for even $N$ and $\Z_N$ classification for odd $N$. For even $N+\eta$, the total angular momentum $L_t$ becomes a multiple of $N/2$ which implies that the second component of the invariant pair $(N_t - 2L_t \Mod 2N, L_t \Mod N/2)$ vanishes and we obtain a $\Z_{2N}$ classification. 
In summary, in class AI we have the following integer invariants: 
\begin{equation}
\S_{\rm AI}=\begin{cases}
N_t\Mod N &: \eta=0,~N\text{ odd}, \\
N_t/2\Mod N/2 &: \eta=1,~N\text{ even}, \\
N_t-2L_t\Mod 2N&: N+\eta\text{ even}. 
\end{cases}
\label{ReductionAI}
\end{equation}

For class AII, we have the constraint that $L_t$ is an even multiple of $L_0$ which immediately implies that it vanishes modulo $N$. In addition, $N_t$ is constrained to be even and changes by $2N$ under adjoining leading to a single $\Z_N$ invariant
\beq
\S_{\rm AII} = N_t/2 \Mod N.
\eeq
The results of the previous analysis for the three classes for $\eta=0,1$ are summarized in Table~\ref{Classification}.

We stress here that the 0D classification obtained in Table~\ref{Classification} can only be used for TCIs protected by rotation only (which means that translation symmetry is not assumed). To obtain the classification for a given wallpaper or space group, translation symmetry has to be included. This is done by considering the 0D classification from Table~\ref{Classification}) relative to all non-general Wyckoff positions. The resulting classification is then given by the product of the 0D classification for all the distinct Wyckoff positions.

One subtlety about the resulting classification is that all phases built by filling the general Wyckoff position (which consists of an orbital at a generic spatial point and its images under the symmetries) are considered trivial. This feature follows from the adjoining operation in the layer construction. It is also consistent with the discussion of Sec.~\ref{DTR} where we argued that a topological response is only possible whenever the tight-binding orbitals do not coincide with the Wannier orbitals. For a general Wyckoff position, the Wannier orbitals can always be moved freely to coincide with any given set of tight-binding orbitals which means that it is impossible to obtain a non-trivial quantized topological response in such atomic insulators. 

As a result of this identification, the classification obtained here differs from classification schemes which include the general Wyckoff position \cite{DominicPoWatanabe} but this difference is easily identifiable. For example, the general interacting classification for class A in one dimension with inversion $\I^2 = 1$ and translation is $\Z \times \Z_2 \times \Z_4$. Here, the $\Z$ factor indicates the total number of bands, $\Z_2$ indicates the total angular momentum, and the $\Z_4$ factor corresponds to the 0D invariant for one of the two $\I$-invariant points (the invariant for the second point is fixed by these three numbers). Our classification is obtained by modding out by the subgroup of phases generated by filling the general Wyckoff position which corresponds to a filling of 2 and angular momentum of 1\footnote{The general Wyckoff position is given by a pair of orbitals $|x\rangle$ and $\I |x\rangle = e^{i\phi} |-x\rangle$. The superpositions $|x\rangle + \I |x\rangle$ and $|x\rangle - \I |x\rangle$ correspond to inversion eigenvalues $\pm 1$ leading to a total angular momentum of 1}. The resulting identification reduces the factor of $\Z \times \Z_2$ to $\Z_4$ (since $\Z \times \Z_2 /\langle(2,1)\rangle \simeq \Z_4$) leading to the classification $\Z_4^2$ which matches the result obtained by taking the 0D invariant at the two $\I$-invariant points.

We can now use these results combined with the atomic insulator expressions in Table~\ref{AtomicRepresentations} to obtain the classification for the shift insulator model in the presence of interactions. For this purpose, it suffices to consider the model with spinless time-reversal symmetry (any extra invariants due to broken time-reversal symmetry will be fixed by the time-reversal symmetric invariants). The Wyckoff positions $a$, $b$, and $c$ are invariant under $C_6$, $C_3$, and $C_2$, respectively. Therefore, we can associate to the shift insulator a $\nu = (\nu_a,\nu_b,\nu_c) \in \Z_{12} \times \Z_3 \times \Z_4$ invariant given by 
\begin{multline}
(\nu_a,\nu_b,\nu_c) \\= (N_a - 2 L_a \Mod 12, N_b \Mod 3, N_c - 2 L_c \Mod 4),
\end{multline}
where $N_x$ and $L_x$ denote the total filling and total angular momentum for position $x$. These invariants are given in the last column of Table~\ref{AtomicRepresentations}. For example, the shift insulator with $L_z=1$ and positive $t$ and $\lambda$ has the atomic expression $c_1 - a_3$. $\nu_a$ is given by the filling of position $a$ which is $-1$ minus twice the angular momentum of this position which is $-2 \times 3$ leading to $\nu_a = -1 + 6 = 5$. $\nu_b$ is obviously 0 while $\nu_c$ is $1 - 2 \times 1 = -1$ leading to $\nu = (5,0,1)$. This means that we need to add 12 copies of the model to neutralize this index and we have a $\Z_{12}$ classification. Another example is the case $L_z=0$ for positive $t$. In this case, $\nu_a = 3 - 0 = 3$, $\nu_b = -2$, and $\nu_c = 1 - 0 = 1$. Although the $\nu_a$ index can be neutralized by adding only 4 copies of the model, this will not be enough to neutralize all three invariants since $4\nu_b = 1 \,\Mod 3$. As a result, the minimum number of copies needed to neutralize all invariants is also 12 leading to a $\Z_{12}$ interaction classification in all cases. This agrees with the interacting classification for a very closely related model where we replace the sixfold rotation symmetry in the 2D shift insulator model with an internal $\Z_6$. Unlike the shift insulator, this model exhibits gapless edge modes and the interacting classification can be obtained by analyzing the stability of edge to interactions following Refs.~\onlinecite{Levin12, Gu14}. The analysis is performed in Appendix \ref{InternalZnShift} leading to a $\Z_{12}$ classification as anticipated \cite{thorngren2018gauging}. 

\section{Conclusion}
\label{Discussion}
We conclude the paper by briefly reviewing the main results and discussing how they relate to other recent results. 

In Sec.~\ref{Model}, we introduced the shift insulator model which is built by stacking two $C_6$-symmetric Haldane models with Chern numbers $\pm 1$ for $p_\pm$ orbitals. The model exhibits several gapped phases separated by gap closing phase transitions but none of them possesses any stable gapless surface states. The absence of gapless surface states should not be surprising in light of the fact the 1D or 0D surface states cannot be stabilized in the presence of spinless time-reversal symmetry alone (class AI) \cite{Teo10}.

Afterwards, we turned to investigating the topological response of the model in detail. Since our model is built from two Chern insulators with opposite Chern numbers and angular momenta, we expect the threading of a magnetic flux to be associated with angular momentum pumping but no charge pumping. The resulting response can be described by a mutual Chern-Simons term $-(S/2\pi) A \wedge \rmd B$ coupling the electromagnetic gauge field $A_{\mu}$ to an effective gauge field $B_{\mu}$ corresponding to $C_6$ rotations. The angular momentum pumping in response to magnetic flux can then be understood as a charge response of the $B$ field to the flux of the $A$ field. Due to the symmetry of the Chern-Simons term, we also expect a charge response of the $A$ field to the flux of the $B$ field. The latter is implemented through a rotational defect (disclination) which is then expected to host a fractional quantized charge.

These predictions were verified, both numerically and analytically, in Sec.~\ref{Response}. There, the charge response to a disclination and the angular momentum response to a magnetic monopole were numerically calculated, before deriving analytic expressions from both the lattice theory and a continuum theory. Afterwards, we wrote down an effective Chern-Simons theory which unifies the two topological responses discussed above, verifying our expectations. The coefficient $S$ of the mutual Chern-Simons term for the shift insulator model is determined to be $-2\sgn(\lambda)L_z+3\sgn(t)$ up to integer multiples of $6$. This value differs from the one expected in any possible atomic insulator with two electrons per unit cell for which $S$ modulo $6$ can only be $0$ or $2$, which signals the absence of a symmetric Wannier representation for our model. 

The question of Wannier representability was then investigated thoroughly in Sec.~\ref{Wannier}. There, we verified our expectation that there is an obstruction to finding a basis of symmetric Wannier states by showing that the symmetry irreps of the model at high symmetry momenta do not match the irreps of any possible atomic insulator with the same filling (2 per unit cell) and the same symmetries. This obstruction was, however, proven to be fragile by showing that the model admits a Wannier representation once we add a set of carefully chosen atomic orbitals to it (Table~\ref{AtomicRepresentations} and Figs.~\ref{Wannier1} and \ref{Wannier2}). The required orbitals as well as the resulting Wannier states after the addition of those orbitals turned out to be different for the different phases of the model. 

The fragility of the Wannier obstruction in our model raises a general question: are all rotation protected TCIs in 2D either atomic or fragile phases? and more generally, are all TCIs which do not possess stable gapless states either atomic or fragile? In Sec.~\ref{Layer}, we showed that the answer to both questions is yes. The proof employs the layer construction to first reduce the TCI to a 0D atomic insulator before using this atomic insulator to build a Wannier representation. This procedure sometimes involves the addition of some atomic orbitals reflecting that these TCIs could be fragile. One important implication of this result is that the two seemingly inequivalent ways to define topology in non-interacting crystalline systems via the existence of gapless anomalous surface states or the existence of stable Wannier obstructions are indeed equivalent.

The layer construction was then used to provide a general classification of rotation-protected 2D TCIs and investigate how this non-interacting classification is reduced in the presence of interactions. The latter problem is greatly simplified by dimensional reduction procedure which maps it to a 0D problem that can be solved using elementary arguments. The resulting interaction-reduced classification is used to deduce that the non-interacting $\Z$ classification of the shift insulator is reduced to $\Z_{12}$ in the presence of interactions. This is consistent with what one would expect in a model where sixfold rotation is replaced by a $\Z_6$ internal symmetry \cite{thorngren2018gauging}.

Before closing, it is worth noting that the existence of a quantized response is not an exclusive feature of stable or fragile phases. In fact, frozen polarization phases such as the inversion-protected 1D SSH chain exhibit a quantized dipole moment and edge charge. It is possible, however, to deduce that a phase is not equivalent to an atomic insulator for certain values of the quantized response and number of filled bands as we have shown in Sec.~\ref{TopologicalAtomic}. This can be done in general by enumerating all possible atomic insulators (by specifying the filling of each symmetry irrep of the little group $g_x$ for each Wyckoff position $x$) and calculating the response to all possible defects for each of these insulators. If the response of a given model does not coincide with the response of any atomic insulator, we can deduce that there is an obstruction to finding a set of symmetric Wannier states. In addition, we can determine whether this obstruction is stable or fragile from the presence or absence of surface states as shown in Sec.~\ref{Layer}.

Finally, we emphasize that our discussion of surface states was restricted to stable gapless states localized at the surface. Such requirement, which was used as the defining feature of surface states and ``higher-order topology'' in several works \cite{Langbehn17, Khalaf18, Geier18, Trifunovic18}, does {\it not} include models with edge or corner {\it charge} such as the quadrupole model of Ref.~\onlinecite{Benalcazar17}. In these models, zero energy corner states exist if we impose an additional particle-hole or chiral symmetry. However, in the absence of such symmetries, these states can be pushed in energy to overlap with the conduction or valence band where they hybridize with the bulk states. In this case, there are no eigenstates of the Hamiltonian localized at the corner but there is still quantized fractional charge localized there \cite{Rhim17}. Such charges were discussed recently in Ref.~\onlinecite{BenalcazarHughes} and the discussion mirrors our analysis of the disclination charge. In particular, an $n$-fold symmetric system with open boundaries with a total filling of N that has $n$ corners will host a charge of $(N-N_0)/n \mod 1$ per corner, where $N_0$ is the total filling at the center considered in the mapping to a 0D problem discussed in Sec.~\ref{Rotation}. This is required by the total charge neutrality (assuming there is an integer number of unit cells) and provides the relationship between the corner charge and the disclination charge given in Table \ref{Classification}. 

\begin{acknowledgements}
We are especially grateful to Dominic~V.~Else, Hoi~Chun~Po and Haruki~Watanabe for sharing a related manuscript \cite{DominicPoWatanabe} and for a discussion regarding Table \ref{Classification}. We also thank Yin-Chen~He, Michael~Hermele, David~Vanderbilt, Chong~Wang and Yi-Zhuang~You for helpful discussions. A.V. was supported by a Simons Investigator award and by NSF-DMR 1411343. 
\end{acknowledgements}
\vspace{\baselineskip}
\noindent{\it Note added}.---Recently, we became aware of two related papers \cite{BenalcazarHughes, DominicPoWatanabe}. Ref.~\onlinecite{BenalcazarHughes} studies corner charges in 2D TCIs protected by rotation symmetry, and their discussion regarding disclination charge overlaps with ours, while Ref.~\onlinecite{DominicPoWatanabe} discusses the effect of interaction on fragile topological phases and its results overlap with the discussion of Sec.~\ref{Rotation} in this paper.

\bibliography{refs.bib}

\appendix
\section{Edge theory}
\label{Edge}
In this appendix, we explain the details of the derivation of the edge theory presented in Sec.~\ref{EdgeTheory}. Our starting point is the continuum Hamiltonian (\ref{HD}) where the edge is implemented by taking the mass parameter $m_0$ to change spatially $m_0 \rightarrow M(\br)$ such that $M(\br) = m_0$ deep inside the sample and $M(\br) = -m_0$ outside it. We decompose the momentum as $\bk = k_t \bt + k_n \bn$ with $\bt$ denoting the unit vector along the tangent to the edge $\bt = (-\sin \varphi, \cos \varphi, 0)$ and denote by $\lambda$ the spatial direction normal to the edge in the plane. Making the substitution $k_n \rightarrow -i \partial_\lambda$, we get
\begin{align}
&\H = - i v_F \bn \cdot \tilde \bsigma \partial_\lambda + v_F k_t \bt \cdot \tilde \bsigma - M(\lambda) \sigma_z \gamma_z \tau_z, \nonumber \\ 
&\tilde \bsigma = (\sigma_x \gamma_z, \sigma_y, \sigma_z \gamma_z). 
\label{Hdl}
\end{align}
We now seek solutions of Eq.~\eqref{Hdl} which are exponentially localized to the edge region. This is achieved by the ansatz
\beq
\Psi(k_t,\lambda) = e^{-\frac{1}{v_F} \int_0^\lambda d\lambda' M(\lambda')} \psi(k_t),
\eeq
which gives
\begin{multline}
\H \Psi(k_t,\lambda) =\\  [-M(\lambda)\sigma_z \gamma_z \tau_z (1 - i \sigma_z \gamma_z \tau_z\bn \cdot \tilde \bsigma) + v_F k_t \bt \cdot \tilde \bsigma] \Psi(k_t,\lambda).
\end{multline}
The first term in the square brackets can be eliminated by choosing $\psi(k_t)$ to satisfy $\psi(k_t) = \tilde P \psi(k_t)$ with the projection operator $\tilde P$ defined as
\beq
\tilde P = \frac{1}{2}(1 + i \sigma_z \gamma_z \tau_z\bn \cdot \tilde \bsigma) = \frac{1}{2} (1 + \cos \varphi \, \sigma_y \tau_z - \sin \varphi \, \sigma_x \gamma_z \tau_z).
\eeq
The projection operator can be diagonalized by introducing the unitary rotation matrix
\beq
U = e^{i \frac{\pi}{4} (\cos \varphi \sigma_y - \sin \varphi \sigma_x \gamma_z) \bn \cdot \btau} e^{i \frac{\pi}{4} \bn \cdot \btau} ,
\eeq
leading to
\beq
P = U^\dagger \tilde P U = \frac{1}{2} (1 - \tau_z).
\eeq
Applying the rotation $U$ followed by the projection $P$, we obtain the edge Hamiltonian
\begin{multline}
P U^\dagger \H U P \rightarrow \H_{\rm edge} = \\ v_F (k_x \sin \varphi - k_y \cos \varphi) (\sigma_x \gamma_z \sin \varphi - \sigma_y \cos \varphi) .
\end{multline}
Here, the arrow indicates picking out the non-zero block of the Hamiltonian. The edge Hamiltonian can be simplified when written in terms of the tangent vector $\bt$ leading to Eq.~\ref{Hedge}.

\section{Continuum Theory of Haldane Model Disclinations}\label{HaldaneContinuumTheory}
\subsection{Infinite mass boundary condition}
In the continuum theory approach to Haldane model disclinations, we will be using an \emph{infinite mass boundary condition} \cite{berryinfmass} which we now briefly review. 

Consider two general Dirac theories
\begin{align}
H_m&=v\boldsymbol{\alpha}\cdot\vec{p}+m\beta, \\
H_M&=v\boldsymbol{\alpha}\cdot\vec{p}+M\beta'
\end{align}
sitting at two sides of a boundary, where $\{\alpha_i,\beta\}$ and $\{\alpha_i,\beta'\}$ are two sets of anticommuting gamma matrices satisfying $\alpha_i^2=\beta^2={\beta'}^2=1$. We assume $v$, $m$ and $M$ to be positive, which is always possible by a proper definition of $\boldsymbol{\alpha}$, $\beta$ and $\beta'$. 
We can regard the $H_M$ side as the ``outside'' and send $M\rightarrow+\infty$. This imposes a boundary condition on wave functions at the $H_m$ side, which we now derive. 

For a given boundary point, we define a normal vector $\bn$ pointing towards the $H_M$ side and a perpendicular tangent vector $\bt$ (pointing towards either of the two directions). We also define a Cartesian coordinate system $(\lambda,\mu)$ such that $\lambda$ and $\mu$ are along $\bn$ and $\bt$, respectively, and that the boundary point has coordinates $(0,0)$. The eigenvalue equation can now be written as 
\begin{align}
&[v(\alpha_t p_t+\alpha_n p_n)+m\beta]\psi=E\psi, \\
&[v(\alpha_t p_t-\alpha_n p_{-n})+M\beta']\psi=E\psi, \label{Mequation}
\end{align}
and the second equation implies 
\begin{align}
\left( -\rmi\frac{v}{M}\alpha_n\alpha_t p_t+\frac{v}{M}\partial_{-\lambda}-\rmi\alpha_n\beta' \right)\psi=(-\rmi\alpha_n)\frac{E}{M}\psi. 
\end{align}
We assume that a limit $(\tilde\psi,\tilde E)=\lim_{M\rightarrow+\infty}(\psi,E)$ exists and that $\tilde{\psi},\tilde E,\lim_{M\rightarrow+\infty} p_t\psi$ are all finite. Also assume for a moment that $\tilde{\psi}(0,0)\neq 0$. Then in the above equation, only the second and third terms at the left-hand side can have significant contribution, or more precisely, 
\begin{equation}
\lim_{M\rightarrow+\infty}\left(\frac{v}{M}\partial_{-\lambda}-\rmi\alpha_n\beta'\right)\psi=0. 
\end{equation}
This implies that $\tilde\psi$ cannot have any component of eigenvalue $-1$ when decomposed into eigenvectors of $\rmi\alpha_n\beta'$, otherwise $\tilde\psi$ will blow up at $\lambda>0$. In other words, we must have
\begin{equation}
\rmi\alpha_n\beta'\tilde{\psi}=\tilde{\psi}
\label{InfiniteMassBdryCondition}
\end{equation}
at boundary points where $\tilde\psi\neq 0$. Since this linear equation is trivially satisfied when $\tilde\psi=0$, it actually applies to all boundary points, and this is exactly the infinite mass boundary condition that we are looking for. We emphasize that this boundary condition is \emph{not} the most general boundary condition \cite{akhmerov2008boundary,roy&stone}. We will drop the tilde on $\psi$ hereafter. 

As a consistency check, we note from Ref. \onlinecite{akhmerov2008boundary} that the infinite mass boundary condition \eqref{InfiniteMassBdryCondition} together with the fact $\{\rmi\alpha_n\beta',\alpha_n\}=0$ imply $\psi^\dagger\alpha_n\psi=0$ for any $\psi$, which means there is no outgoing current at the boundary and guarantees the Hermiticity of the Hamiltonian. 

Following Ref.~\onlinecite{haldanemodelcone}, we postulate that Haldane model disclinations can be described by an infinite mass boundary condition at the disclination hole. Requiring zero Chern number and the $C_6$ rotation symmetry pins down the mass term gamma matrix in the hole to $\sigma_x\gamma_x$ up to a sign. This term can in fact be realized at lattice level as the Fries-Kekul\'{e} structure \cite{roy&stone}. Now we need to decide the sign of this mass term. We would like to demand the relative sign between the mass term and the Fermi velocity to be fixed. This guarantees that, if we reverse the sign of the Haldane model Hamiltonian, the whole field theory spectrum is also reversed. We then write
\begin{equation}
H_\text{hole}=v_F(p_x\sigma_x\gamma_z+p_y\sigma_y)-\sgn(v_F)M\sigma_x\gamma_x. 
\end{equation}
The sign of $M$ is still unfixed, and we have to compare the field theory solution with the actual spectrum to remove this ambiguity. In fact from our analysis in Sec.~\ref{LargerHole}, both signs of $M$ are allowed depending on the atomic detail of the disclination hole. In the minimal hole case, $M\rightarrow+\infty$ turns out to be the good choice. We therefore obtain the boundary condition
\begin{equation}
(\gamma_y,-\sigma_z\gamma_x)\cdot\bn\psi=\psi.  
\end{equation}

\subsection{The continuum theory}
In the branch cut gauge of a disclination, the Hamiltonian takes its normal form 
\begin{equation}
H=-\rmi v_F(\sigma_x\gamma_z\partial_x+\sigma_y\partial_y)-m\sigma_z\gamma_z, 
\end{equation}
but the (envelope) wave function $\psi(\vec{r})$ satisfies a nontrivial angular boundary condition: using the representation of $\hat C_6$, and taking into account the phase jump due to magnetic flux, one can show that 
\begin{equation}
	\psi(\phi=2\pi)=\left(\sigma_x\gamma_x\rme^{\rmi(2\pi/3)\sigma_z\gamma_z}\right)^{n_\Omega}\rme^{\rmi\Delta}\psi(\phi=0),  
\end{equation}
where we have defined a rescaled polar coordinate $\phi\equiv\varphi/(1-n_\Omega/6)$ and $\Delta=n_\Omega\pi L_z/3+2\pi\Phi$. Following the strategy of Ref.~\onlinecite{graphenecone,haldanemodelcone}, we now do a unitary transformation to reduce the problem to a solvable form. Define $\psi_4=U_4U_3U_2U_1U_0\psi\equiv\mathcal{U}\psi$, where 
\begin{align}
	U_0&=\left( \frac{1+\gamma_z}{2}+\frac{1-\gamma_z}{2}\sigma_x \right),\\
	U_1&=\rme^{\rmi\pi\sigma_z(1-\gamma_z)/4},\\
	U_2&=\rme^{\rmi\theta\sigma_z/2},~\text{with $\theta=\pi/2+\phi(1-n_\Omega/6)$,}\\
	U_3&=\exp\left( -\rmi\phi\left( \frac{\Delta}{2\pi}+\frac{n_\Omega}{4}\gamma_y \right) \right),\\
	U_4&=\frac{1}{\sqrt{2}}(1-\rmi\gamma_x). 
\end{align}
The effect of $U_0$ and $U_1$ is to simplify the gamma matrices, $U_2$ performs a local frame rotation such that the Hamiltonian takes a simple form in polar coordinates, $U_3$ replaces a complicated angular boundary condition with a gauge field, and $U_4$ finally block diagonalizes the Hamiltonian. 
The transformed Hamiltonian is
\begin{align}
	H_4&=(-\rmi v_F)\left[ \frac{1}{(1-n_\Omega/6)r}\sigma_x\left( \partial_\phi+\frac{1}{4}\rmi n_\Omega\gamma_z+\rmi\frac{\Delta}{2\pi} \right)\right.\nonumber\\
	&\left.-\sigma_y\left( \partial_r+\frac{1}{2r} \right) \right]-m\sigma_z, 
	\label{ContinuumTheoryH4}
\end{align}
and the boundary conditions for $\psi_4$ are
\begin{align}
	\psi_4(\phi=2\pi)&=-\psi_4(\phi=0),\\
	\gamma_z\sigma_x\psi_4(r=\rho)&=\psi_4(r=\rho),  
\end{align}
where $\rho$ is the disclination hole radius (consider a round hole centered at $r=0$). The eigenvalue problem of $H_4$ is now straightforward to solve by a partial wave expansion $\psi_4=\sum_j\chi^{(j)}(r)\rme^{\rmi j\phi}$ with $j\in \mathbb{Z}+1/2$. 

Note that $\gamma_z$ commutes with the Hamiltonian $H_4$ in Eq.~\ref{ContinuumTheoryH4}, the theory therefore splits into two sectors $\gamma_z=\pm\gamma$ where $\gamma\equiv\sgn(v_F)\sgn(m)$. It turns out that bound states ($|E|<|m|$) satisfying the boundary condition at $r=\rho$ are only possible in the $\gamma_z=-\gamma$ sector. The radial wave function $\chi$ for such a bound state is 
\begin{equation}
	\chi_{\nu,E}=
	\begin{pmatrix}
	K_{\nu-1/2}(\kappa r)\\
	-\frac{m+E}{\kappa v_F}K_{\nu+1/2}(\kappa r)
	\end{pmatrix}, 
\end{equation}
subject to the constraint
\begin{equation}
	\frac{K_{\nu-1/2}(\kappa\rho)}{K_{\nu+1/2}(\kappa\rho)}=\sqrt\frac{m+E}{m-E}, 
	\label{BoundStateEneryEq}
\end{equation}
where $\kappa=\frac{1}{|v_F|}\sqrt{m^2-E^2}$ and 
\begin{equation}
\nu=\frac{1}{1-n_\Omega/6}\left( j-\frac{1}{4}n_\Omega\gamma+\frac{\Delta}{2\pi} \right).  
\end{equation}
We then see that, at $L_z=0$, a zero energy bound state exists when $\Phi=n_\Omega\gamma/4+\text{half integer}$, which is exactly what we found in Sec.~\ref{ZeroEnergyStatesSection}. Moreover, one can study the rotation property of the zero energy bound state at $n_\Omega=0,\Phi=1/2$, and the result is also consistent with the $\pi/3$ periodicity we found previously. Recall that this enables us to compute the ground state spin change due to torus monopole fluxes. 

Let us take a look at the effective theories of the two sectors $\gamma_z=\pm\gamma$. The Hamiltonians and boundary conditions at $r=\rho$ are 
\begin{align}
H_{4\pm}&=(-\rmi v_F)\left[ \frac{1}{(1-n_\Omega/6)r}\sigma_x\left( \partial_\phi+\rmi\Phi_\pm\right)\right.\nonumber\\
&\left.-\sigma_y\left( \partial_r+\frac{1}{2r} \right) \right]-m\sigma_z,\\
\pm&\gamma\sigma_x\psi_{4\pm}(r=\rho)=\psi_{4\pm}(r=\rho), 
\end{align}
where $\Phi_\pm$ are now just numbers and can be read off from Eq.~\ref{ContinuumTheoryH4}. The $1/(1-n_\Omega/6)$ factor can be dropped out since it just renormalizes the angular Fermi velocity and does not affect any topological properties. Thus if we consider a pair of the original disclination systems, the two sectors correspond to two Haldane models \emph{without} disclination. Interestingly, boundary conditions for these two sectors at the disclination hole represent infinite mass Haldane models with the \emph{same} and \emph{opposite} signs of mass, respectively, which is important for correctly reproduce the disclination charge as we did in Sec.~\ref{ResponsesContinuumApproach}. 

\section{Edge stability of internal $\mathbb{Z}_n$ shift insulator}\label{InternalZnShift}
\subsection{The model}
In this section, we study the topological classification of the internal $\mathbb{Z}_n$ shift insulator which is similar to the shift insulator model considered in the main text but with the rotation symmetry replaced by an internal $\mathbb{Z}_n$ symmetry. This symmetry is generated by $Z_n$ with the action
\begin{equation}
	Z_nc_{\br,+}Z_n^{-1}=\rme^{\rmi\frac{2\pi}{n}}c_{\br,+},\quad Z_nc_{\br,-}Z_n^{-1}=c_{\br,-}. 
	\label{ZnTransformation}
\end{equation}
We also require the charge $U(1)$ symmetry generated by $Q$ with $[Q,c_{\br,\pm}]=c_{\br,\pm}$. 

In the noninteracting theory, the internal $\mathbb{Z}_n~(n\geq2)$ shift insulator admits a $\mathbb{Z}$ classification: no matter how many copies we have, backscattering terms of the form $c^\dagger_+c_-$ are never allowed and therefore the edge is always gapless. 

\subsection{Edge stability with interaction}\label{InternalShiftClassification}
We now study the effect of interaction on the classification from the edge stability approach. 
Consider $M$ copies of the internal $\mathbb{Z}_n$ shift insulator. The bulk is characterized by the following effective Chern-Simons theory (see, for example, Ref.~\onlinecite{wen2004book}): 
\begin{equation}
\mathcal{L}_\text{bulk}~d^3x=\frac{K_{IJ}}{4\pi}a_I\wedge \rmd a_J-\frac{1}{2\pi}\tau_I A\wedge \rmd a_I, 
\end{equation}
where $A$ is the external electromagnetic
potential. The $K$ matrix is a $2M\times 2M$ symmetric nondegenerate integer matrix of the form $K=\mathrm{diag}(1,1,\cdots,1,-1,\cdots,-1)$, and we take the charge vector $\tau$ to be $\tau=(1,1,\cdots,1)^T$. Quasiparticle excitations in this system are described by integer gauge charge vectors $l$, and the physical electric charge of each excitation is given by 
\begin{equation}
q_l=l^TK^{-1}\tau
\end{equation}
in the unit of $e$. In particular, local excitations are of the form $l=K\Lambda$ for some integer vector $\Lambda$. We can only use these local degrees of freedom to construct local interaction terms. 

The corresponding edge theory is a Luttinger Liquid (chiral compact boson) theory characterized by 
\begin{align}
\mathcal{L}_\text{edge}&=\frac{1}{4\pi}(K_{IJ}\partial_x\Phi_I\partial_t\Phi_J-V_{IJ}\partial_x\Phi_I\partial_x\Phi_J)\nonumber\\
&+\frac{1}{2\pi}\epsilon^{\mu\nu}\tau_I \partial_\mu\Phi_I A_\nu,
\end{align}
and the chirality condition \cite{gross1985heterotic} $(K\partial_t-V\partial_x)\Phi=0$, where $V_{IJ}$ is a positive-definite velocity matrix. Note that in $1+1$ dimensions, $\epsilon^{\mu\nu}$ is numerically equal to $-\epsilon_{\mu\nu}$, i.e. $\epsilon^{01}=-1$. 
Quasiparticle creation operators are of the form $\rme^{-\rmi l^T\Phi}$. Given our choice of the charge vector $\tau$, creation operators for $p_\pm$ electrons are given by $\rme^{\rmi\Phi_k}$ and $\rme^{-\rmi\Phi_{k+M}}$, respectively, with $1\leq k\leq M$. 

To study the interaction effect at the edge, we follow the method in Ref.~\onlinecite{Levin12,Gu14}. Consider interaction terms of the following form. 
\begin{equation}
U(\Lambda)=U(x)\cos(\Lambda^TK\Phi-\alpha(x)). 
\end{equation}
In order to gap out the edge without breaking the $U(1)\times\mathbb{Z}_n$ symmetry, we need to add $M$ such terms, i.e. $\sum_{i=1}^MU(\Lambda_i)$, with linearly independent integer vectors $\Lambda_i$ satisfying the following conditions: 
\begin{enumerate}
	\item Explicitly preserve the symmetry. 
	
	For the $U(1)$ symmetry $\Phi\mapsto \Phi+\varphi K^{-1}\tau$, we need 
	\begin{equation}
	\Lambda_i^T\tau=0. 
	\end{equation}
	For the $\mathbb{Z}_n$ symmetry $\Phi\mapsto\Phi+(2\pi/n)K^{-1}\chi$ where $\chi^T=(1,1,\cdots,1,0,\cdots,0)$, we need
	\begin{equation}
	\Lambda_i^T\chi=0\mod n. 
	\end{equation} 
	\item Do not spontaneously break the symmetry. 
	
	Note that the interaction terms will freeze the values of $\Lambda_i^TK\Phi$. If, for some coprime integers $a_1,\cdots,a_M$, the linear combination $\sum_i a_i\Lambda_i$ is nonprimitive, i.e. $\sum_i a_i\Lambda_i=k\Lambda$ for some integer vector $\Lambda$ and an integer $k>1$, then the value of $\Lambda^TK\Phi$ is also frozen and this may spontaneously break the $\mathbb{Z}_n$ symmetry. It is proven in Ref.~\onlinecite{Levin12} that, a nonprimitive linear combination $\sum_i a_i\Lambda_i$ exists for some coprime integers $a_1,\cdots,a_M$, if and only if the set of $
	\begin{pmatrix}
	2M\\
	M
	\end{pmatrix}
	$ $M\times M$ minors of the matrix $(\Lambda_1,\cdots,\Lambda_M)$ have a nontrivial common divisor. We will require ruling out this possibility. Also note that having at least one nonzero $M\times M$ minor implies linear independence, so we do not need to check that separately. 
	\item Fully gap out the edge. 
	
	This is guaranteed by the following Haldane null vector condition \cite{haldane1995stability}: 
	\begin{equation}
	\Lambda^T_iK\Lambda_j=0~~~\forall i,j. 
	\end{equation}
\end{enumerate}
Our task now reduces to finding the set of $\Lambda_i$ vectors obeying the above constraints. 

\begin{theorem}\label{oddcase}
	When $n$ is odd, the edge of $M=n$ copies of $\mathbb{Z}_n$ shift insulator can be gapped out. In particular, if $n$ is an odd prime number, this implies $\mathbb{Z}_n$ classification (assuming that a single copy is nontrivial). 
\end{theorem}
\begin{proof}
	We claim that the following matrix of $\Lambda^T_i$'s will do the job. 
	\begin{equation}
	\begin{pmatrix}
	\Lambda^T_1\\
	\vdots\\
	\Lambda^T_n
	\end{pmatrix}=
	\begin{pmatrix}[ccccc|ccccc]
	1 & -1 & & & &  1 & -1 & & & \\
	1 & & -1 & & &  1 & & -1 & & \\
	\vdots & & & \ddots & &   \vdots & & & \ddots & \\
	1 & & & & -1 &   1 & & & & -1\\
	1 & 1 & 1 & \cdots & 1 &   -1 & -1 & -1 & \cdots & -1 
	\end{pmatrix}. 
	\end{equation}
	Explicit symmetry and the null vector condition can be directly checked. To check the primitivity condition, note that the determinant for the first through $n$-th columns is $n$, and the determinant for the second through $(n+1)$-th columns is $(-1)^{n-1}(n-2)$. Since $n$ and $(n-2)$ are already coprime for an odd $n$, we conclude that the set of $n\times n$ minors of the matrix above do not have a nontrivial common divisor. 
	
	The above result does not in general imply $\mathbb{Z}_n$ classification, because it is possible that fewer number of copies of the system is already trivial. However, when $n$ is an odd prime number, since $\mathbb{Z}_n$ does not have any nontrivial proper subgroup, we can safely conclude $\mathbb{Z}_n$ classification as long as a single copy is nontrivial. At least for the interaction term considered here, there is no symmetry allowed choice which can gap out the edge states of a single $\mathbb{Z}_n$ shift insulator. 
\end{proof}

\begin{theorem}\label{evencase}
	When $n$ is even, the edge of $M=2n$ copies of $\mathbb{Z}_{n}$ shift insulator can be gapped out. In particular, if $n=2p$ with $p$ being an odd prime number, this implies $\mathbb{Z}_{4p}$ classification.  
\end{theorem}
\begin{proof}
	Consider the following matrix. 
	\begin{equation}
	\begin{pmatrix}
	\Lambda^T_1\\
	\vdots\\
	\Lambda^T_{2n}
	\end{pmatrix}=
	\begin{pmatrix}[ccccc|ccccc]
	1 & -1 & & & &  1 & -1 & & & \\
	1 & & -1 & & &  1 & & -1 & & \\
	\vdots & & & \ddots & &   \vdots & & & \ddots & \\
	1 & & & & -1 &   1 & & & & -1\\
	1 & 0 & 1 & \cdots & 0 &   0 & -1 & 0 & \cdots & -1 
	\end{pmatrix}. 
	\end{equation}
	Explicit symmetry and the null vector condition can be directly checked. To check the primitivity condition, note that the determinant for the first through $2n$-th columns is $n$, and the determinant for the second through $(2n+1)$-th columns is $(-1)^{2n-1}(n-1)$. Since $n$ and $(n-1)$ are already coprime, we conclude that the set of $2n\times 2n$ minors of the matrix above do not have a nontrivial common divisor. 
	
	If $n=2p$ with $p$ being an odd prime number, we can say more about the classification. Suppose the actual classification is $\mathbb{Z}_m$, then $m$ must be a divisor of $4p$. Also, there should exist a $\mathbb{Z}_{2p}$ symmetric interaction term which can gap out the edge of $m$ copies of the system. Note that a $\mathbb{Z}_{2p}$ symmetric term is also $\mathbb{Z}_2$ symmetric, and the result of Ref.~\onlinecite{Gu14} implies that the classification in the $\mathbb{Z}_2$ case is $\mathbb{Z}_4$, so $m$ must be a multiple of $4$. The only allowed choice of $m$ is then $4p$. 
\end{proof}
\subsection{Relation to $C_6$ shift insulators}
Let us now try to match $C_6$ rotation and internal $\mathbb{Z}_6$ shift insulators. According to Theorem \ref{evencase} in Sec.~\ref{InternalShiftClassification}, the $\mathbb{Z}_6$ shift insulator admits $\mathbb{Z}_{12}$ classification, but this does not imply the same classification for all $C_6$ shift insulators because a $C_6$ shift insulator is not necessarily mapped to a single copy ($M=1$) of the internal $\mathbb{Z}_6$ shift insulator. To establish the correct mapping, we can compare the charge response to symmetry defects. In the $C_6$ shift insulator case, the symmetry defects are nothing but disclinations and the number of trapped electrons is given by $({n_\Omega}/{6})\sum_l c_{a,l}$ using the atomic insulator representation discussed in Sec.~\ref{WannierRep}. In the case of internal $\mathbb{Z}_6$ shift insulator, a symmetry defect is a local object such that the braiding of electrons around it reproduces the symmetry transformation in Eq.~\ref{ZnTransformation}. The defect can then be identified as the \emph{fractional} gauge charge vector $l_d=\frac{1}{6}(1,\cdots,1,0,\cdots,0)^T$ or physically a $2\pi/6$ flux in the ``$+$'' sector, and its braiding statistics with an electron can be verified using the formula $\theta_{ll'}=2\pi l^T K^{-1}l'$. 
From quantum Hall effect, we know that the number of electrons trapped by this single defect is $-\frac{1}{6}M$. Therefore, if we identify the ``+'' sector $2\pi/6$ flux with the disclination $n_\Omega=1$ (or $-1$) in the rotation symmetry case, we conclude that the correspondence between $C_6$ and internal shift insulators must satisfy $\sum_l c_{a,l}=\mp M\mod6$. This relation is already enough for showing that the topological classifications of $C_6$ and internal shift insulators match with each other. For $C_6$ shift insulators with only the rotation symmetry (class A, and no translation symmetry), the classification is determined by $N_a-2 L_a\mod 12$ and $L_a\mod3$, where $N_a=\sum_lc_{a,l}$ and $L_a=\sum_l lc_{a,l}$ are respectively the total number and total angular momentum of $a$-type atomic orbitals at one plaquette center. Since the $C_6$ shift insulator model has spinless time-reversal symmetry (although for now we do not consider it as a protecting symmetry), we always have $L_a=0\mod3$ and the only nontrivial invariant left is $N_a-2L_a\mod 12$. From the disclination charge formula $N_a=6\Delta Q|_{n_\Omega=1}=-2\sgn(\lambda)L_z+3\sgn(t)\mod 6$, we know that $N_a$ is always an odd number, and therefore we have
\begin{align}
	&\gcd(12,N_a-2L_a)=\gcd(12,\mathrm{mod}(N_a,6))\nonumber\\
	&=\gcd(12,\mathrm{mod}(M,6))=\gcd(12,M), 
\end{align}
where $\gcd$ abbreviates greatest common divisor. This proves that the classification of $C_6$ and internal $\mathbb{Z}_6$ shift insulators are indeed the same. 
\end{document}